\documentclass{amsart}

\usepackage{amsfonts, amssymb, amsmath, amsthm} % AMS symbols
\usepackage{mathrsfs} % Math symbols
\usepackage{setspace} % Spacing
\usepackage{geometry} % Spacing
\usepackage{xcolor} % Text coloring
\usepackage{graphicx} % Graphics
\usepackage{slashed} % Slashed
\usepackage{comment} % Comments
\usepackage{tikz} % Diagrams

\usepackage[backend=biber, style=alphabetic, sorting=nyt, giveninits]{biblatex}
\AtEveryBibitem{\clearfield{issn}}
\AtEveryBibitem{\clearfield{url}}
\AtEveryBibitem{\clearfield{doi}}
\AtEveryBibitem{\clearfield{month}}
\renewbibmacro{in:}{}
\usepackage{tensor}

\newcommand{\W}{\tensor{W}}
\newcommand{\R}{\tensor{R}}

\newcommand{\al}{\alpha}
\newcommand{\be}{\beta}
\newcommand{\ga}{\gamma}
\newcommand{\de}{\delta}
\newcommand{\la}{\lambda}
   
\newcommand{\li}{\mi{L}_\rho}

\usepackage{enumerate} % http://ctan.org/pkg/enumerate

\newtheorem{theorem}{Theorem} % Theorem environment
\newtheorem{corollary}[theorem]{Corollary} % Corollary environment
\newtheorem{lemma}[theorem]{Lemma} % Lemma environment
\newtheorem{proposition}[theorem]{Proposition} % Proposition environment
\newtheorem{definition}[theorem]{Definition} % Definition environment
\newtheorem{question}[theorem]{Question} % Question environment
\newtheorem{remark}[theorem]{Remark} % Remark environment
\numberwithin{equation}{section} % Equation numbering
\numberwithin{theorem}{section} % Theorem numbering

\newcommand{\mf}[1]{\mathfrak{#1}} % Mathfrak
\newcommand{\mi}[1]{\mathscr{#1}} % Mathscr
\newcommand{\mc}[1]{\mathcal{#1}} % Mathcal
\newcommand{\ms}[1]{\mathsf{#1}} % Math Sans serif
\newcommand{\N}{\mathbb{N}} % Natural numbers
\newcommand{\paren}[1]{\left(#1\right)} % Resized parentheses
\newcommand{\brak}[1]{\left[#1\right]} % Resized brackets
\newcommand{\abs}[1]{\left|#1\right|} % Resized absolute value
\newcommand\norm[1]{\left\lVert#1\right\rVert} % Resized norm

\newcommand{\gv}{\mathsf{g}} % Vertical metric
\newcommand{\hv}{\mathsf{h}} % Vertical Riemannian metric
\newcommand{\Rv}{\mathsf{R}} % Vertical curvature
\newcommand{\Rcv}{\mathsf{Rc}} % Vertical Ricci curvature
\newcommand{\Rsv}{\mathsf{Rs}} % Vertical scalar curvature
\newcommand{\Dv}{\mathsf{D}} % Vertical Levi-Civita connection
\newcommand{\wv}{\mathsf{W}} % Vertical Weyl curvature
\newcommand{\Gammav}{\bar{\Gamma}} % Vertical Christoffel symbols
\newcommand{\kv}{\mathsf{k}} % Vertical extrinsic curvature

\newcommand{\gm}{\mf{g}} % Boundary metric, zero-order (boundary)
\newcommand{\sphere}{\mathring{\gamma}} % Metric of the (n-1)-sphere
\newcommand{\Dm}{\mf{D}} % Boundary Levi-Civita connection
\newcommand{\Rm}{\mf{R}} % Boundary curvature, zero-order (boundary)
\newcommand{\Rcm}{\mf{Rc}} % Boundary Ricci curvature
\newcommand{\Rsm}{\mf{Rs}} % Boundary scalar curvature

\newcommand{\nablam}{\bar{\nabla}} % Mixed derivative
\newcommand{\Dvm}{\bar{\Dv}} % Mixed vertical derivative
\newcommand{\Boxm}{\bar{\Box}} % Mixed wave operator
\newcommand{\ix}[1]{\bar{#1}} % Multi-index
\newcommand{\ixd}[2]{\hat{#1}_{#2}} % Multi-index, with deletion
\newcommand{\ixr}[3]{\hat{#1}_{#2}\![#3]} % Multi-index, with replacement

\let\oldtocsection=\tocsection
\let\oldtocsubsection=\tocsubsection
\let\oldtocsubsubsection=\tocsubsubsection
\renewcommand{\tocsection}[2]{\hspace{0em}\oldtocsection{#1}{#2}}
\renewcommand{\tocsubsection}[2]{\hspace{1em}\oldtocsubsection{#1}{#2}}
\renewcommand{\tocsubsubsection}[2]{\hspace{2em}\oldtocsubsubsection{#1}{#2}}

\begin{document}

\title{Holographic Characterisation of Locally Anti-de Sitter Spacetimes}
\author{Alex McGill}

\address{School of Mathematical Sciences\\
Queen Mary University of London\\
London E1 4NS\\ United Kingdom}
\email{a.mcgill@qmul.ac.uk}

\begin{abstract}
It is shown that an $(n+1)$--dimensional asymptotically anti-de Sitter solution of the Einstein-vacuum equations is locally isometric to pure anti-de Sitter spacetime near the conformal boundary if and only if the boundary metric is conformally flat and (for $n \neq 4$) the boundary stress-energy tensor vanishes, subject to (i) sufficient (finite) regularity in the metric and (ii) the satisfaction of a geometric criterion on the boundary. A key tool in the proof is the Carleman estimate of \cite{McGill_2020}, which is applied to prove a unique continuation result for the Weyl curvature at the conformal boundary given vanishing to sufficiently high order over a sufficiently long timespan.
\end{abstract}

\maketitle

\tableofcontents

%\newpage

\section{Introduction} \label{sec.intro}

\subsection{Overview}

Asymptotically anti-de Sitter (aAdS) spacetimes play a central role in the conjectured \textit{AdS/CFT correspondence} \cite{Maldacena_1999}, which posits a duality between the bulk gravitational theory in such spacetimes and \textit{conformal field theories} (CFT) living on their lower--dimensional `boundaries'. In this sense, it is a realisation of the more general \textit{holographic principle}.

In order to rigorously determine the validity of this conjecture, one should aim to prove a $1-1$ correspondence between $(n+1)$--dimensional aAdS solutions of the Einstein-vacuum equations with a (conveniently normalised) negative cosmological constant $\Lambda$, \footnote{Here $Rc$ and $Rs$ represent the $g$--Ricci and scalar curvatures respectively.}
\begin{equation}
\label{eq.eins} Rc - \frac{1}{2} Rs \cdot g + \Lambda \cdot g = 0 \text{,} \qquad \Lambda := - \frac{ n (n - 1) }{ 2 } \text{.}
\end{equation}
and suitably-defined `data' which characterises a particular boundary CFT. To this end, significant progress has been made in the \textit{Riemannian} \cite{Biquard_2008} and \textit{stationary Lorentzian} \cite{Chrusciel_2011} cases. The series of articles \cite{Holzegel_2016,Holzegel_2017,McGill_2020} initiated the study of the \textit{non-stationary Lorentzian} setting by considering uniqueness properties of wave equations on fixed aAdS backgrounds near the boundary. The upcoming work \cite{Holzegel_2021} aims to use the theoretical framework developed in these articles to address the `full' problem (i.e. considering solutions of \eqref{eq.eins} rather than wave equations on fixed backgrounds). The present article aims to demonstrate that, even without applying the thus-far developed theory to the Einstein-vacuum equations themselves, this framework is already sufficiently powerful to prove physically interesting correspondence statements.

Specifically, we will provide necessary and sufficient conditions on the boundary data for a given aAdS solution of \eqref{eq.eins} to be locally isometric to the \textit{pure AdS} spacetime \footnote{This is the maximally symmetric solution of \eqref{eq.eins}.} near the boundary. Such spacetimes locally have constant curvature but may have nontrivial global topology- see, for example, the so-called \textit{BTZ black holes} in (2+1) dimensions \cite{Banados_1992,Banados_1993} and their higher-dimensional counterparts \cite{Minneborg_1996,Banados_1997}, which arise via identification of points along orbits of Killing vectors in pure AdS. \footnote{As the name suggests, these \textit{topological black holes} have the global structure of a black hole; they may be characterised by a `mass' and and `angular momentum', they have horizons and it has been shown that they form via gravitational collapse \cite{Mann_1993}.}

The article \cite{Skenderis_2000} identified the boundary data defined by such spacetimes; the contribution of the present article is to prove the converse statement. We therefore have a rigidity statement for vacuum aAdS spacetimes locally isometric to pure AdS near the boundary that is formulated entirely in terms of boundary data. In other words, it is holographic in nature. 

The preceding discussion raises the following questions which will be informally addressed in the remainder of this section:
\begin{itemize}
    \item In what sense does an aAdS spacetime have a boundary?
    \item What constitutes appropriate boundary data in this setting?
    \item What is the precise statement of the holographic rigidity result?
    \item How does this result follow from the framework developed in the above-referenced series of articles?
\end{itemize}

\subsection{aAdS Spacetimes and the Correspondence Problem} \label{subsec.background}

Pure AdS spacetime is the maximally symmetric solution of \eqref{eq.eins}. It has the representation $(\mi{M}_0,g_0)$, where
\begin{equation} \label{eq.global_ads_metric}
    \mi{M}_0=\mathbb{R}^{n+1} \text{,} \qquad g_0=(1+r^2)^{-1}dr^2-(1+r^2)dt^2+r^2 \cdot \sphere \text{.}
\end{equation}
Here, we've covered the manifold using polar coordinates; $\sphere$ is the unit round metric. Under the coordinate transformation $r=\frac{1}{4}\rho^{-1}(2+\rho)(2-\rho)$ it has the representation $(\mi{M} , g)$, where
\begin{equation}\label{eq.ads_fg_metric_intro}
    \mi{M}=(0,2] \times \mathbb{R} \times \mathbb{S}^{n-1} \text{,} \qquad g=\rho^{-2}\left[d\rho^2-dt^2+\sphere-\frac{1}{2} \rho^2\paren{dt^2+\sphere} +\frac{1}{16} \rho^4\paren{-dt^2+\sphere}\right] \text{.}
\end{equation}
Here, $\rho$ is the coordinate for the $(0,2]$ component. In particular, we may attach a timelike \textit{conformal boundary}
\begin{equation}\label{eq.confbdry}
    (\mi{I},\mf{g}) \simeq (\mathbb{R} \times \mathbb{S}^{n-1} , -dt^2+\sphere)
\end{equation}
to $\rho^2 g$ at $\rho=0$. 

An \textit{asymptotically} anti-de Sitter (aAdS) spacetime region $\paren{\mi{M},g}$ has a manifold of the form
\begin{align}
    \mi{M} &= (0,\rho_0] \times \mi{I} \text{,}\qquad \rho_0>0 \text{,}
\end{align}
for smooth manifolds $\mi{I}$. On such $\mi{M}$, it is natural to consider so-called \textit{vertical} tensor fields, which only have components in directions tangential to $\mi{I}$. In other words, a vertical tensor field on an aAdS region $\mi{M}$ can be viewed as a $\rho \in (0,\rho_0]$--parametrised family of tensor fields on $\mi{I}$. They thus admit an intuitive and consistent notion of `boundary limit' as $\rho \searrow 0$.

On these manifolds, we consider metrics of the form
\begin{align}
    g=\rho^{-2}\paren{d\rho^2+\gv_{ab} \, dx^a dx^b} \text{,} \label{eq.FG_gauge}
\end{align}
where $\gv$ is a vertical tensor field ($x^a$ are $\rho$--transported coordinates on $\mi{I}$) and there exists some Lorentzian metric $\gm$ on $\mi{I}$ such that $\gv \to \gm$ as $\rho \searrow 0$. $\paren{\mi{I},\gm}$ is then the conformal boundary associated with $\paren{\mi{M},g,\rho}$ and the form \eqref{eq.FG_gauge} is referred to as \textit{Fefferman-Graham gauge} \cite{Fefferman_1985}. Accordingly, we refer to this class of spacetimes as \textit{FG-aAdS segments}.

\begin{remark}
    Imposing Fefferman-Graham gauge does not result in any loss of generality. One can always transform more general spacetime metrics with the same $\rho \searrow 0$ asymptotics into the form \eqref{eq.FG_gauge} via an appropriate change of coordinates. See, for example, \cite{Graham_1991,Graham_1999_1,Graham_1999_2}.
\end{remark}

In this article, we will restrict our attention to \textit{vacuum solutions}, i.e. to those FG-aAdS segments whose metrics solve \eqref{eq.eins}. In \cite{Shao_2020}, it was demonstrated that such segments admit partial expansions at the conformal boundary $\mi{I}$ even for vertical metrics of only finite regularity:
\begin{align} \label{eq.FG_metric}
    \gv = 
        \begin{cases} 
            \sum_{k=0}^{\frac{n-1}{2}} \rho^{2k} \mf{g}^{(2k)} + \rho^n\mf{g}^{(n)} + \rho^n \ms{r} \qquad &\text{$n$ odd,} \\ 
            \sum_{k=0}^{\frac{n-2}{2}} \rho^{2k} \mf{g}^{(2k)} +\rho^n\log\rho \, \mf{g}^{(\star)}+ \rho^n\mf{g}^{(n)} + \rho^n \ms{r} \qquad &\text{$n$ even,}
        \end{cases}
\end{align}
where $\ms{r}$ is a vertical tensor field with vanishing boundary limit and $\mf{g}^{(0)}=\gm$ is the conformal boundary metric introduced above. Of particular interest is the fact that all terms below order $n$ and the trace/divergence of $\gm^{(n)}$ are formally determined by $\mf{g}^{(0)}$ (itself is \textit{not} formally determined by \eqref{eq.eins}), while all terms beyond order $n$ are formally determined by $\mf{g}^{(0)}$ and the trace/divergence--free part of $\mf{g}^{(n)}$ (also not formally determined by \eqref{eq.eins}) \footnote{In the context of the AdS/CFT correspondence, $\mf{g}^{(n)}$ arises from the expectation value of the boundary CFT stress-energy tensor \cite{DeHaro_2001}.}. This identifies the pair $\paren{\mf{g}^{(0)} , \mf{g}^{(n)}}$ as appropriate data that one could prescribe on $\mi{I}$ in an attempt to characterise the bulk spacetime.

An important observation is that there exist coordinate transformations $\paren{\rho,x^a} \to \paren{\tilde{\rho},\tilde{x}^a}$ preserving the Fefferman-Graham gauge \eqref{eq.FG_gauge}, i.e.
\begin{align}
    g=\tilde{\rho}^{-2}\paren{d\tilde{\rho}^2+\tilde{\gv}_{ab} \, d\tilde{x}^a d\tilde{x}^b} \text{,}
\end{align}
but altering the coefficients in the partial expansion of the corresponding vertical metric $\tilde{\gv}$ (now given in terms of $\tilde{\rho}$). In particular, such a coordinate transformation induces a \textit{conformal transformation} of $\mf{g}^{(0)}$; for the higher-order coefficients, the effect is more complicated \cite{Imbimbo_2000,Skenderis_2001}. Since the underlying physical theory is invariant under such transformations, it therefore only makes sense for us to speak of \textit{gauge-equivalent} classes of boundary data $\brak{\gm^{(0)} , \gm^{(n)}}$ in the above sense.

A naive way to rigorously formulate the above-described AdS/CFT problem would be as a \textit{boundary Cauchy problem} for \eqref{eq.eins} with data $\brak{\gm^{(0)} , \gm^{(n)}}$ prescribed on $\mi{I}$. However, in analogy with the boundary Cauchy problem for the wave equation on a cylinder, one in fact expects this problem to be \textit{ill-posed} given general (i.e. possibly non-analytic) boundary data. Nevertheless, it remains apposite to ask whether solutions, if they exist, remain unique:
\begin{question}\label{q.EVE_UC}
    Assuming that a solution of \eqref{eq.eins} exists, is it uniquely determined in the bulk (up to isometry) by $\brak{ \gm^{(0)} , \gm^{(n)} }$?
\end{question}

\subsection{The Rigidity Result}

Let's consider the above question in the case of pure AdS. In \cite{Skenderis_2000}, the authors proved that if a vacuum FG-aAdS segment ($n>2$) is locally isometric to pure AdS \footnote{\textit{Not} necessarily globally isometric to pure AdS; the boundary topology could differ from that given in \eqref{eq.confbdry}.} then its metric is given near the conformal boundary by
\begin{align}
    g=\rho^{-2}\paren{d\rho^2+\gv_{ab}} \text{,}\qquad \gv=\gm_{ab}-\mf{p}_{ab} \cdot \rho^2 +\frac{1}{4}\,\gm^{cd}\,\mf{p}_{ac}\,\mf{p}_{bd}\cdot\rho^4 \text{,}
\end{align}
where $\gm$ is conformally flat and $\mf{p}$ is the $\gm$--\textit{Schouten tensor}
\begin{align}\label{eq.schouten_def}
    \mf{p}:=\frac{1}{n-2}\paren{\mf{Rc}-\frac{1}{2(n-1)}\mf{Rs} \cdot \gm}\text{,}
\end{align}
in which $\mf{Rc}$ and $\mf{Rs}$ are the $\gm$--Ricci and scalar curvatures respectively.

It is then natural to ask if the converse statement is true, i.e. if a vacuum FG-aAdS segment $\paren{\mi{M},g}$ has boundary data $\brak{\gm^{(0)},\gm^{(n)}}$ in which $\gm$ is conformally flat and
\begin{align} \label{eq.pure_ads_boundary_data_intro}
    \gm^{(n)}=
    \begin{cases}
        \frac{1}{4}\, \gm^{-1}\,\mf{p}\,\mf{p} \qquad&\text{,}\qquad n=4 \text{,} \\
        0 \qquad&\text{,}\qquad n \neq 4 \text{,}
    \end{cases}
\end{align}
must $\paren{\mi{M},g}$ be locally isometric to pure AdS near $\mi{I}$? In this article we provide a positive answer to this question subject to a key geometric condition on the boundary that was first introduced in \cite{McGill_2020}:
\begin{definition} \footnote{This definition is stated in terms of a smooth time function $t$ on the boundary. Formally, this must be a restriction of a \textit{global time} as given in Definition \ref{def.aads_time}.}
    Suppose $\paren{\mi{M},g}$ is a vacuum FG-aAdS segment. We say that the \textbf{null convexity criterion} is satisfied on $\paren{\mi{I} , \gm}$ if the bounds
    \begin{equation}\label{eq.NCC} 
        \mf{p} ( Z, Z ) \geq C ( Z t )^2 \text{,} \qquad
        | \Dm^2_{ Z Z } t | \leq B ( Z t )^2 \text{,}
    \end{equation}
    hold for some constants $0 \leq B < C$, where $Z$ is any vector field on $\mi{I}$ satisfying $\gm ( Z, Z ) = 0$ and $\Dm$ is the $\gm$--Levi-Civita connection.
\end{definition}

The main result of this article may now be stated as follows:
\begin{theorem}\footnote{The rigorous statement of this result in Theorem \ref{th.primary} makes precise sense of `near' using Definition \ref{def.Cweight}.} \label{th.introcorrespondence}
    Fix $n>2$. Suppose $(\mi{M},g)$ is an $(n+1)$--dimensional vacuum FG-aAdS segment for which:
    \begin{itemize}
        \item $\gv$ is sufficiently (finitely) regular.
        \item $\mi{I}$ has compact cross-sections.
        \item The null convexity criterion \eqref{eq.NCC} is satisfied.
        \item $\mi{T} = \mi{I} \cap \{ t_0<t<t_1 \}$ is a sufficiently long timespan on the boundary. \footnote{Determined by the constants $B$ and $C$  featuring in the null convexity criterion.}
    \end{itemize}
    Then $\paren{\mi{M},g}$ is locally isometric to pure AdS near $\mi{T}$ if and only if $\paren{\mi{M},g}$ has boundary data for which the following hold on $\mi{T}$: 
    \begin{itemize}
        \item $\gm^{(0)}$ is conformally flat.
        \item$\gm^{(n)}=
        \begin{cases}
            \frac{1}{4} \, \gm^{-1} \cdot \mf{p} \cdot \mf{p} \text{,}\qquad &n=4\text{,}\\
            0 \text{,}\qquad &n \neq 4\text{.}
        \end{cases}$
    \end{itemize}
\end{theorem}

\begin{remark}
    Proposition \ref{th.gauge_invariance} demonstrates that the stated conditions on $\gm^{(0)}$ and $\gm^{(n)}$ in Theorem \ref{th.introcorrespondence} are `gauge-invariant' in the sense that if one representative of a gauge-equivalent class of boundary data satisfies these conditions (and therefore corresponds to a FG-aAdS segment that is locally isometric to pure AdS) then the same must be true for all other representatives of the class.
\end{remark}

\begin{remark} \label{re.gen_ncc}
    The null convexity criterion as given in \eqref{eq.NCC} is gauge-dependent in the sense that it may hold for one representative of a gauge-equivalent class (in the above-defined sense) but not for another. However, for our purposes it is sufficient to demonstrate that \eqref{eq.NCC} holds for one particular representative of the gauge-equivalent class; once this is obtained, its key role is in the construction of a \textbf{pseudoconvex foliation} \cite{Holzegel_2016,Holzegel_2017,McGill_2020} of the near-boundary region (a gauge-invariant notion).
    
    In the upcoming work \cite{Chatzikaleas_2021} the authors formulate a gauge-invariant version of the null convexity criterion. One may straightforwardly replace the null convexity assumption and associated constructions in Theorem \ref{th.introcorrespondence} with this generalised criterion.
\end{remark}

%%%%%%%%%%%%%%%%%%%%%%%%%%%%%%%%%%%%%%%%%%%%%%%%%%%%%%%%

\subsection{Background and Proof Outline} \label{subsec.proofoutlines}

\subsubsection{Unique Continuation for the Weyl Curvature}

The \textit{Weyl curvature} is the traceless part of the Riemann tensor; if it vanishes for a solution of \eqref{eq.eins} then the solution has constant curvature (i.e. it is \textit{maximally symmetric}) and is thus locally isometric to pure AdS since \eqref{eq.eins} involves a negative cosmological constant \cite{Carroll_2004}. The key observation is that the Weyl curvature of a vacuum FG-aAdS segment satisfies a wave equation. Hence, to prove our result, we expect to be able to apply the theory developed in the sequence of articles \cite{Holzegel_2016,Holzegel_2017,McGill_2020} which considered tensorial solutions to wave equations of the form
\begin{equation} \label{eq.KG}
    (\Box_g+\sigma)u = \mathcal{G}(u,\nabla u) \text{,}\qquad \sigma \in \mathbb{R} \text{,}
\end{equation}
on fixed aAdS backgrounds \footnote{$\mathcal{G}(u,\nabla u)$ represents lower-order, possibly non-linear terms.}. Specifically, these articles aimed to determine if a solution $u$ of \eqref{eq.KG} with vanishing boundary data on $\mi{I}$ must necessarily vanish in the interior. This is what is known as a \textit{unique continuation} problem.

A positive answer to this problem was provided in the first of these articles, \cite{Holzegel_2016}. However, this rested on the assumptions that:
\begin{enumerate}
    \item The boundary metric is static.
    \item $u$ vanishes at a sufficiently fast rate \textit{along a sufficiently long timespan} on $\mi{I}$.
\end{enumerate}
The first assumption constitutes a significant restriction on the class of spacetimes for which the result may be applied. Furthermore, it was not known if the second assumption was strictly necessary.

The subsequent article \cite{Holzegel_2017} weakened assumption (1) by generalising to spacetimes with boundary metrics of only bounded `non-stationarity'. However, this still left open the question of whether assumption (2) could be weakened (or even removed). 

Most recently, \cite{McGill_2020} further weakened assumption (1) by permitting spacetimes with general time functions; the new assumption was then formulated as the null convexity criterion \eqref{eq.NCC}. This article also rigorously justified the necessity of assumption (2) by linking it to the trajectories of 

21 warnings
 near-boundary null geodesics. Specifically, upper and lower bounds (depending on the constants $B,C$ featuring in \eqref{eq.NCC}) were proved on the `time of return' of such geodesics to $\mi{I}$. Since counterexamples to unique continuation were constructed in \cite{Alinhac_1995} using geometric optics methods with these geodesics, \footnote{These counterexamples were only constructed for the case when the mass $\sigma$ is the specific value that makes the corresponding Klein-Gordon equation conformally invariant. In the upcoming work \cite{Guisset_2021} the authors intend to extend the construction to the non-conformal case.} the bounds show that there is:
\begin{itemize}
    \item A maximum timespan across which such counterexamples may be constructed.
    \item A minimum timespan across which we must assume vanishing of our field in order to eliminate the possibility of such counterexamples existing. Crucially, this minimum timespan matches the one in assumption (2).
\end{itemize}

FG-aAdS segments are constructed in such a way that it is natural to view tensor fields as \textit{mixed}, containing vertical components that are treated using $\gv$ and spacetime components that are treated using $g$. \cite{McGill_2020} accordingly developed a mixed covariant formalism which makes sense of higher-order derivatives acting on vertical tensor fields. To this end, an extension $\Dvm$ of the $\gv$--Levi-Civita connection $\Dv$ is constructed \footnote{In exactly such a way that it acts as a tensor derivation and is $\gv$-compatible.} so as to permit covariant derivatives of vertical tensor fields in all $\mi{M}$--directions. A `mixed' connection $\nablam$ is then defined in such a way that it acts as the standard $g$--Levi-Civita connection $\nabla$ on the spacetime components and $\Dvm$ on the vertical components of a given tensor field; the vertical wave operator $\Boxm_g$ is defined as the $g$--trace of $\nablam^2$.

Working in this framework, the article studied vertical tensor field solutions $\ms{u}$ of wave equations of the form
\begin{align}\label{eq.KG_vertical}
    (\Boxm_g+\sigma)\ms{u} = \mathcal{G}(\ms{u},\Dvm\ms{u}) \text{,}\qquad \sigma \in \mathbb{R} \text{.}
\end{align}
The main result was as follows:
\begin{theorem}\cite{McGill_2020}\label{th.carleman_est}
    Assume the following:
    \begin{itemize}
        \item $( \mi{M}, g )$ is a FG-aAdS segment satisfying the null convexity criterion \eqref{eq.NCC}.
    
        \item There is some $p > 0$ such that $\mathcal{G}$ in \eqref{eq.KG_vertical} satisfies the bound
        \begin{equation}\label{eq.intro_uc_G} 
            | \mc{G} ( \ms{u}, \Dvm \ms{u} ) |^2 \lesssim \rho^{ 4 + p } | \Dvm \ms{u} |^2 + \rho^{ 3 p } | \ms{u} |^2 \text{.}
        \end{equation}
        
        \item $\ms{u}$ is a solution of \eqref{eq.KG_vertical} for which---for sufficiently large $\kappa$ depending on $\sigma$, $\gv$, $t$ and the rank of $\ms{u}$---the limit
        \begin{equation}
            \rho^{-\kappa} \ms{u} \to  0 \text{,} \qquad \rho \searrow 0 \text{,}
        \end{equation}
        holds over a sufficiently long timespan $\mi{T} \subseteq \mi{I}$ determined by $B,C$ from \eqref{eq.NCC}.
    \end{itemize}
    Then $\ms{u} \equiv 0$ in some interior neighbourhood of $\mi{T}$. \footnote{In general, one also requires a compact support assumption for $\ms{u}$ on level sets of $(\rho,t)$. In our case this need not be of concern since we assume $\mi{I}$ has compact cross-sections.}
\end{theorem}

The key tool in the proof of this unique continuation statement was a novel \textit{Carleman estimate}, a form of which we state in Theorem \ref{th.cest} before applying for our own purposes in the present article.

The spacetime wave equation for the Weyl curvature can be decomposed into a system of vertical wave equations for the vertical tensor fields
\begin{align}
    \wv^0_{abcd}:= \rho^2 W_{abcd}\text{,}\qquad \wv^1_{abc}:= \rho^2 W_{\rho abc}\text{,}\qquad \wv^2_{ab}:= \rho^2 W_{\rho a \rho b}\text{,}
\end{align}
that, together, fully determine the spacetime Weyl curvature. Crucially, the lower-order nonlinearities in these equations satisfy \eqref{eq.intro_uc_G}. Theorem \ref{th.carleman_est} is used to demonstrated that the spacetime Weyl curvature identically vanishes near the boundary given vanishing to sufficiently high order on approach to the boundary for $\wv^0$, $\wv^1$ and $\wv^2$. \footnote{These `vanishing rates' should be understood in the context of the Fefferman-Graham expansions for the vertical components of the Weyl curvature as given in Corollary \ref{th.fg_exp_w}.}

\subsubsection{Connection to the Boundary Data}

\cite[Proposition 2.25]{Shao_2020} uses \eqref{eq.eins} to relate each of the vertical components of the Weyl curvature to vertical metric quantities:
\begin{align}
    \wv^0_{abcd}&=\Rv_{abcd}+\frac{1}{2}\li\gv_{a[c}\li \gv_{d]b}+\rho^{-1}\paren{\gv_{a[c}\li\gv_{d]b}-\gv_{b[c}\li \gv_{d]a}} \text{,} \\
    \wv^1_{cab}&=\Dv_{[b}\li\gv_{a]c} \text{,} \\
    \wv^2_{ab}&=-\frac{1}{2}\li ^2\gv_{ab}+\frac{1}{2}\rho^{-1}\li \gv_{ab} +\frac{1}{4}\gv^{cd} \li \gv_{ac} \li \gv_{bd} \text{.}
\end{align} 
Moreover, we have boundary expansions for the right-hand sides in which the leading-order terms feature the $\gm$--Weyl and Cotton tensors $\mf{W}$ and $\mf{C}$, whose vanishing follows from the condition that $\gm$ is conformally flat: \footnote{See Proposition \ref{th.weyl_leading_order} for a precise statements and proof of these leading order expressions.}
\begin{align} \label{eq.weyl_leading_order}
    \wv^0= \mf{W} + o\paren{1} \text{,} \qquad
    \wv^1=\frac{\rho}{2(n-2)} \mf{C} + o\paren{\rho} \text{,} \qquad
    \wv^2= o\paren{1} \text{.}
\end{align}
This yields a base level of vanishing for each of $\wv^0$, $\wv^1$ and $\wv^2$ which we improve by making use of the equations obtained by expressing the second Bianchi identity in terms of vertical objects (see Proposition \eqref{eq.vert_Bianchi_eqns}). Iteratively substituting the partial Fefferman-Graham expansions of $\wv^0$, $\wv^1$ and $\wv^2$ into these equations and applying the condition on $\gm^{(n)}$, we obtain
\begin{align}
    \wv^0 = o(\rho^{n-2}) \text{,}\qquad \wv^1 = o(\rho^{n-1}) \text{,}\qquad \wv^2 = o(\rho^{n-2}) \text{.}
\end{align}
Once this is obtained, it is possible to iteratively integrate the vertical Bianchi equations; at each iteration, degrees of vertical regularity are exchanged for additional orders of vanishing. The process is continued until the vanishing rate required for unique continuation is obtained.

%%%%%%%%%%%%%%%%%%%%%%%%%%%%%%%%%%%%%%%%%%%%%%%%%%%%%

\subsection{Organisation} \label{subsec.organisation}

\begin{itemize}
    \item In Section \ref{sec.aads} we formally define FG-aAdS segments. Vertical tensor fields are introduced and a notion of boundary limits for such objects is established.
    
    \item In Section \ref{sec.aads_mixed} we introduce the mixed tensor calculus which will enable us to make sense of a wave operator acting on a vertical tensor field in a consistent way. We also present formulae used to convert spacetime equations into their mixed counterparts.
    
    \item In Section \ref{sec.aads_vacuum}, we study FG-aAdS segments whose metrics solve \eqref{eq.eins}. The spacetime Weyl curvature for such segments satisfies a wave equation; we use the tools developed in Section \ref{sec.aads_mixed} to decompose this into a system of wave equations for each of the vertical components of the spacetime Weyl curvature.
    
    \item In Section \ref{sec.tf_ce} we introduce global time functions and state the Carleman estimate from \cite{McGill_2020} in a form suited to our purposes.
    
    \item In Section \ref{sec.results2} we prove some preliminary results before stating and proving the main result of this article.
\end{itemize}

%%%%%%%%%%%%%%%%%%%%%%%%%%%%%%%%%%%%%%%%%%%%%%%%%%%%%%%%%

\subsection{Acknowledgements}

The author would like to thank Arick Shao for his support via discussions on a number of topics and the provision of notes regarding the computations involved in the proofs of Propositions \ref{thm.prelim_comm} and \ref{thm.prelim_decomp}.

%%%%%%%%%%%%%%%%%%%%%%%%%%%%%%%%%%%%%%%%%%%%%%%%%%%%%%%%%

\section{Asymptotically AdS Spacetimes} \label{sec.aads}

We begin by recalling some basic definitions from \cite{McGill_2020} concerning the spacetime manifolds on which we will be working and the natural tensor fields to consider on them.

\begin{definition} \label{def.aads_manifold}
An \textbf{aAdS region} is a manifold of the form 
\begin{equation}
\label{eq.aads_manifold} \mi{M} := ( 0, \rho_0 ] \times \mi{I} \text{,} \qquad \rho_0 > 0 \text{,}
\end{equation}
in which $\mi{I}$ is a smooth $n$-dimensional manifold for some $n \in \N$. Given an aAdS region $\mi{M}$, $\rho$ denotes the coordinate function on $\mi{M}$ projecting onto the $( 0, \rho_0 ]$-component and $\partial_\rho$ denotes the $\mi{M}$--lift of the canonical vector field $d\rho$ on $(0, \rho_0]$.
\end{definition}

\begin{definition} \label{def.aads_vertical}
The \textbf{vertical bundle} $\ms{V}^k_l \mi{M}$ of rank $( k, l )$ over $\mi{M}$ is the manifold of all rank $( k, l )$ tensors on level sets of $\rho$ in $\mi{M}$:
\begin{equation}
\label{eq.aads_vertical} \ms{V}^k_l \mi{M} = \bigcup_{ \sigma \in ( 0, \rho_0 ] } T^k_l \{ \rho = \sigma \} \text{.}
\end{equation}
Sections of $\ms{V}^k_l \mi{M}$ are called \textbf{vertical tensor fields} of rank $( k, l )$.\footnote{A vertical tensor field of rank $( k, l )$ on an aAdS region $\mi{M}$ can be equivalently viewed as a $\rho \in ( 0, \rho_0 ]$-parameterised family of rank $( k, l )$ tensor fields on $\mi{I}$.}
\end{definition}

\begin{definition} \label{def.aads_tensor}
We adopt the following notational conventions and natural identifications:
\begin{itemize}
\item Italicized font (as in $g$) denotes tensor fields on $\mi{M}$.

\item Serif font (as in $\gv$) denotes vertical tensor fields. Any vertical tensor field $\ms{A}$ can be uniquely identified with a tensor field on $\mi{M}$ by demanding that the contraction of any component of $\ms{A}$ with $\partial_\rho$ or $d \rho$ identically vanishes.

\item Fraktur font (as in $\mf{g}$) denotes tensor fields on $\mi{I}$. If $\mf{A}$ is a tensor field on $\mi{I}$ then $\mf{A}$ will also denote the vertical tensor field on $\mi{M}$ obtained by extending $\mf{A}$ as a $\rho$-independent family of tensor fields on $\mi{I}$.
\end{itemize}
\end{definition}

\begin{definition} \label{def.aads_vertical_lie}
Let $\mi{M}$ be an aAdS region, and let $\ms{A}$ be a vertical tensor field.
\begin{itemize}
\item Given any $\sigma \in ( 0, \rho_0 ]$, $\ms{A} |_\sigma$ denotes the tensor field on $\mi{I}$ obtained from restricting $\ms{A}$ to the level set $\{ \rho = \sigma \}$ and then identifying $\{ \rho = \sigma \}$ with $\mi{I}$.

\item The \textbf{$\rho$-Lie derivative of $\ms{A}$}, denoted $\li \ms{A}$, is defined to be the vertical tensor field satisfying
\begin{equation}
\label{eq.aads_vertical_lie} \li \ms{A} |_\sigma = \lim_{ \sigma' \rightarrow \sigma } ( \sigma' - \sigma )^{-1} ( \ms{A} |_{ \sigma' } - \ms{A} |_\sigma ) \text{,} \qquad \sigma \in ( 0, \rho_0 ] \text{.}
\end{equation}
\end{itemize}
\end{definition}

Next, we establish coordinate system conventions on $\mi{I}$ and $\mi{M}$:

\begin{definition}
    Suppose $\mi{M}$ is an aAdS region and $( U, \varphi )$ is a coordinate system on $\mi{I}$. We write $\varphi_\rho := ( \rho, \varphi )$ to denote the corresponding lifted coordinates on $( 0, \rho_0 ] \times U$ and adopt the following notational conventions:
    \begin{itemize}
        \item Latin indices $a, b, c, \dots$ denote $\varphi$-coordinate components.
        \item Greek indices $\alpha, \beta, \mu, \nu, \dots$ denote $\varphi_\rho$-coordinate components.
    \end{itemize}
\end{definition}
\begin{definition} \label{def.aads_coord}
    Suppose $\mi{M}$ is an aAdS region. A coordinate system $( U, \varphi )$ on $\mi{I}$ is called \textbf{compact} iff $\bar{U}$ is a compact subset of $\mi{I}$ and $\varphi$ extends smoothly to an open neighborhood of $\bar{U}$.
\end{definition}

We now define a notion of magnitude for vertical tensor fields with respect to a given coordinate system and use this to make sense of boundary limits in a natural way:

\begin{definition} \label{def.aads_limit}
    Let $\mi{M}$ be an aAdS region and fix some $M \geq 0$.
    Furthermore, let $\ms{A}$ and $\mf{A}$ be a rank $(k, l)$ vertical tensor field and a rank $(k, l)$ tensor field on $\mi{I}$ respectively.
    \begin{itemize}
        \item Given a compact coordinate system $( U, \varphi )$ on $\mi{I}$, we define \footnote{Each $\partial_{a_i}$ denotes a $\varphi$--coordinate derivative.}
        \begin{equation}
            \label{eq.aads_norm} | \ms{A} |_{ M, \varphi } := \sum_{ m = 0 }^M \sum_{ \substack{ a_1, \dots, a_m \\ b_1, \dots, b_k \\ c_1, \dots, c_l } } | \partial^m_{ a_1 \dots a_m } \ms{A}^{ b_1 \dots b_k }_{ c_1 \dots c_l } | \text{.}
        \end{equation}
        
        \item We write $\ms{A} \rightarrow^M \mf{A}$ iff given any compact coordinate system $( U, \varphi )$ on $\mi{I}$,
        \begin{equation}
            \label{eq.aads_limit} \lim_{ \sigma \searrow 0 } \sup_{ \{ \sigma \} \times U } | \ms{A} - \mf{A} |_{ M, \varphi } = 0 \text{.}
        \end{equation}
        
        \item $\ms{A}$ is \textbf{weakly locally bounded} iff for any compact coordinate system $( U, \varphi )$ on $\mi{I}$,
        \begin{equation}\label{eq.metric_weakly_locally_bounded}
            \sup_U\int_0^{\rho_0} \left.\abs{\ms{A}}_{0 \text{,} \varphi} \right|_{\sigma} d\sigma < \infty \text{.}
        \end{equation}

        \item We additionally define a local uniform norm of $\ms{A}$:
        \begin{equation}\label{eq.local_uniform_norm}
            \norm{\ms{A}}_{M \text{,} \varphi}:=\sup_{(0,\rho_0] \times U} \abs{\ms{A}}_{M \text{,} \varphi} \text{.}
        \end{equation}
        
        \item $\ms{A}$ is \textbf{locally bounded in $C^M$} iff for any compact coordinate system $( U, \varphi )$ on $\mi{I}$,
        \begin{equation} \label{eq.metric_locally_bounded}
            \norm{\ms{A}}_{M \text{,} \varphi} < \infty \text{.}
        \end{equation}
    \end{itemize}
\end{definition}

Now let's use the above-defined notion of a boundary limit to rigorously define the class of metrics we're interested in.

\begin{definition} \label{def.aads_metric}
$( \mi{M}, g )$ is called a \textbf{FG-aAdS segment} iff the following hold:
\begin{itemize}
\item $\mi{M}$ is an aAdS region and $g$ is a Lorentzian metric on $\mi{M}$.

\item There exists a rank $( 0, 2 )$ vertical tensor field $\gv$ such that
\begin{equation}\label{eq.aads_metric} 
g := \rho^{-2} ( d \rho^2 + \gv ) \text{.}
\end{equation}

\item There exists a Lorentzian metric $\gm$ on $\mi{I}$ such that
\begin{equation}
\label{eq.aads_metric_limit} \gv \rightarrow^0 \gm \text{.}
\end{equation}
\end{itemize}
Given such a FG-aAdS segment,
\begin{itemize}
\item We refer to the form \eqref{eq.aads_metric} for $g$ as the \textbf{Fefferman--Graham gauge condition}.

\item $( \mi{I}, \gm )$ is the \textbf{conformal boundary} associated with $( \mi{M}, g, \rho )$. \footnote{As noted in Section \ref{sec.intro}, there exist coordinate transformations preserving \eqref{eq.aads_metric} but altering $\gv$. Such transformations induce conformal transformations of $\gm$, so it only makes sense to speak of the conformal boundary metric up to a conformal factor.}
\end{itemize}
\end{definition}

We also define a suitable regularity class for the vertical metrics considered in this article.

\begin{definition}
    We say that a FG-aAdS segment $\paren{\mi{M},g}$ is \textbf{$k$--regular} if $\gv$ is locally bounded in $C^{k+2}$ and $\li\gv$ is weakly locally bounded.
\end{definition}

For the sake of clarity, let's define some further notational conventions before continuing.

\begin{definition} \label{def.aads_covar}
Given a FG-aAdS segment $( \mi{M}, g )$, 
    \begin{itemize}
        \item $g^{-1}$, $\nabla$, $\nabla^\#$, $R$, $Rc$ and $Rs$ respectively denote the metric dual, Levi-Civita connection, gradient, Riemann curvature, Ricci curvature and scalar curvature with respect to $g$. 
    
        \item $\gv^{-1}$, $\Dv$, $\Dv^\#$, $\Rv$, $\Rcv$ and $\Rsv$ respectively denote the above objects with respect to $\gv$.
    
        \item $\gm^{-1}$, $\Dm$, $\Dm^\#$, $\Rm$, $\Rcm$ and $\Rsm$ respectively denote the above objects with respect to $\gm$.
    \end{itemize}
\end{definition}

%%%%%%%%%%%%%%%%%%%%%%%%%%%%%%%%%%%%%%%%%%%%%%%%%%%%%%%%%%%%

\section{The Mixed Tensor Calculus} \label{sec.aads_mixed}

\subsection{The Formalism}

In this section (which again follows the presentation of \cite{McGill_2020}), our aim is to make sense of a $g$-wave operator acting on a vertical tensor field---in such a way that it is compatible with standard covariant operations. Our first step in this direction is to construct connections on the vertical bundles which permit covariant derivatives of vertical tensor fields in all directions along $\mi{M}$.

\begin{definition}
    We denote multi-indices by $\ix{\mu} := \mu_1 \dots \mu_k$. Additionally, we write
    \begin{itemize}
        \item $\ixr{\mu}{i}{\al}$ to denote $\ix{\mu}$ with the $i^\text{th}$ component replaced with an $\al$-component.
        
        \item $\ixr{\mu}{i,j}{\al,\be}$ to denote $\ix{\mu}$ with the $i\textsuperscript{th}$ and $j\textsuperscript{th}$ components replaced with $\al$ and $\be$ components respectively.
    \end{itemize}
\end{definition}

\begin{proposition} \label{def.aads_vertical_connection}
    Let $( \mi{M}, g )$ be a FG-aAdS segment. There exists a unique connection $\Dvm$ on $\ms{V}^k_l\mi{M}$ such that the following hold for rank $\paren{k,l}$ vertical tensor fields $\ms{A}$ with respect to any coordinate system $( U, \varphi )$ on $\mi{I}$:
    \begin{align}
        \Dvm_c \ms{A}^{\ix{a}}_{\ix{b}} &= \Dv_c \ms{A}^{\ix{a}}_{\ix{b}} \text{,}\label{eq.c_deriv}\\
        \Dvm_\rho \ms{A}^{\ix{a}}_{\ix{b}} &= \li \ms{A}^{\ix{a}}_{\ix{b}} + \frac{1}{2} \sum_{ i = 1 }^{k} \gv^{a_i c} \li\gv_{cd} \ms{A}^{\ixr{a}{i}{d}}_{\ix{b}} - \frac{1}{2} \sum_{ j = 1 }^{l} \gv^{cd} \li\gv_{b_j c} \ms{A}^{\ix{a}}_{\ixr{b}{j}{d}} \text{.}\label{eq.rho_deriv}
    \end{align}
    Furthermore, for any vector field $X$ on $\mi{M}$,
    \begin{itemize}
        \item The following holds for vertical tensor fields $\ms{A}$ and $\ms{B}$:
    \begin{equation}
        \label{eq.aads_vertical_connection_leibniz} \bar{\Dv}_X ( \ms{A} \otimes \ms{B} ) = \bar{\Dv}_X \ms{A} \otimes \ms{B} + \ms{A} \otimes \bar{\Dv}_X \ms{B} \text{.}
    \end{equation}

    \item The following holds for vertical tensor fields $\ms{A}$ and tensor contraction operations $\mc{C}$:
    \begin{equation}
        \label{eq.aads_vertical_connection_contract} \bar{\Dv}_X ( \mc{C} \ms{A} ) = \mc{C} ( \bar{\Dv}_X \ms{A} ) \text{.}
    \end{equation}

    \item The connection $\Dvm$ is $\gv$--compatible:
    \begin{equation}
        \label{eq.aads_vertical_connection_compat} \bar{\Dv}_X \gv = 0 \text{,} \qquad \bar{\Dv}_X \gv^{-1} = 0 \text{.}
    \end{equation}
    \end{itemize}
\end{proposition}

\begin{proof}
See \cite[Proposition 2.23]{McGill_2020}.
\end{proof}

The connections $\bar{\Dv}$ extend the vertical Levi-Civita connections $\Dv$ to permit covariant derivatives of vertical fields in the $\rho$--direction. In order to construct the $g$-wave operator for vertical tensor fields in this spirit, we must first define some further tensorial objects on $\mi{M}$.

\begin{definition} \label{def.aads_mixed}
Let $( \mi{M}, g )$ be a FG-aAdS segment. The \textbf{mixed bundle} of rank $( \kappa, \lambda; k, l )$ over $\mi{M}$ is given by
\begin{equation}
\label{eq.aads_mixed} T^\kappa_\lambda \ms{V}^k_l \mi{M} := T^\kappa_\lambda \mi{M} \otimes \ms{V}^k_l \mi{M} \text{.}
\end{equation}
Sections of $T^\kappa_\lambda \ms{V}^k_l \mi{M}$ are called \textbf{mixed tensor fields} of rank $( \kappa, \lambda; k, l )$. Furthermore, the \textbf{bundle connection} $\bar{\nabla}$ on the mixed bundle $T^\kappa_\lambda \ms{V}^k_l \mi{M}$ is defined as the tensor product connection of $\nabla$ on $T^\kappa_\lambda \mi{M}$ and $\bar{\Dv}$ on $\ms{V}^k_l \mi{M}$. 
\end{definition}

\begin{proposition} \label{thm.aads_mixed_connection}
Let $( \mi{M}, g )$ be a FG-aAdS segment.
Then:
\begin{itemize}
\item For any vector field $X$ on $\mi{M}$ and mixed tensor fields $\mathbf{A}$ and $\mathbf{B}$,
\begin{equation}
\label{eq.aads_mixed_connection_leibniz} \nablam_X ( \mathbf{A} \otimes \mathbf{B} ) = \nablam_X \mathbf{A} \otimes \mathbf{B} + \mathbf{A} \otimes \nablam_X \mathbf{B} \text{.}
\end{equation}

\item For any vector field $X$ on $\mi{M}$,
\begin{equation}
\label{eq.aads_mixed_connection_compat} \nablam_X g = 0 \text{,} \qquad \nablam_X g^{-1} = 0 \text{,} \qquad \nablam_X \gv = 0 \text{,} \qquad \nablam_X \gv^{-1} = 0 \text{.}
\end{equation}
\end{itemize}
\end{proposition}

\begin{proof}
See \cite[Proposition 2.28]{McGill_2020}.
\end{proof}

Generally speaking, the mixed connections $\nablam$ behave like $\nabla$ on spacetime components and $\bar{\Dv}$ on vertical components. The properties demonstrated in Proposition \ref{thm.aads_mixed_connection} are analogous to the properties of covariant derivatives that enable the standard integration by parts formulae; we are thus able to extend these directly to mixed tensor fields.

We now define higher covariant derivatives in the context of mixed bundles:

\begin{definition} \label{def.aads_mixed_wave}
Let $( \mi{M}, g )$ be a FG-aAdS segment and $\mathbf{A}$ be a mixed tensor field of rank $( \kappa, \lambda; k, l )$.
\begin{itemize}
\item The \textbf{mixed covariant differential} of $\mathbf{A}$ is the mixed tensor field $\bar{\nabla} \mathbf{A}$, of rank $( \kappa, \lambda + 1; k, l )$, that maps each vector field $X$ on $\mi{M}$ to $\bar{\nabla}_X \mathbf{A}$.

\item The \textbf{mixed Hessian} $\bar{\nabla}^2 \mathbf{A}$ is defined as the mixed covariant differential of $\bar{\nabla} \mathbf{A}$.

\item The \textbf{wave operator} $\bar{\Box} \mathbf{A}$ is the $g$-trace of $\bar{\nabla}^2 \mathbf{A}$.

\end{itemize}
\end{definition}

%%%%%%%%%%%%%%%%%%%%%%%%%%%%%%%%%%%%%%%%%%%%%%%%%%%%

\subsection{Conversion Formulae}\label{sec.st_vert_conversion}

In this article we will need to convert equations for spacetime quantities to corresponding equations for vertical quantities. In this subsection, we present a systematic method for doing so. The schematic notation and the computations involved in the proofs for this section have been assembled by Arick Shao in preparation for the upcoming work \cite{Holzegel_2021}, and kindly shared with the author for the present article. 

Let's begin by fixing schematic notations for asymptotic error terms:

\begin{definition} \label{def.prelim_O}
    Let $( \mi{M}, g )$ be a FG-aAdS segment, fix an integer $M \geq 0$ and let $h \in C^\infty ( \mi{M} )$. Then $\mi{O}_M ( h )$ refers to any vertical tensor field $\ms{a}$ satisfying
    \begin{equation}
        \label{eq.prelim_OO} | \ms{a} |_{ M, \varphi } \lesssim_\varphi h \text{,}
    \end{equation}
    for any compact coordinate system $( U, \varphi )$ on $\mi{I}$.

    Furthermore, given a vertical tensor field $\ms{B}$, $\mi{O}_M ( h; \ms{B} )$ refers to any vertical tensor field $\ms{A}$ that is expressible in the form
    \begin{equation}
        \label{eq.prelim_O} \ms{A} = \sum_{ k = 1 }^N \mi{C}_k ( \ms{a}_k \otimes \ms{B}^\ast_k ) \text{,}
    \end{equation}
    where $N \geq 0$ and for each $1 \leq k \leq N$,
    \begin{itemize}
        \item $\ms{B}^\ast_k$ is $\ms{B}$ composed with some permutation of its components.

        \item $\ms{a}_k$ is a vertical tensor field satisfying $\ms{a}_k = \mi{O}_M ( h )$.

        \item $\mi{C}_k$ is a composition of zero or more contractions and $\gv$-metric contraction operations.
    \end{itemize}
\end{definition}

Next, we establish some commutation identities for vertical tensor fields:

\begin{proposition} \label{thm.prelim_comm}
Let $( \mi{M}, g )$ be a FG-aAdS segment, fix $M \geq 2$ and assume
\begin{equation}
\label{eq.prelim_regular} \gv = \mi{O}_M ( 1 ) \text{,} \qquad \li \gv = \mi{O}_{ M - 2 } ( \rho ) \text{.}
\end{equation}
Then the following commutation identities hold for any vertical tensor field $\ms{A}$ and $p \in \mathbb{R}$:
\begin{align}
\label{eq.prelim_comm} \Dvm_\rho ( \Dv \ms{A} ) &= \Dv ( \Dvm_\rho \ms{A} ) + \mi{O}_{ M - 2 } ( \rho ; \Dv \ms{A} ) + \mi{O}_{ M - 3 } ( \rho; \ms{A} ) \\
\label{eq.prelim_comm_rho} \Boxm ( \rho^p \ms{A} ) &= \rho^p \Boxm \ms{A} + 2 p \rho^{p+1} \Dvm_\rho \ms{A} - p ( n - p ) \rho^p \ms{A} + \mi{O}_{ M - 2 } ( \rho^2; \rho^p \ms{A} ) \text{.}
\end{align}
\end{proposition}

\begin{proof}
    See Appendix \ref{sec.st_vert_conversion_proof1}.
\end{proof}

We now fix some further notation so as to be able to express the conversion formulae in a compact form:

\begin{definition}\label{def.conversion_conventions}
    Suppose $A$ is a tensor field on $\mi{M}$ of rank $( 0, r_1 + r_2 )$, where $r_1, r_2 \geq 0$, and let $\ms{A}$ be the corresponding rank $( 0, r_2 )$ vertical tensor field defined with respect to any coordinates $( U, \varphi )$ on $\mi{I}$ by
    \begin{equation}
        \label{eq.prelim_decomp_phi} \ms{A}_{ \ix{a} } := A_{ \ix{\rho} \ix{a} } \text{,}
    \end{equation}
    where the multi-index $\ix{\rho} := \rho \dots \rho$ represents $r_1$ copies of $\rho$, while $\ix{a} := a_1 \dots a_{ r_2 }$. Then:
    \begin{itemize}
    \item For any $1 \leq i \leq r_1$, the rank $( 0, r_2 + 1 )$ vertical tensor field $\ms{A}^\rho_i$ is given by
    \begin{equation}
        \label{eq.prelim_decomp_phi_rv} ( \ms{A}^\rho_i )_{ b \ix{a} } := A_{ \ixr{\rho}{i}{b} \ix{a} } \text{,}
    \end{equation}
    
    \item For any $1 \leq j \leq r_2$, the rank $( 0, r_2 - 1 )$ vertical tensor field $\ms{A}^v_j$ is given by
    \begin{equation}
        \label{eq.prelim_decomp_phi_vr} ( \ms{A}^v_j )_{ \ixd{a}{j} } := A_{ \ix{\rho} \ixr{a}{j}{\rho} } \text{,}
    \end{equation}

    \item For any $1 \leq i, j \leq r_1$ with $i \neq j$, the rank $( 0, r_2 + 2 )$ vertical field $\ms{A}^{ \rho, \rho }_{ i, j }$ is given by
    \begin{equation}
        \label{eq.prelim_decomp_phi_rrvv} ( \ms{A}^{ \rho, \rho }_{ i, j } )_{ c b \ix{a} } := A_{ \ixr{\rho}{i,j}{c,b} \ix{a} } \text{.}
    \end{equation}

    \item For any $1 \leq i, j \leq r_2$ with $i \neq j$, the rank $( 0, r_2 - 2 )$ vertical field $\ms{A}^{ v, v }_{ i, j }$ is given by
        \begin{equation}
        \label{eq.prelim_decomp_phi_vvrr} ( \ms{A}^{ v, v }_{ i, j } )_{ \ixd{a}{i,j} } := A_{ \ix{\rho} \ixr{a}{i,j}{\rho,\rho} } \text{.}
    \end{equation}

    \item For any $1 \leq i \leq r_1$ and $1 \leq j \leq r_2$, the rank $( 0, r_2 )$ vertical field $\ms{A}^{ \rho, v }_{ i, j }$ is given by
    \begin{equation}
        \label{eq.prelim_decomp_phi_rvvr} ( \ms{A}^{ \rho, v }_{ i, j } )_{ b \ixd{a}{j} } := A_{ \ixr{\rho}{i}{b} \ixr{a}{j}{\rho} } \text{.}
    \end{equation}
    \end{itemize}
\end{definition}

We are now in a position to state the conversion formulae:

\begin{proposition} \label{thm.prelim_decomp}
Let $( \mi{M}, g )$ be a FG-aAdS segment and assume \eqref{eq.prelim_regular} holds for some $M \geq 2$.
Let $A$ be a tensor field on $\mi{M}$ of rank $( 0, r_1 + r_2 )$, where $r_1, r_2 \geq 0$, and let $\ms{A}$ be the associated rank $( 0, r_2 )$ vertical tensor field defined by \eqref{eq.prelim_decomp_phi}. 

Then the following identities hold with respect to any coordinates $( U, \varphi )$ on $\mi{I}$:
\begin{align}
\label{eq.prelim_decomp_rho} \nabla_\rho A_{ \ix{\rho} \ix{a} } &= \rho^{ - r_1 - r_2 } \Dvm_\rho ( \rho^{ r_1 + r_2 } \ms{A} )_{ \ix{a} } \text{,} \\
\label{eq.prelim_decomp_a} \nabla_c A_{ \ix{\rho} \ix{a} } &= \Dvm_c \ms{A}_{ \ix{a} } + ( \rho^{-1} \delta^b_c - \frac{1}{2} \gv^{bd} \li \gv_{dc} ) \sum_{ i = 1 }^{ r_1 } ( \ms{A}^\rho_i )_{ b \ix{a} } - \sum_{ j = 1 }^{ r_2 } ( \rho^{-1} \gv_{c a_j} - \frac{1}{2} \li \gv_{c a_j} ) \, ( \ms{A}^v_j )_{ \ixd{a}{j} } \text{.} \\
\notag &= \Dvm_c \ms{A}_{ \ix{a} } + \rho^{-1} \sum_{ i = 1 }^{ r_1 } ( \ms{A}^\rho_i )_{ c \ix{a} } - \rho^{-1} \sum_{ j = 1 }^{ r_2 } \gv_{ c a_j } \, ( \ms{A}^v_j )_{ \ixd{a}{j} } \\
\notag &\qquad + \sum_{ i = 1 }^{ r_1 } \mi{O}_{ M - 2 } ( \rho; \ms{A}^\rho_i )_{ c \ix{a} } + \sum_{ j = 1 }^{ r_2 } \mi{O}_{ M - 2 } ( \rho; \ms{A}^v_j )_{ c \ix{a} } \text{,} \\
\label{eq.box} \Box A_{ \ix{\rho} \ix{a} } &= \rho^{ - r_1 - r_2 } \Boxm ( \rho^{ r_1 + r_2 } \ms{A} )_{ \ix{a} } + 2 \rho \sum_{ i = 1 }^{ r_1 } \gv^{ b c } \, \Dvm_b ( \ms{A}^\rho_i )_{ c \ix{a} } - 2 \rho \sum_{ j = 1 }^{ r_2 } \Dvm_{ a_j } ( \ms{A}^v_j )_{ \ixd{a}{j} } - ( n r_1 + r_2 ) \, \ms{A}_{ \ix{a} } \\
\notag &\qquad - 2 \sum_{ i = 1 }^{ r_1 } \sum_{ j = 1 }^{ r_2 } ( \ms{A}^{ \rho, v }_{ i, j } )_{ a_j \ixd{a}{j} } + 2 \sum_{ 1 \leq i < j \leq r_1 } \gv^{ b c } \, ( \ms{A}^{ \rho, \rho }_{ i, j } )_{ b c \ix{a} } + 2 \sum_{ 1 \leq i < j \leq r_2 } \gv_{ a_i a_j } \, ( \ms{A}^{ v, v }_{ i, j } )_{ \ixd{a}{i,j} } \\
\notag &\qquad + \sum_{ i = 1 }^{ r_1 } \mi{O}_{ M - 2 } ( \rho^3; \Dvm \ms{A}^\rho_i )_{ \ix{a} } + \sum_{ j = 1 }^{ r_2 } \mi{O}_{ M - 2 } ( \rho^3; \Dvm \ms{A}^v_j )_{ \ix{a} } + \sum_{ i = 1 }^{ r_1 } \mi{O}_{ M - 3 } ( \rho^3; \ms{A}^\rho_i )_{ \ix{a} } \\
\notag &\qquad + \sum_{ j = 1 }^{ r_2 } \mi{O}_{ M - 3 } ( \rho^3; \ms{A}^v_j )_{ \ix{a} } + \mi{O}_{ M - 2 } ( \rho^2; \ms{A} )_{ \ix{a} } + \sum_{ i = 1 }^{ r_1 } \sum_{ j = 1 }^{ r_2 } \mi{O}_{ M - 2 } ( \rho^2; \ms{A}^{ \rho, v }_{ i, j } )_{ \ix{a} } \\
\notag &\qquad + \sum_{ 1 \leq i < j \leq r_1 } \mi{O}_{ M - 2 } ( \rho^2; \ms{A}^{ \rho, \rho }_{ i, j } )_{ \ix{a} } + \sum_{ 1 \leq i < j \leq r_2} \mi{O}_{ M - 2 } ( \rho^2; \ms{A}^{ v, v }_{ i, j } )_{ \ix{a} } \text{,}
\end{align}

\end{proposition}

\begin{proof}
    See Appendix \ref{sec.st_vert_conversion_proof2}.
\end{proof}

%%%%%%%%%%%%%%%%%%%%%%%%%%%%%%%%%%%%%%%%%%%%%%%%%%%%%%%%%%%

\section{Vacuum Spacetimes} \label{sec.aads_vacuum}

\subsection{The Metric}\label{sec.boundary_data}

\begin{definition} \label{def.aads_vacuum}
    An $(n+1)$--dimensional FG-aAdS segment $( \mi{M}, g )$ is called a \textbf{vacuum FG-aAdS segment} iff it satisfies \eqref{eq.eins}.
\end{definition}

The following boundary limits were derived in \cite[Theorem 3.3]{Shao_2020}:
\begin{theorem} \label{th.metric_limits}
    Fix $n>2$ and $M_0 \geq n+2$. Suppose that $(\mi{M},g)$ is an $(n+1)$--dimensional, $M_0$--regular vacuum FG-aAdS segment. Then
    \begin{align} \label{eq.metric_limits}
        \gv \rightarrow ^{M_0} \, \gm \text{,}\qquad \gv^{-1} \rightarrow ^{M_0} \, \gm^{-1} \text{,}
    \end{align}
    and, for $0 \leq k < n$, there exists tensor fields $\gm^{(k)}$ on $\mi{I}$ such that
    \begin{align} \label{eq.g_boundary_limits}
        \li^k\gv \rightarrow^{M_0-k}\, k! \, \gm^{(k)} \text{,} \qquad \rho\li^{k+1}\gv \rightarrow^{M_0-k} \, 0 \text{,}
    \end{align}
    where
    \begin{itemize}
        \item $\gm^{(0)} = \gm$ \text{.}
        \item $\gm^{(2)} = -\mf{p}$, where $\mf{p}$ is the $\gm$--Schouten tensor:
        \begin{align}
            \mf{p} = \frac{1}{n-2} \paren{ \Rcm - \frac{1}{2(n-2)} \Rsm \cdot \gm } \text{.}
        \end{align}
        \item $\gm^{(k)}=0$ if $n$ is odd.
    \end{itemize}
    Furthermore, there exist tensor fields $\gm^{(\star)}$ and $\gm^{(\dagger)}$ \footnote{$\gm^{\dagger}$ is in $C^{M_0-n}$ on $\mi{I}$.} on $\mi{I}$ such that
    \begin{align}
        \rho\li^{n+1}\gv \rightarrow^{M_0-n} \, n! \, \gm^{(\star)} \text{,}\qquad \li^n\gm-n!\,(\log\rho)\gm^{(\star)} \rightarrow^{M_0-n} \, n! \, \gm^{(\dagger)} \text{,}
    \end{align}
    where $\gm^{(\star)}=0$ if $n$ is odd.
\end{theorem}

Note that the above result implies that any vacuum FG-aAdS segment is a `strongly' FG-aAdS segment, as defined in \cite[Definition 2.13]{McGill_2020}; this was the condition required for the Carleman estimate (and hence the unique continuation result) of that article to hold.

Using the above limits, the following precise statement of the Fefferman-Graham expansion \eqref{eq.FG_metric} for sufficiently regular vacuum FG-aAdS segments was given in \cite[Theorem 3.6]{Shao_2020}: \footnote{Similar expansions were also derived for $\Rv$ and $\Dv\li\gv$. For the sake of brevity, we do not reproduce these here.}
\begin{corollary}\label{th.metric_expansion_for_vac_fgaads}
    Fix $n>2$ and $M_0 \geq n+2$. Suppose that $(\mi{M},g)$ is an $(n+1)$--dimensional, $M_0$--regular vacuum FG-aAdS segment. Let $\gm^{(k)}$ for $0 \leq k < n$ and $\gm^{(\star)}$ be as in Theorem \ref{th.metric_limits}. Then there exists a $C^{M_0-n}$ tensor field $\gm^{(n)}$ on $\mi{I}$ and a vertical tensor field $\ms{r}$ such that 
    \begin{align}\label{eq.metric_FG_expansion}
        \gv=
        \begin{cases}
            \sum_{k=0}^{\frac{n-1}{2}} \rho^{2k} \, \gm^{(2k)} + \rho^n \, \gm^{(n)} + \rho^n \, \ms{r} \text{,} \qquad &\text{$n$ odd}\text{,} \\
            \sum_{k=0}^{\frac{n-2}{2}} \rho^{2k} \, \gm^{(2k)} + \rho^n\log\rho \, \gm^{(\star)} + \rho^n \, \gm^{(n)} + \rho^n \, \ms{r} \text{,} \qquad &\text{$n$ even}\text{,}
        \end{cases}
    \end{align}
    where the `remainder' $\ms{r}$ satisfies
    \begin{align}\label{eq.metric_expansion_remainder}
        \ms{r} \to^{M_0-n} \,\, 0 \text{.}
    \end{align}
\end{corollary}

\begin{remark}
    The above shows that if $\paren{\mi{M},g}$ is a $M_0$--regular vacuum FG-aAdS segment then \eqref{eq.prelim_regular} holds for $M=M_0$.
\end{remark}

%%%%%%%%%%%%%%%%%%%%%%%%%%%%%%%%%%%%%%%%%%%%%%%%%%%%%%

\subsection{The Weyl Curvature}

\begin{proposition} \label{thm.einstein_ex}
Suppose that $( \mi{M}, g )$ is an $(n+1)$--dimensional vacuum FG-aAdS segment. Let $( U, \varphi )$ be a coordinate system on $\mi{I}$.
Then the following hold with respect to $\varphi_\rho$-coordinates:
\begin{equation}\label{eq.einstein_ex} 
    Rc_{ \alpha \beta } = - n \cdot g_{ \alpha \beta } \text{,} \qquad Rs = - n (n + 1) \text{,} \qquad W_{ \alpha \beta \gamma \delta } = R_{ \alpha \beta \gamma \delta } + g_{ \alpha \gamma } g_{ \beta \delta } - g_{ \alpha \delta } g_{ \beta \gamma } \text{,}
\end{equation}
where $W$ is the $g$--Weyl curvature. \footnote{That is, the $g$--traceless part of the $g$--Riemann curvature.} Furthermore, the Weyl curvature satisfies the wave equation
\begin{align}\label{eq.stwave}
    \paren{\Box_g+2n}\W{_\al_\be_\ga_\de}=4\W{^\la_\al^\mu_{[\de|}} \W{_\la_\be_\mu_{|\ga]}}-\W{^\la^\mu_\ga_\de}\W{_\al_\be_\la_\mu}\text{.}
\end{align}
\end{proposition}

\begin{proof}
    The first two identities in \eqref{eq.einstein_ex} follow by taking the trace of \eqref{eq.eins}. The third identity is a substitution of these expressions into the Weyl curvature,
    \begin{align}
        W_{ \al \be \ga \de } = R_{ \al \be \ga \de } - \frac{2}{n-1} \paren{g_{ \al [\ga } Rc_{\de ] \be} - g_{ \be [\ga } Rc_{\de] \al}} + \frac{2}{n(n-1)} Rs \cdot g_{\al[\ga} g_{\de]\be}\text{.}
    \end{align}
    
    Next, consider the Bianchi equation
    \begin{align}\label{eq.bianchi}
        \nabla_{[\mu} W_{\al\be]\ga\de}=0 \text{,}
    \end{align}
    and its trace
    \begin{align}\label{eq.divwey}
        \nabla^\mu\W{_\al_\be_\de_\mu}=0 \text{.}
    \end{align}
    Taking the divergence of \eqref{eq.bianchi} gives
    \begin{align}
        \Box_g\W{_\al_\be_\ga_\de}+g^{\nu\mu}\nabla_\nu\nabla_\ga\W{_\al_\be_\de_\mu}+g^{\nu\mu}\nabla_\nu\nabla_\delta\W{_\al_\be_\mu_\ga}=0 \text{.}
    \end{align}
    The covariant derivatives are permuted using
    \begin{align}
        \left[\nabla_\nu , \nabla_\ga\right]\W{_\al_\be_\de_\mu}=-\R{^\la_\al_\nu_\ga}\W{_\la_\be_\de_\mu}-\R{^\la_\be_\nu_\gamma}\W{_\al_\la_\de_\mu}-\R{^\la_\de_\nu_\ga}\W{_\al_\be_\la_\mu}-\R{^\la_\mu_\nu_\ga}\W{_\al_\be_\de_\la} \text{,}
    \end{align}
    to give
    \begin{align}
    &\Box_g\W{_\al_\be_\ga_\de}+2\nabla_{[\ga|} \nabla^\mu\W{_\al_\be_{|\de]}_\mu}-2g^{\nu\mu}\left(\R{^\la_\al_\nu_{[\ga|}}\W{_\la_\be_{|\de]}_\mu}\right. \\
    &\left.\qquad+\R{^\la_\be_\nu_{[\ga|}}\W{_\al_\la_{|\de]}_\mu}+\R{^\la_{[\de|}_\nu_{|\ga]}}\W{_\al_\be_\la_\mu}+\R{^\la_\mu_\nu_{[\ga|}}\W{_\al_\be_{|\de]}_\la}\right)=0 \text{,} \notag
    \end{align}
    in which the second term vanishes by \eqref{eq.divwey}. For the remaining terms, one uses \eqref{eq.einstein_ex} to replace the Riemann curvature with the Weyl curvature and the metric. After simplifying, one finds
    \begin{align}
        \Box_g\W{_\al_\be_\ga_\de}+2n\W{_\al_\be_\ga_\de}+2\W{^\la_\al^\mu_\ga}\W{_\la_\be_\mu_\de}+\W{^\la^\mu_\ga_\de}\W{_\al_\be_\la_\mu}-2\W{^\la_\al^\mu_\de}\W{_\la_\be_\mu_\ga} = 0 \text{,}
    \end{align}
    from which \eqref{eq.stwave} follows immediately.
\end{proof}

Our goal is to use the mixed covariant formalism to convert the spacetime wave equation \eqref{eq.stwave} into a system of wave equations for each of the independent vertical components of the Weyl curvature, defined below.

\begin{definition} \label{def.vertweylfields}
    Suppose $\paren{\mi{M}, g}$ is a FG-aAdS segment and $\paren{U,\varphi}$ is a coordinate system on $\mi{I}$. With respect to $\varphi_\rho$--coordinates, we define the following independent vertical components of the Weyl curvature:
    \begin{align}
        \wv^0_{abcd}:=\rho^2W_{abcd} \text{,}\qquad \wv^1_{abc}:=\rho^2W_{\rho abc}\text{,}\qquad \wv^2_{ab}:=\rho^2W_{\rho a \rho b}\text{.}
    \end{align}
\end{definition}

\begin{remark}
    Since the spacetime Weyl curvature is trace-free, one has that
    \begin{align} \label{eq.w2_tr_w0}
        \wv^2_{ab}=g^{\rho\rho}W_{\rho a \rho b}=-g^{cd}W_{cadb}=-\gv^{cd}\wv^0_{cadb}.
    \end{align}
    In other words, $-\wv^2$ is a $\gv$--trace of $\wv^0$.
\end{remark}

\begin{definition}
    Suppose $\paren{\mi{M}, g}$ is a FG-aAdS segment and $\paren{U,\varphi}$ is a coordinate system on $\mi{I}$. With respect to $\varphi_\rho$--coordinates, the $\gv$-traceless part of $\wv^0$ is 
    \begin{equation}\label{eq.w0_traceless}
        \hat{\wv}^0_{abcd}:=\wv^0_{abcd}+\frac{2}{n-2}\paren{\gv_{a[c}\wv^2_{d]b}+\gv_{b[d}\wv^2_{c]a}}\text{.}
    \end{equation}
\end{definition}

In \cite{Shao_2020} it was demonstrated that the vertical components of the Weyl curvature for a vacuum FG-aAdS segment can be expressed in terms of $\gv$, $\gv^{-1}$, $\li\gv$ and $\Rv$:

\begin{proposition}\cite[Proposition 2.25]{Shao_2020} \label{th.w_g_rels}
    Suppose $\paren{\mi{M},g}$ is a vacuum FG-aAdS segment and $\paren{U,\varphi}$ is a coordinate system on $\mi{I}$. Then the following relations hold with respect to $\varphi$--coordinates:
    \begin{align}
        \wv^0_{abcd}&=\Rv_{abcd}+\frac{1}{2}\li\gv_{a[c}\li \gv_{d]b}+\rho^{-1}\paren{\gv_{a[c}\li\gv_{d]b}-\gv_{b[c}\li \gv_{d]a}} \text{,} \label{eq.w0.rel}\\
        \wv^1_{cab}&=\Dv_{[b}\li\gv_{a]c} \text{,} \label{eq.w1.rel}\\
        \wv^2_{ab}&=-\frac{1}{2}\li ^2\gv_{ab}+\frac{1}{2}\rho^{-1}\li \gv_{ab} +\frac{1}{4}\gv^{cd} \li \gv_{ac} \li \gv_{bd} \text{.} \label{eq.w2.rel}
    \end{align} 
\end{proposition}

The limits proved in \cite[Theorem 3.6]{Shao_2020} may be straightforwardly applied to derive similar limits for derivatives of the vertical components of the Weyl curvature, and hence to write down Fefferman-Graham expansions similar to \eqref{eq.metric_FG_expansion}:

\begin{theorem} \label{th.w_limits}
    Fix $n>2$ and $M_0 \geq n+2$. Suppose that $(\mi{M},g)$ is an $(n+1)$--dimensional, $M_0$--regular vacuum FG-aAdS segment. Then, for $0 \leq k <n-2$ and $i=0,1,2$, there exist tensor fields $\mf{W}_{i}^{k}$ on $\mi{I}$ such that
    \begin{align}
        \li^k\wv^0 &\rightarrow^{M_0-k-2} \, k! \, \mf{W}_0^{(k)} \text{,}\qquad \rho\li^{k+1}\wv^0 \rightarrow^{M_0-k-2} \, 0 \text{,}\qquad 0 \leq k < n-2 \text{,} \label{eq.w0_lims}\\
        \li^k\wv^1 &\rightarrow^{M_0-k-2} \, k! \, \mf{W}_1^{(k)} \text{,}\qquad \rho\li^{k+1}\wv^1 \rightarrow^{M_0-k-2} \, 0 \text{,}\qquad 0 \leq k < n-1 \text{,}\label{eq.w1_lims}\\
        \li^k\wv^2 &\rightarrow^{M_0-k-2} \, k! \, \mf{W}_2^{(k)} \text{,}\qquad \rho\li^{k+1}\wv^2 \rightarrow^{M_0-k-2} \, 0 \text{,}\qquad 0 \leq k < n-2 \text{,}\label{eq.w2_lims}
    \end{align}
    where $\mf{W}_0^{(k)}=\mf{W}_2^{(k)}=0$ if $k$ is odd and $\mf{W}_1^{(k)}=0$ if $k$ is even. Furthermore, there exist tensor fields $\mf{W}_i^{(\star)}$ and $\mf{W}_i^{(\dagger)}$ on $\mi{I}$ such that
    \begin{align}
        &\begin{cases}
            \rho\li^{n-1}\wv^0 \rightarrow^{M_0-n} \, (n-2)! \, \mf{W}_0^{(\star)} \text{,}\\
            \li^{n-2}\wv^0-(n-2)!\,(\log\rho)\mf{W}_0^{(\star)} \rightarrow^{M_0-n} \, (n-2)! \, \mf{W}_0^{(\dagger)} \text{,}
        \end{cases}\\
        &\begin{cases}
            \rho\li^{n}\wv^1 \rightarrow^{M_0-n-1} \, (n-1)! \, \mf{W}_1^{(\star)} \text{,}\\
            \li^{n-1}\wv^1-(n-1)!\,(\log\rho)\mf{W}_1^{(\star)} \rightarrow^{M_0-n-1} \, (n-1)! \, \mf{W}_1^{(\dagger)} \text{,}
        \end{cases}\\
        &\begin{cases}
            \rho\li^{n-1}\wv^2 \rightarrow^{M_0-n} \, (n-2)! \, \mf{W}_2^{(\star)} \text{,}\\
            \li^{n-2}\wv^2-(n-2)!\,(\log\rho)\mf{W}_2^{(\star)} \rightarrow^{M_0-n} \, (n-2)! \, \mf{W}_2^{(\dagger)} \text{,}
        \end{cases}
    \end{align}
    where $\mf{W}_i^{(\star)}=0$ if $n$ is odd.
\end{theorem}

\begin{corollary}\label{th.fg_exp_w}
    Fix $n>2$ and $M_0 \geq n+2$. Suppose $\paren{\mi{M},g}$ is an $(n+1)$--dimensional, $M_0$--regular vacuum FG-aAdS segment. Then, with respect to any compact coordinate system on $\mi{I}$,
    \begin{align}
        \wv^0&=
        \begin{cases}
            \sum_{k=0}^{\frac{n-3}{2}} \rho^{2k} \cdot \mf{W}_0^{(2k)}+\rho^{n-2}\cdot\mf{W}_0^{(n-2)}+\rho^{n-2}\cdot\ms{r}_0 \text{,} &\text{$n$ odd,} \\
            \sum_{k=0}^{\frac{n-4}{2}} \rho^{2k} \cdot \mf{W}_0^{(2k)}+\rho^{n-2}\log\rho \cdot \mf{W}_0^{(\star)}+\rho^{n-2}\cdot\mf{W}_0^{(n-2)}+\rho^{n-2}\cdot\ms{r}_0 \text{,} &\text{$n$ even,}
        \end{cases} \label{eq.w0_exp}\\
        \wv^1&=
        \begin{cases}
            \sum_{k=0}^{\frac{n-3}{2}} \rho^{2k+1} \cdot \mf{W}_1^{(2k+1)}+\rho^{n-1}\cdot\mf{W}_1^{(n-1)}+\rho^{n-1}\cdot\ms{r}_1 \text{,} &\text{$n$ odd,} \\
            \sum_{k=0}^{\frac{n-4}{2}} \rho^{2k+1} \cdot \mf{W}_1^{(2k+1)}+\rho^{n-1}\log\rho \cdot \mf{W}_1^{(\star)}+\rho^{n-1}\cdot\mf{W}_1^{(n-1)}+\rho^{n-1}\cdot\ms{r}_1 \text{,} &\text{$n$ even,}
        \end{cases} \label{eq.w1_exp}\\
        \wv^2&=
        \begin{cases}
            \sum_{k=0}^{\frac{n-3}{2}} \rho^{2k} \cdot \mf{W}_2^{(2k)}+\rho^{n-2}\cdot\mf{W}_2^{(n-2)}+\rho^{n-2}\cdot\ms{r}_2 \text{,} &\text{$n$ odd,} \\
            \sum_{k=0}^{\frac{n-4}{2}} \rho^{2k} \cdot \mf{W}_2^{(2k)}+\rho^{n-2}\log\rho \cdot \mf{W}_2^{(\star)}+\rho^{n-2}\cdot\mf{W}_2^{(n-2)}+\rho^{n-2}\cdot\ms{r}_2 \text{,} &\text{$n$ even,}
        \end{cases} \label{eq.w2_exp}
    \end{align}
    where:
    \begin{itemize}
        \item $\mf{W}_0^{(n-2)}$, $\mf{W}_1^{(n-1)}$ and $\mf{W}_2^{(n-2)}$ are $C^{M_0-n}$, $C^{M_0-n-1}$ and $C^{M_0-n}$ $\mi{I}$--tensor fields respectively.
        \item $\ms{r}_0$, $\ms{r}_1$ and $\ms{r}_2$ are vertical tensor fields for which
        \begin{align}
            \ms{r}_0 \to^{M_0-n} 0 \text{,}\qquad \ms{r}_1 \to^{M_0-n-1} 0 \text{,}\qquad \ms{r}_2 \to^{M_0-n} 0 \text{.}
        \end{align}
    \end{itemize}
\end{corollary}

In fact, one can use Proposition \ref{th.w_g_rels} to compute the leading--order terms in each of these boundary expansions:

\begin{proposition} \label{th.weyl_leading_order}
    Fix $n > 2$ and $M_0 \geq n+2$. Suppose $\paren{\mi{M},g}$ is a $(n+1)$--dimensional, $M_0$--regular vacuum FG-aAdS segment. Then
    \begin{align}
        \mf{W}_{0 \, abcd}^{(0)} = \mf{W}_{abcd} \text{,}\qquad \mf{W}_{1 \, abc}^{(1)} = \frac{1}{n-2} \mf{C}_{abc} \text{,}\qquad \mf{W}_{2 \, ab}^{(0)} = 0 \text{,}
    \end{align}
    where $\mf{W}$ and $\mf{C}$ are the $\gm$--Weyl curvature and Cotton tensors respectively, and
    \begin{align}
        &n=3: \qquad \mf{W}_{2 \, ab}^{(1)} = -\frac{3}{2} \gm^{(3)}_{ab} \text{,} \label{eq.w2_1}\\
        &n=4: \qquad \mf{W}_{2 \, ab}^{(\star)} = -4\gm^{(\star)}_{ab} \text{,}\qquad \mf{W}_{2 \, ab}^{(2)} = -3\gm^{(\star)}_{ab} - 4 \gm^{(4)}_{ab} + \gm^{cd} \mf{p}_{ac} \mf{p}_{bd} \text{,} \label{eq.w2_2}\\
        &n>4: \qquad\mf{W}_{2 \, ab}^{(2)} = -4 \gm^{(4)}_{ab} + \gm^{cd} \mf{p}_{ac} \mf{p}_{bd} \text{,} \label{eq.w2_3}\\
        &n>6: \qquad \mf{W}_{2 \, ab}^{(4)} = -12\gm^{(6)}_{ab}-2\gm^{cd}\paren{\mf{p}_{ac}\gm^{(4)}_{bd}+\gm^{(4)}_{ac}\mf{p}_{bd}} + \gm^{ce}\gm^{df}\mf{p}_{ef}\mf{p}_{ac}\mf{p}_{bd} \text{.} \label{eq.w2_4}
    \end{align}
\end{proposition}

\begin{proof}
    Substituting the appropriate expansions into the right-hand sides of \eqref{eq.w0.rel} and \eqref{eq.w1.rel} yield
    \begin{align}
        \mf{W}_{0 \, abcd} ^{(0)}&= \Rm_{abcd}+2\gm_{a[c}\mf{p}_{d]b}-2\gm_{b[c}\mf{p}_{d]a} = \mf{W}_{abcd} \text{,} \\
        \mf{W}_{1 \, abc}^{(1)} &= 2 \, \Dm_{[c}\mf{p}_{b]a} = \frac{1}{n-2} \mf{C}_{abc} \text{,} \nonumber
    \end{align}
    as required. To prove the remaining expressions, we first require the leading-order terms in the boundary expansion of $\gv^{-1}$. These are obtained by taking derivatives of
    \begin{align}
        \gv^{ab} \gv_{bc} = \de^a_c \text{,}
    \end{align}
    and computing boundary limits by applying Theorem \ref{th.metric_limits}. In particular, one finds that \footnote{Additional limits are computed for the case when $\gm$ is conformally flat in Corollary \ref{th.conf_flat_met_exp}.}
    \begin{align}
        \gm^{-1} \rightarrow^{M_0} \, \gm^{-1} \text{,}\qquad  \li\gm^{-1} \rightarrow^{M_0-1} \, 0 \text{,}\qquad \li^2\gm^{-1} \rightarrow^{M_0-2} \, 2 \, \gm^{-1} \gm^{-1} \mf{p} \text{.}
    \end{align}
    Substituting this along with the boundary expansions of $\li\gv$ and $\li^2\gv$ (on a dimensional case-by-case basis) into the right-hand side of \eqref{eq.w2.rel} yields the required expressions.
\end{proof}

The second Bianchi identity for the spacetime Weyl curvature can be decomposed to obtain a series of first-order equations involving only vertical objects:

\begin{proposition} \label{eq.vert_Bianchi_eqns}
   Fix $n>2$ and $M_0 \geq n+2$. Suppose $\paren{\mi{M},g}$ is an $(n+1)$--dimensional, $M_0$--regular vacuum FG-aAdS segment. If $\paren{U , \varphi}$ is a coordinate system on $\mi{I}$ then the following \textbf{vertical Bianchi equations} hold for the vertical Weyl fields with respect to $\varphi_\rho$--coordinates:
   \begin{align}
       \rho \Dvm_\rho \wv^0_{abcd} &= 2\rho \Dv_{[a}\wv^1_{b]cd} -2\gv_{c[a}\wv^2_{b]d}+2\gv_{d[a}\wv^2_{b]c} \label{eq.Bi3}\\
       &\qquad \underbrace{ +\rho \gv^{ef}\li\gv_{e[a}\wv^0_{b]fcd}+\rho\li\gv_{c[a}\wv^2_{b]d}-\rho\li\gv_{d[a}\wv^2_{b]c} }_{\mi{O}_{M_0-2}(\rho^2;\wv^0)} \text{,} \nonumber \\
       \rho^{n-2} \Dvm_\rho \paren{ \rho^{-(n-2)} \wv^1_{abc} } &= - \gv^{de} \Dv_d \wv^0_{eabc} - \underbrace{ \gv^{de} \paren{ \li \gv_{e[b} \wv^1_{c]da} + \li \gv_{d[e} \wv^1_{a]bc} } }_{\mi{O}_{M_0-2}(\rho ; \wv^1)} \text{,} \label{eq.Bi2} \\
       \rho \Dvm_\rho \paren{ \rho^{-1} \wv^1_{abc} } &= 2 \Dv_{[b} \wv^2_{c]a} \underbrace{ - \gv^{de} \paren{ \li \gv_{be} \wv^1_{(a|d|c)} - \li \gv_{ce} \wv^1_{(a|d|b)} } }_{\mi{O}_{M_0-2}(\rho ; \wv^1)} \text{,} \label{eq.Bi1}\\
       \rho^{n-2} \Dvm_\rho \paren{ \rho^{-(n-2)} \wv^2_{bd} } &= - 2 \gv^{ac} \Dv_{[a} \wv^1_{b]cd} \label{eq.Bi4}\\
       &\qquad - \underbrace{ \gv^{ac} \paren{ \gv^{ef} \li \gv_{e[a} \wv^0_{b]fcd} + \li \gv_{c[a} \wv^2_{b]d} - \frac{1}{2} \li \gv_{da} \wv^2_{bc} } }_{\mi{O}_{M_0-2}( \rho ; \wv^0 )} \text{.} \nonumber
   \end{align}
\end{proposition}

\begin{remark}
    One can apply \eqref{eq.rho_deriv} to \eqref{eq.Bi3}--\eqref{eq.Bi4} to express them in terms of $\li$--derivatives:
    \begin{align}
        \rho\li\wv^0_{abcd}&=2\rho\Dv_{[a}\wv^1_{b]cd}-2\gv_{c[a}\wv^2_{b]d}+2\gv_{d[a}\wv^2_{b]c} \label{eq.renorm_bianchi_3}\\
        &\qquad \underbrace{ - \rho \, \gv^{ef}\li\gv_{e[c}\wv^0_{d]fab}+\rho\li\gv_{c[a}\wv^2_{b]d}-\rho\li\gv_{d[a}\wv^2_{b]c} }_{\mi{O}_{M_0-2}(\rho^2;\wv^0)} \text{,} \nonumber\\
        \rho^{n-2} \li \paren{ \rho^{-(n-2)} \wv^1_{abc} } &= -\gv^{de} \Dv_d \wv^0_{eabc} - \frac{1}{2} \gv^{de} \big( \li \gv_{de} \wv^1_{abc} - 2 \li \gv_{da} \wv^1_{ebc} \label{eq.ex_bi_1_a}\\
        &\qquad\qquad\qquad\qquad \underbrace{\qquad\qquad- \li \gv_{bd} \wv^1_{eac} - \li \gv_{cd} \wv^1_{cba} \big)}_{\mi{O}_{M_0-2}(\rho ; \wv^1)} \text{,} \nonumber \\
        \rho \li ( \rho^{-1} \wv^1_{abc} ) &= 2 \Dv_{[b} \wv^2_{c]a} \underbrace{ - \frac{1}{2} \gv^{de} \paren{ \li \gv_{be} \wv^1_{cda} - \li \gv_{ce} \wv^1_{bda} + \li \gv_{ae} \wv^1_{dcb} } }_{\mi{O}_{M_0-2}(\rho ; \wv^1)} \text{,} \label{eq.renorm_bianchi_2}\\
        \rho^{n-2} \li \paren{ \rho^{-(n-2)} \wv^2_{bd} } &= - 2 \gv^{ac} \Dv_{[a} \wv^1_{b]cd} - \underbrace{ \frac{1}{2} \gv^{ac} \gv^{ef} \li \gv_{ea} \wv^0_{bfcd} + \gv^{ef} \li \gv_{e[b} \wv^2_{f]d} }_{\mi{O}_{M_0-2}( \rho ; \wv^0 )} \text{.} \label{eq.ex_bi_2}
    \end{align}
\end{remark}

\begin{proof}
    Consider the following components of the spacetime Bianchi equation \eqref{eq.bianchi}:
    \begin{align}
        \nabla_b W_{\rho c \rho a}+\nabla_\rho W_{cb \rho a}+\nabla_c W_{b \rho\rho a}&=0\text{,} \label{eq.st_bi_1}\\
        \nabla_a W_{\rho bcd}+\nabla_\rho W_{bacd}+\nabla_b W_{a \rho cd}&=0 \text{,}\label{eq.st_bi_3}\\
        \nabla_d W_{eabc}+\nabla_b W_{eacd}+\nabla_c W_{eadb}&=0\text{.}\label{eq.st_bi_2}
    \end{align}
    \begin{itemize}
        \item \eqref{eq.Bi1} follows by application of \eqref{eq.prelim_decomp_rho} and \eqref{eq.prelim_decomp_a} to convert all spacetime quantities in \eqref{eq.st_bi_1} to vertical quantities.
        
        \item \eqref{eq.Bi3} follows by application of \eqref{eq.prelim_decomp_rho} and \eqref{eq.prelim_decomp_a} to \eqref{eq.st_bi_3}.
        
        \item \eqref{eq.Bi4} follows by taking the $\gv$-trace of \eqref{eq.Bi3}.
        
        \item \eqref{eq.Bi2} follows by taking the $\gv$-trace of the expression obtained upon application of \eqref{eq.prelim_decomp_a} to \eqref{eq.st_bi_2}, then applying \eqref{eq.w2_tr_w0} and \eqref{eq.Bi1}. \qedhere
    \end{itemize}
\end{proof}

Similarly, the wave equation satisfied by the spacetime Weyl curvature can be decomposed into a series of wave equations involving only vertical objects:

\begin{proposition} \label{th.vert_wave_eqns}
    Fix $n>2$ and $M_0 \geq n+2$. Suppose $\paren{\mi{M},g}$ is an $(n+1)$--dimensional, $M_0$--regular vacuum FG-aAdS segment. Then
    \begin{align}
        \Boxm_g\hat{\wv}^0 &=\mi{O}_{M_0-2}(\rho^2;\hat{\wv}^0)+\mi{O}_{M_0-3}(\rho^3;\wv^1)+\mi{O}_{M_0-2}(\rho^2;\wv^2) \label{eq.vert_w0} \\
        &\qquad +\mi{O}_{M_0-2}(\rho^3;\Dvm\wv^1) \text{,} \nonumber\\
        \paren{\Boxm_g+(n-1)}\wv^1 &= \mi{O}_{M_0-3}(\rho^3;\hat{\wv}^0) + \mi{O}_{M_0-2}(\rho^2;\wv^1) + \mi{O}_{M_0-3}(\rho^3;\wv^2) \label{eq.vert_w1} \vspace{2mm}\\
        &\qquad+\mi{O}_{M_0-2}(\rho^3;\Dvm\hat{\wv}^0)+\mi{O}_{M_0-2}(\rho^3;\Dvm\wv^2) \text{,} \nonumber\\
        \paren{\Boxm_g+2(n-2)}\wv^2 &= \mi{O}_{M_0-2}(\rho^2;\hat{\wv}^0)+\mi{O}_{M_0-3}(\rho^3;\wv^1)+\mi{O}_{M_0-2}(\rho^2;\wv^2) \label{eq.vert_w2}\\
        &\qquad+\mi{O}_{M_0-2}(\rho^3;\Dvm\wv^1) \text{.}\nonumber
    \end{align}
\end{proposition}

\begin{proof}
    \eqref{eq.stwave} gives
    \begin{gather}
        \rho^2\Box_g W_{\rho a \rho b}=-2n\wv^2_{ab}+\rho^2 \cdot \ms{Q}^2_{ab} \label{eq.LHS2}\\
        \rho^2\Box_g W_{\rho abc}=-2n\wv^1_{ab}+\rho^2 \cdot \ms{Q}^1_{abc} \label{eq.LHS1}\\
        \rho^2\Box_g W_{abcd}=-2n\wv^0_{abcd} + \rho^2 \cdot \ms{Q}^0_{abcd} \label{eq.LHS0}
    \end{gather}
    in which $\ms{Q}^0$, $\ms{Q}^1$, $\ms{Q}^2$ are vertical tensor fields consisting only of terms quadratic in $\hat{\wv}^0$, $\wv^1$, $\wv^2$, the precise form of which are found in Appendix \ref{sec.quadratic_terms}; in particular, one finds
    \begin{align}
        \ms{Q}^2 &= \mc{O}_{M_0-2} \paren{ 1; \hat{\wv}^0 } + \mc{O}_{M_0-3} \paren{ \rho; \wv^1 } + \mc{O}_{M_0-2} \paren{ 1; \wv^2 } \text{,} \\
        \ms{Q}^1 &= \mc{O}_{M_0-3} \paren{ \rho; \hat{\wv}^0 } + \mc{O}_{M_0-2} \paren{ 1; \wv^1 } + \mc{O}_{M_0-3} \paren{ \rho; \wv^2 } \text{,} \\
        \ms{Q}^0 &= \mc{O}_{M_0-2} \paren{ 1; \hat{\wv}^0 } + \mc{O}_{M_0-3} \paren{ \rho; \wv^1 } + \mc{O}_{M_0-2} \paren{ 1; \wv^2 } \text{.}
    \end{align}
    Using \eqref{eq.box} to replace the left-hand side of \eqref{eq.LHS2} with vertical objects, one obtains
    \begin{align} \label{eq.st_vert_box_w2}
        \rho^2\Box_g W_{\rho a \rho b}=&\rho^{-2}\Boxm_g(\rho^2\wv^2_{ab})+2\rho\gv^{cd}\paren{\Dv_c\wv^1_{bda}+\Dv_c\wv^1_{adb}} \\ 
        &\qquad -2(n-1)\wv^2_{ab}+2\gv^{cd}\wv^0_{cadb} \notag\\
        &\qquad +\mi{O}_{M_0-2}(\rho^2;\hat{\wv}^0)+\mi{O}_{M_0-2}(\rho^2;\wv^2) \notag\\
        &\qquad +\mi{O}_{M_0-3}(\rho^3;\wv^1)+\mi{O}_{M_0-2}(\rho^3;\Dvm\wv^1) \text{.}\notag
    \end{align}
    \eqref{eq.prelim_comm_rho} is used to extract the $\rho^2$ factor from the $\Boxm_g$-terms:
    \begin{align}
        \rho^{-2}\Boxm_g(\rho^2\wv^2_{ab})=\Boxm_g\wv^2_{ab}+4\rho\Dvm_\rho\wv^2_{ab}-2(n-2)\wv^2_{ab}+\mi{O}_{M_0-2}(\rho^2;\wv^2)\text{.}
    \end{align}
    Finally, one applies \eqref{eq.w2_tr_w0} and \eqref{eq.Bi4} to \eqref{eq.st_vert_box_w2} in order to deal with the terms involving $\wv^0$ and $\Dvm\wv^1$. This yields \eqref{eq.vert_w2}, as required.
    
    Using \eqref{eq.box} to replace the left-hand side of \eqref{eq.LHS1} with vertical objects, one obtains
    \begin{align} \label{eq.st_vert_box_w1}
        \rho^2\Box_gW_{\rho abc}=&\rho^{2}\Boxm_g(\rho^2\wv^1_{abc})+2\rho\gv^{de}\Dv_d\wv^0_{eabc}-2\rho\paren{\Dv_b\wv^2_{ac}-\Dv_c\wv^2_{ab}} \\
        &\qquad-(n+3)\wv^1_{ab}+2(\wv^1_{abc}+\wv^1_{cab}+\wv^1_{bca})\notag\\
        &\qquad+\mi{O}_{M_0-2}(\rho^2;\wv^1)+\mi{O}_{M_0-3}(\rho^3;\hat{\wv}^0)+\mi{O}_{M_0-3}(\rho^3;\wv^2) \notag\\
        &\qquad+\mi{O}_{M_0-2}(\rho^3;\Dvm\hat{\wv}^0)+\mi{O}_{M_0-2}(\rho^3;\Dvm\wv^2) \text{.}\notag
    \end{align}
    As before, \eqref{eq.prelim_comm_rho} is used to extract the $\rho^2$ factor from the $\Boxm_g$-terms. In \eqref{eq.st_vert_box_w1}, one uses the first Bianchi identity
    \begin{equation}
        \wv^1_{abc}+\wv^1_{cab}+\wv^1_{bca}=0 \text{,}
    \end{equation}
    along with \eqref{eq.Bi1} and \eqref{eq.Bi2} to deal with the terms involving $\Dvm\wv^0$ and $\Dvm\wv^2$. This yields \eqref{eq.vert_w1}, as required.
    
    Using the definition of $\hat{\wv}^0$ and the fact that $\Dvm$ is compatible with $\gv$, 
    \begin{align}
         \Boxm_g\hat{\wv}^0_{abcd}=&\Boxm_g\wv^0_{abcd}+\frac{1}{n-2}\paren{\gv_{ac}\Boxm_g\wv^2_{bd}-\gv_{ad}\Boxm_g\wv^2_{bc}+\gv_{bd}\Boxm_g\wv^2_{ac}-\gv_{bc}\Boxm_g\wv^2_{ad}}\text{.}
    \end{align}
    The equation \eqref{eq.vert_w2} is used to replace the $\Boxm_g\wv^2$ terms on the right hand side with lower-order terms:
    \begin{align}
        \Boxm_g\hat{\wv}^0_{abcd}=&\Boxm_g\wv^0_{abcd}-2\paren{\gv_{ac}\wv^2_{bd}-\gv_{ad}\wv^2_{bc}+\gv_{bd}\wv^2_{ac}-\gv_{bc}\wv^2_{ad}}\\
        &\qquad+\mi{O}_{M_0-2}(\rho^2;\hat{\wv}^0)+\mi{O}_{M_0-3}(\rho^3;\wv^1)+\mi{O}_{M_0-2}(\rho^2;\wv^2)+\mi{O}_{M_0-2}(\rho^3;\Dvm\wv^1) \text{.} \nonumber
    \end{align}
    Next, \eqref{eq.prelim_comm_rho} and \eqref{eq.box} imply that
    \begin{align}
        \Boxm_g\wv^0_{abcd}=&\rho^{-2}\Boxm_g(\rho^2 \wv^0_{abcd})-4\rho\Dvm_\rho \wv^0_{abcd}+2(n-2)\wv^0_{abcd}+\mi{O}_{M_0-2}(\rho^2;\wv^0) \\
        =&\rho^2\Box_gW_{abcd}+2\rho\paren{\Dv_a\wv^1_{bcd}-\Dv_b\wv^1_{acd}+\Dv_c\wv^1_{dab}-\Dv_d\wv^1_{cab}} \nonumber\\
        &\qquad+4\wv^0_{abcd}-2\paren{\gv_{ac}\wv^2_{bd}-\gv_{ad}\wv^2_{bc}+\gv_{bd}\wv^2_{ac}-\gv_{bc}\wv^2_{ad}} \nonumber\\
        &\qquad-4\rho\Dvm_\rho\wv^0_{abcd}+2(n-2)\wv^0_{abcd}+\mi{O}_{M_0-2}(\rho^2;\wv^0) \text{.} \nonumber
    \end{align}
    Finally, one applies \eqref{eq.LHS0} and \eqref{eq.Bi3} to exchange the $\Box_gW$ and $\Dv\wv^1$ terms respectively for terms involving $\wv^0$, $\Dvm_\rho\wv^0$ and $\gv \cdot \wv^2$. This yields \eqref{eq.vert_w0}, as required.
\end{proof}

%%%%%%%%%%%%%%%%%%%%%%%%%%%%%%%%%%%%%%%%%%%%%%%%%%%%%

\section{Time Foliations and the Carleman Estimate} \label{sec.tf_ce}

\begin{definition} \label{def.aads_time}
Let $( \mi{M}, g )$ be a FG-aAdS segment.
A smooth function $t: \mi{I} \rightarrow \mathbb{R}$ is called a \textbf{global time} for $( \mi{M}, g )$ iff
\begin{itemize}
\item The nonempty level sets of $t$ are Cauchy hypersurfaces of $( \mi{I}, \gm )$.

\item There exists some $C > 1$ such that
\footnote{Here $t$ is extended in a $\rho$-independent manner to $\mi{M}$.}
\begin{equation}
\label{eq.aads_time} C^{-1} \leq - \gv ( \Dv^\sharp t, \Dv^\sharp t ) \leq C
\end{equation}
(in other words, $\Dv^\sharp t$ is uniformly timelike).
\end{itemize}
For such a global time we additionally define the shorthands
\footnote{Standard results in Lorentzian geometry \cite{ONeill_1983,Wald_1984} imply $\mi{I}$ is diffeomorphic to $( t_-, t_+ ) \times \mc{S}$ for some $( n - 1 )$-dimensional manifold $\mc{S}$, where $t$ is the projection onto the $( t_-, t_+ )$ component.}
\begin{equation}
\label{eq.aads_timespan} t_+ := \sup_{ \mi{I} } t \text{,} \qquad t_- := \inf_{ \mi{I} } t \text{.}
\end{equation}
\end{definition}

Our results will specifically pertain to spacetimes with the following class of boundaries:

\begin{definition}
    Let $(\mi{M},g)$ be a FG-aAdS segment with a global time $t$. We say that $(\mi{M},g)$ has \textbf{compact cross-sections} if and only if the sets
    \begin{align}
        \mi{I} \cap \{ t=\tau \} \text{,}\qquad t_- < \tau < t_+ \text{,}
    \end{align}
    are compact.
\end{definition}

In order to state our Carleman estimate, we require a coordinate-independent notion of the `size' of a tensor field. To achieve this, we use a global time function to define a corresponding \textit{Riemannian} metric:

\begin{definition} \label{def.aads_riemann}
Let $( \mi{M}, g )$ be a FG-aAdS segment and $t$ a global time. The \textbf{vertical Riemannian metric} associated with $( \gv, t )$ is defined by
\begin{equation}
\label{eq.aads_riemann_vertical} \hv := \gv - \frac{2}{ \gv ( \Dv^\sharp t, \Dv^\sharp t ) } \, dt^2 \text{.}
\end{equation}

Furthermore, given a rank $(k,l)$ vertical tensor field $\ms{A}$, we define its $\hv$--norm $| \ms{A} |_{ \hv }^2$ in terms of coordinates on $\mi{I}$ by
\begin{align}
    | \ms{A} |_{ \hv }^2:=\Pi_{i=1}^{k} \hv_{a_i c_i} \cdot \Pi_{j=1}^{l} \hv^{b_j d_j} \cdot \ms{A}^{a_1 \cdots a_k}_{b_1 \dots b_l} \ms{A}^{c_1 \cdots c_k}_{d_1 \dots d_l} \text{.}
\end{align}
\end{definition}

We now define the weight that will feature in the Carleman estimate.

\begin{definition} \label{def.Cweight}
    Suppose $\paren{\mi{M},g}$ is a FG-aAdS segment and $t$ is a global time. With respect to some $0 \leq b < c$, let
    \begin{align} \label{eq.bigT}
        \mc{T}_{b,c}:=(c^2-b^2)^{-\frac{1}{2}}\underbrace{\tan^{-1}\brak{-b^{-1}(c^2-b^2)^{\frac{1}{2}}}}_{\in [\frac{\pi}{2} , \pi)} \text{,}
    \end{align}
    and suppose $t_0 \in \mathbb{R}$ is such that
    \begin{align}\label{eq.5.5analogue}
        t_0 \pm \mc{T}_{b,c} \in (t_- , t_+) \text{.}
    \end{align}
    We then define the domain
    \begin{align} \label{eq.Omega}
        \Omega := \Omega_{t_0,b,c}:=\{p \in \mi{M} \, | \, \abs{t(p) - t_0} < \mc{T}_{b,c}\} \text{,}
    \end{align}
    and the function $f:=f_{t_0,b,c}: \Omega \to \mathbb{R}$ by
    \begin{align} \label{eq.f}
        f:=\frac{\rho}{\eta(t-t_0)} \text{,}
    \end{align}
    where $\eta:=\eta_{b,c}: \mathbb{R} \to \mathbb{R}$ is given by
    \begin{align} \label{eq.eta}
        \eta(\tau)=e^{-b\abs{\tau}} \cdot \sin\bigg[ \underbrace{\tan^{-1}\paren{-b^{-1}(c^2-b^2)^{\frac{1}{2}}}}_{\in [\frac{\pi}{2} , \pi)} - (c^2-b^2)^{\frac{1}{2}} \cdot \abs{\tau} \bigg] \text{.}
    \end{align}
    Finally, given some $f_\ast > 0$, we define the region
    \begin{align}\label{e.Omega_ast}
        \Omega(f_\ast) := \Omega \cap \{ f < f_\ast \} \text{.}
    \end{align}
\end{definition}

The Carleman estimate of \cite[Theorem 5.11]{McGill_2020} may now be stated in a form suitable for our purposes as follows.

\begin{theorem}\label{th.cest}
    Suppose $\paren{\mi{M},g}$ is a vacuum FG-aAdS segment and $t$ is a global time. Furthermore, suppose
    \begin{itemize}
        \item $\mi{I}$ has compact cross-sections.
        \item The null convexity criterion \eqref{eq.NCC} holds on $\mi{I}$ with associated constants $0 \leq B < C$.
        \item $\Omega$, $f$ are defined with respect to constants $B < b <c <C$ for which \eqref{eq.5.5analogue} holds.
    \end{itemize}
    Fix $\sigma \in \mathbb{R}$. Then, for any integers $k,l \geq 0$, there exist constants $C_0 \geq 0$, $C_1 > 0$ (both depending on $\gv$, $t$, $B$, $C$, $b$, $c$, $k$, $l$) such that for all
    \begin{itemize}
        \item $\kappa \in \mathbb{R}$ satisfying
        \begin{align}\label{eq.kappa_bounds}
            2\kappa \geq n-1+C_0 \text{,}\qquad \kappa^2-(n-2)\kappa+\sigma - (n-1) - C_0 \geq 0 \text{,}
        \end{align}
        \item constants $f_\ast, \lambda, p$ satisfying
        \begin{align}
            0 < f_\ast \ll 1 \text{,}\qquad \lambda \gg \abs{\kappa}+\abs{\sigma}\text{,}\qquad 0 < p < \frac{1}{2},
        \end{align}
        \item rank  $(k,l)$ vertical tensor fields $\ms{u}$ on $\mi{M}$ for which $\ms{u}$ and $\nablam\ms{u}$ vanish on $\Omega \cap \{ f=f_\ast \}$,
    \end{itemize}
    the following Carleman estimate holds:
    \begin{align} \label{eq.Cests}
        \int_{\Omega(f_\ast)} &e^{-\lambda p^{-1} f^p} \cdot f^{n-2-p-2\kappa} \cdot \abs{(\Boxm+\sigma)\ms{u}}_\hv^2 dg \\
        &+C_1 \lambda^3 \limsup_{\rho' \searrow 0} \int_{\Omega(f_\ast) \cap \{ \rho = \rho' \}} \abs{\Dvm_\rho \paren{\rho^{-\kappa} \ms{u}}}_\hv^2 + \abs{\Dv_{\Dv^\sharp t} \paren{\rho^{-\kappa}\ms{u}}}_\hv^2 + \abs{\rho^{-\kappa-1} \ms{u}}_\hv^2 d\gv |_{\rho'} \notag\\
        & \qquad \geq \lambda \int_{\Omega(f_\ast)} e^{-\lambda p^{-1} f^p} \cdot f^{n-2-2\kappa} \cdot \paren{f\rho^3 \abs{\Dvm_\rho \ms{u}}_\hv^2 + f\rho^3 \abs{\Dv \ms{u}}_\hv^2 + f^{2p} \abs{\ms{u}}_\hv^2} dg \text{.} \notag
    \end{align}
\end{theorem}

\begin{proof}
See \cite[Theorem 5.11]{McGill_2020}. Note that our spacetime is automatically a \textit{strongly} FG-aAdS segment (in the terminology used therein) due to our regularity assumptions and the fact that \eqref{eq.eins} is satisfied.
\end{proof}

%%%%%%%%%%%%%%%%%%%%%%%%%%%%%%%%%%%%%%%%%%%%%%%%%%%%%%%

\section{Local AdS Rigidity via the Boundary Data} \label{sec.results2}

\subsection{Preliminary Results}

We begin by applying our Carleman estimate to prove unique continuation for the Weyl curvature given sufficiently fast vanishing along a sufficiently long timespan on the boundary. 

\begin{lemma} \label{th.UC_van_rates}
    Fix $n>2$ and $M_0 \geq n+2$. Suppose $\paren{\mi{M},g}$ is an $(n+1)$--dimensional, $M_0$--regular vacuum FG-aAdS segment. Furthermore, suppose $\paren{\mi{M},g}$ has a global time $t$ for which $\mi{I}$ has compact cross-sections and the null convexity criterion holds on $\mi{I}$ with associated constants $0 \leq B <C$. Fix constants $B<b<c<C$ and $t_0 \in \mathbb{R}$ such that \eqref{eq.5.5analogue} holds. Let $C_0$ be the corresponding constant featuring in Theorem \ref{th.cest} and define
    \begin{align} \label{eq.kappa_vanishing}
        \kappa_i:=\max \left\{ \frac{1}{2}\brak{n-1+C_0} \text{,}\,\, \frac{1}{2}\brak{n-2+\sqrt{n^2-4i(n-i)+4C_0}} \right\} \text{,}\qquad i=0,1,2 \text{.}
    \end{align}
    Suppose that, on $\mi{I} \cap \Omega$,
    \begin{align} 
        \hat{\wv}^0 &= \mc{O}_{1}\paren{\rho^{\kappa_0+2}}\text{,}\qquad \Dvm_\rho \hat{\wv}^0=\mc{O}_{0}\paren{\rho^{\kappa_0+1}}\text{,} \label{eq.tilde_van_cond}\\
        \wv^1 &= \mc{O}_{1}\paren{\rho^{\kappa_1+2}}\text{,}\qquad \Dvm_\rho \wv^1 = \mc{O}_{0}\paren{\rho^{\kappa_1+1}}\text{,} \label{eq.tilde_van_cond_1}\\
        \wv^2 &= \mc{O}_{1}\paren{\rho^{\kappa_2+2}}\text{,}\qquad \Dvm_\rho \wv^2 = \mc{O}_{0}\paren{\rho^{\kappa_2+1}}\text{.} \label{eq.tilde_van_cond_2}
    \end{align}
    Then the spacetime Weyl curvature identically vanishes on $\Omega \cap \{f < \frac{1}{2} f_\ast \}$. \footnote{$f_\ast$ is the corresponding constant in Theorem \ref{th.cest}.}
\end{lemma}

\begin{remark}\label{rem.kappa}
    The following bounds hold for each of the above-defined $\kappa_i$:
    \begin{align}
        \kappa_0 \geq n-1 \text{,}\qquad \kappa_1 \geq n-2 \text{,}\qquad \kappa_2 \geq n-3 \text{.}
    \end{align}
\end{remark}

\begin{proof}
    Take some smooth cutoff function
    \begin{align}
        \bar{\chi}:[0,f_\ast]\to[0,1] \text{,}\qquad
        \bar{\chi}(s)=
        \begin{cases}
            1\text{,}\qquad 0\leq s\leq\frac{1}{2}f_\ast \text{,}\\
            0\text{,}\qquad \frac{3}{4}f_\ast \leq s \text{.}
        \end{cases} \nonumber
    \end{align}
    Let $\chi:=\bar{\chi}\circ f$ and, for convenience, define
    \begin{align}
        \ms{w}^0 := \hat{\wv}^0 \text{,}\qquad
        \ms{w}^1 := \wv^1 \text{,}\qquad
        \ms{w}^2 := \wv^2 \text{.}
    \end{align}
    Let $\sigma_i := i(n-i)$ for $i=0,1,2$. Then
    \begin{align}
        ( \Boxm_g + \sigma_i ) ( \chi \cdot \ms{w}^i ) =& \chi ( \Boxm_g + \sigma_i ) \ms{w}^i + \chi'' \, \Dvm^\alpha f \Dvm_\alpha f \, \ms{w}^i \\
        &\qquad +\chi'\paren{2 \Dvm_\alpha f \Dvm^\alpha \ms{w}^i + \Boxm_g f \cdot \ms{w}^i} \text{,} \nonumber
    \end{align}
    where $\chi'$ denotes the derivative of $\chi$ with respect to $f$. Using that $\chi',\chi''$ are supported in $[ \frac{1}{2} f_\ast , \frac{3}{4} f_\ast ]$ and applying the relations derived in \cite[Proposition 5.7]{McGill_2020}, one finds
    \begin{align}
        \sum_{i=0}^2 \abs{(\Boxm_g+\sigma_i)(\chi\cdot\ms{w}^i)}\lesssim
        \begin{cases}
            \sum_{i=0}^2\left(\abs{(\Boxm_g+\sigma_i)\ms{w}^i}_\hv^2+\rho^2f^2\abs{\Dvm_\rho \ms{w}^i}_\hv^2 \right. \\
            \qquad\qquad\qquad\left.+\rho^2f^4\abs{\Dv\ms{w}^i}_\hv^2+f^2\abs{\ms{w}^i}_\hv^2\right)\text{,}\qquad \frac{1}{2}f_\ast \leq f \leq \frac{3}{4}f_\ast \\
            \sum_{i=0}^2\abs{(\Boxm_g+\sigma_i)\ms{w}^i}_\hv^2\text{,}\qquad\qquad\qquad\qquad\qquad\quad 0 \leq f \leq \frac{1}{2}f_\ast \text{.}
        \end{cases}
    \end{align}
    \eqref{eq.vert_w0}, \eqref{eq.vert_w1} and \eqref{eq.vert_w2} in addition to the fact that $f \simeq 1$ in $[ \frac{1}{2} f_\ast , \frac{3}{4} f_\ast ]$ implies that
    \begin{align}
        \sum_{i=0}^2 \abs{(\Boxm_g+\sigma_i)(\chi\cdot\ms{w}^i)}\lesssim
        \begin{cases}
            \sum_{i=0}^2\paren{\rho^2\abs{\Dvm_\rho \ms{w}^i}_\hv^2+\rho^2\abs{\Dv\ms{w}^i}_\hv^2+\abs{\ms{w}^i}_\hv^2}\text{,}\qquad \frac{1}{2}f_\ast \leq f \leq \frac{3}{4}f_\ast \\
            \sum_{i=0}^2\paren{\rho^6\abs{\Dv \ms{w}^i}_\hv^2+\rho^4\abs{\ms{w}^i}_\hv^2}\text{,}\qquad\qquad\qquad\qquad 0 \leq f \leq \frac{1}{2}f_\ast \text{.}
        \end{cases}
    \end{align}
    We define regions
    \begin{align}
        \Omega_i:=&\Omega(f_\ast) \cap \{ f < \frac{1}{2} f_\ast \} \\
        \Omega_e:=&\Omega(f_\ast) \cap \{ \frac{1}{2} f_\ast < f < \frac{3}{4} f_\ast \}
    \end{align}
    and sum the Carleman estimates \eqref{eq.Cests}, as applied to $\bar{\ms{w}^i}:=\chi\cdot\ms{w}^i$. The left hand side $L$ of the sum of these Carleman estimates can be estimated by
    \begin{align}
        L \lesssim& \sum_{i=0}^2 \int_{\Omega_e} e^{-\lambda p^{-1} f^p}f^{n-2-2\kappa_i-p}\brak{\rho^2\abs{\Dvm_\rho \ms{w}^i}_\hv^2+\rho^2\abs{\Dv\ms{w}^i}_\hv^2+\abs{\ms{w}^i}_\hv^2} dg \\
        &\qquad+\sum_{i=0}^2 \int_{\Omega_i} e^{-\lambda p^{-1} f^p}f^{n-2-2\kappa_i-p}\brak{\rho^6\abs{\Dv \ms{w}^i}_\hv^2+\rho^4\abs{\ms{w}^i}_\hv^2} dg \nonumber\\
        &\qquad+\mc{C}\lambda^3\sum_{i=0}^2 \limsup_{ \rho_\ast \searrow 0 } \int_{ \Omega_{ t_0 } ( f_\ast ) \cap \{ \rho = \rho_\ast \} } [ | \Dvm_\rho ( \rho^{ - \kappa_i } \bar{\ms{w}}^i ) |_\hv^2 + | \Dv ( \rho^{ - \kappa_i } \bar{\ms{w}}^i ) |_\hv^2 + | \rho^{ - \kappa_i - 1 } \bar{\ms{w}}^i |_\hv^2 ] \, d \gv |_{ \rho_\ast } \nonumber\\
        :=& L_1 + L_2 + L_3 \text{.} \nonumber
    \end{align}
    $L$ is bounded below by
    \begin{align}
        L \gtrsim \lambda \sum_{i=0}^2 \int_{ \Omega_i } e^{ - \lambda p^{-1} f^p } f^{ n - 2 - 2 \kappa_i } ( f \rho^3 | \Dvm_\rho \ms{w}^i |_\hv^2 + f \rho^3 | \Dv \ms{w}^i |_\hv^2 + f^{ 2 p } | \ms{w}^i |_\hv^2 ) \, dg
    \end{align}
    $L_2$ can be absorbed by taking $\la$ sufficiently large. By \eqref{eq.kappa_bounds}, the vanishing assumptions \eqref{eq.tilde_van_cond} and the fact that $|\partial_\rho \chi|+|\partial_a \chi| \lesssim \rho^{-1}$, we also have that $L_3 \to 0$ as $\rho_\ast \searrow 0$.

    The $e^{-\lambda p^{-1} f^p}f^{n-2-2\kappa_i}$ factors can be bounded above in $\Omega_e$ and bounded below in $\Omega_i$:
    \begin{align}
        e^{-\lambda p^{-1}f^p}f^{n-2-2\kappa_i}
        \begin{cases}
            \leq e^{-\lambda p^{-1}(\frac{f_\ast}{2})^p} \paren{\frac{f_\ast}{2}}^{n-2-2\kappa_i} \text{,}\qquad \text{in $\Omega_e$} \\
            \geq e^{-\lambda p^{-1}(\frac{f_\ast}{2})^p} \paren{\frac{f_\ast}{2}}^{n-2-2\kappa_i} \text{,}\qquad \text{in $\Omega_i$.}
        \end{cases}
    \end{align}
    Hence, for large $\lambda$,
    \begin{align} \label{eq.uceqn}
        \sum_{i=0}^2 \int_{\Omega_e} \abs{\ms{w}^i}_\hv^2 + \rho^2 \abs{\Dvm_\rho \ms{w}^i}_\hv^2 + \rho^2\abs{\Dv \ms{w}^i}_\hv^2 dg \gtrsim \lambda \sum_{i=0}^2 \int_{\Omega_i} f^{2p} \abs{\ms{w}^i}_\hv^2 dg
    \end{align}
    The left hand side of \eqref{eq.uceqn} is bounded above by
    \begin{align}
        \lesssim \sum_{i=0}^2 \int_{\Omega_e} \abs{\rho^{-\kappa_i-1}\ms{w}^i}_\hv^2 + \abs{\Dvm_\rho(\rho^{-\kappa_i} \ms{w}^i)}_\hv^2 + \abs{\Dv(\rho^{-\kappa_i} \ms{w}^i)}_\hv^2 dg \text{,}
    \end{align}
    and so is finite by the vanishing assumptions \eqref{eq.tilde_van_cond}. As a result, taking $\lambda \to \infty$ in \eqref{eq.uceqn} yields $\ms{w}^0, \ms{w}^1, \ms{w}^2 \equiv 0$ (i.e. $\hat{\wv}^0, \wv^1, \wv^2 \equiv 0$) on $\Omega \cap \{f < \frac{1}{2} f_\ast \}$. In other words, the full spacetime Weyl curvature $W \equiv 0$ on $\Omega \cap \{f < \frac{1}{2} f_\ast \}$ as required.
\end{proof}

\begin{lemma} \label{th.conf_flat_implies_vanishing_star}
    Fix $n>2$ and $M_0 \geq n+2$. Suppose $\paren{\mi{M},g}$ is an $(n+1)$--dimensional, $M_0$--regular vacuum FG-aAdS segment. If $\gm$ is conformally flat then
    \begin{align}
        \wv^0 &= \rho^{n-2} \cdot \mf{W}_{0}^{(n-2)} + \rho^{n-2} \cdot \ms{r}_{0}\text{,}\qquad \ms{r}_0 \to^{M_0-n} 0 \text{,}\label{eq.improved_w0_exp}\\
        \wv^1 &= \rho^{n-1} \cdot \mf{W}_{1}^{(n-1)} + \rho^{n-1} \cdot \ms{r}_{1}\text{,}\qquad \ms{r}_1 \to^{M_0-n-1} 0 \text{,}\label{eq.improved_w1_exp}\\
        \wv^2 &= \rho^{n-2} \cdot \mf{W}_{2}^{(n-2)} + \rho^{n-2} \cdot \ms{r}_{2}\text{,}\qquad \ms{r}_2 \to^{M_0-n} 0 \text{.} \label{eq.improved_w2_exp}
    \end{align}
\end{lemma}

\begin{proof}
    Recall the leading-order expressions for $\wv^0$, $\wv^1$ and $\wv^2$ in Proposition \ref{th.weyl_leading_order}. If $n=3$ then $\mf{W}$ identically vanishes and conformal flatness of $\gm$ implies that $\mf{C}$ vanishes. Corollary \ref{th.fg_exp_w} then yields \eqref{eq.improved_w0_exp}--\eqref{eq.improved_w2_exp} as required.
    
    If $n > 3$ then conformal flatness of $\gm$ implies that $\mf{W}$ vanishes, which in turn implies that $\mf{C}$ vanishes. Hence
    \begin{align} \label{eq.in_step}
        \mf{W}_{0}^{(0)} = \mf{W}_{1}^{(1)} = \mf{W}_{2}^{(0)} = 0 \text{.}
    \end{align}
    
    If $n=4$ then this implies
    \begin{align}
        \wv^0 &= \rho^{2} \log \rho \cdot \mf{W}_{0}^{(\star)} + \rho^{2} \cdot \mf{W}_{0}^{(2)} + \rho^{2} \cdot \ms{r}_{0}\text{,} \label{eq.almost_0_4}\\
        \wv^1 &= \rho^{3} \log \rho \cdot \mf{W}_{1}^{(\star)} + \rho^{3} \cdot \mf{W}_{1}^{(3)} + \rho^{3} \cdot \ms{r}_{1}\text{,} \label{eq.almost_1_4}\\
        \wv^2 &= \rho^{2} \log \rho \cdot \mf{W}_{2}^{(\star)} + \rho^{2} \cdot \mf{W}_{2}^{(2)} + \rho^{2} \cdot \ms{r}_{2}\text{.} \label{eq.almost_2_4}
    \end{align}
    Given \eqref{eq.almost_2_4}, the left-hand side of the vertical Bianchi equation \eqref{eq.ex_bi_2} reads
    \begin{align}
        \rho^2 \li \paren{ \rho^{-2} \cdot \wv^2 } = \rho \cdot \mf{W}_2^{(\star)} + o(\rho) \text{,}
    \end{align}
    where $o(\rho)$ denotes a vertical tensor field $\ms{t}$ for which $\rho^{-1} \cdot \ms{t} \rightarrow^0 0$. The factors of $\gv^{-1}$ and $\li \gv$ present in the right-hand side of \eqref{eq.ex_bi_2} are $\mi{O}(1)$ and $\mi{O}(\rho)$ respectively; given \eqref{eq.almost_0_4} and \eqref{eq.almost_1_4}, the right-hand side thus only contains $\mi{O}(\rho^3 \log\rho)$ terms. Collecting strictly order $\rho$ terms in \eqref{eq.ex_bi_2}, one finds
    \begin{align}
        \mf{W}_{2}^{(\star)} = 0 \text{,}
    \end{align}
    which yields \eqref{eq.improved_w2_exp} as required.
    
    Given \eqref{eq.almost_0_4}, the left-hand side of \eqref{eq.renorm_bianchi_3} reads
    \begin{align}
        \rho \li \wv^0 = \rho^2 ( 1 + 2 \log \rho ) \cdot \mf{W}_0^{(\star)} + 2 \rho^2 \cdot \mf{W}_0^{(2)} + o(\rho^2) \text{.} \nonumber
    \end{align}
    As above, the factors of $\gv^{-1}$ and $\li \gv$ present in the right-hand side of \eqref{eq.renorm_bianchi_3} are $\mi{O}(1)$ and $\mi{O}(\rho)$ respectively; given \eqref{eq.almost_0_4} and \eqref{eq.improved_w2_exp}, the right-hand side thus only contains $\mi{O}(\rho^2)$ terms. Collecting strictly order $\rho^2 \log \rho$ terms in \eqref{eq.renorm_bianchi_3}, one finds
    \begin{align}
        \mf{W}_{0}^{(\star)} = 0 \text{,}
    \end{align}
    which yields \eqref{eq.improved_w0_exp} as required. Note that is was essential for us to derive \eqref{eq.improved_w2_exp} \textit{before} completing this step; otherwise, the right-hand side of \eqref{eq.renorm_bianchi_3} would still contain strictly order $\rho^2 \log \rho$ terms.
    
    Finally, given \eqref{eq.almost_1_4}, the left-hand side of \eqref{eq.renorm_bianchi_2} reads
    \begin{align}
        \rho \li \paren{ \rho^{-1} \cdot \wv^1 } = \rho^2 ( 1 + 2 \log \rho ) \cdot \mf{W}_1^{(\star)} + 2 \rho^2 \cdot \mf{W}_1^{(3)} + o(\rho^2) \text{.} \notag
    \end{align}
    Once more, the factors of $\gv^{-1}$ and $\li \gv$ present in the right-hand side of \eqref{eq.renorm_bianchi_2} are $\mi{O}(1)$ and $\mi{O}(\rho)$ respectively. Given \eqref{eq.almost_1_4} and \eqref{eq.improved_w2_exp}, the right-hand side thus only contains $\mi{O}(\rho^2)$ terms; collecting strictly order $\rho^2 \log \rho$ terms in \eqref{eq.renorm_bianchi_2}, one finds
    \begin{align}
        \mf{W}_{1}^{(\star)} = 0 \text{,}
    \end{align}
    which yields \eqref{eq.improved_w1_exp} as required. Again, it was crucial to derive \eqref{eq.improved_w2_exp} before completing this step to remove all order $\rho^2 \log \rho$ terms from the right-hand side of \eqref{eq.renorm_bianchi_3}. This completes the proof of the $n=4$ case.
    
    For $n>4$, \eqref{eq.in_step} yields
    \begin{align}
        \wv^0&=
        \begin{cases}
            \sum_{k=1}^{\frac{n-3}{2}} \rho^{2k} \cdot \mf{W}_0^{(2k)}+\rho^{n-2}\cdot\mf{W}_0^{(n-2)}+\rho^{n-2}\cdot\ms{r}_0 \text{,} &\text{$n$ odd,} \\
            \sum_{k=1}^{\frac{n-4}{2}} \rho^{2k} \cdot \mf{W}_0^{(2k)}+\rho^{n-2}\log\rho \cdot \mf{W}_0^{(\star)}+\rho^{n-2}\cdot\mf{W}_0^{(n-2)}+\rho^{n-2}\cdot\ms{r}_0 \text{,} &\text{$n$ even,}
        \end{cases} \label{eq.almost_0_n}\\
        \wv^1&=
        \begin{cases}
            \sum_{k=1}^{\frac{n-3}{2}} \rho^{2k+1} \cdot \mf{W}_1^{(2k+1)}+\rho^{n-1}\cdot\mf{W}_1^{(n-1)}+\rho^{n-1}\cdot\ms{r}_1 \text{,} &\text{$n$ odd,} \\
            \sum_{k=1}^{\frac{n-4}{2}} \rho^{2k+1} \cdot \mf{W}_1^{(2k+1)}+\rho^{n-1}\log\rho \cdot \mf{W}_1^{(\star)}+\rho^{n-1}\cdot\mf{W}_1^{(n-1)}+\rho^{n-1}\cdot\ms{r}_1 \text{,} &\text{$n$ even,}
        \end{cases} \label{eq.almost_1_n}\\
        \wv^2&=
        \begin{cases}
            \sum_{k=1}^{\frac{n-3}{2}} \rho^{2k} \cdot \mf{W}_2^{(2k)}+\rho^{n-2}\cdot\mf{W}_2^{(n-2)}+\rho^{n-2}\cdot\ms{r}_2 \text{,} &\text{$n$ odd,} \\
            \sum_{k=1}^{\frac{n-4}{2}} \rho^{2k} \cdot \mf{W}_2^{(2k)}+\rho^{n-2}\log\rho \cdot \mf{W}_2^{(\star)}+\rho^{n-2}\cdot\mf{W}_2^{(n-2)}+\rho^{n-2}\cdot\ms{r}_2 \text{,} &\text{$n$ even.}
        \end{cases} \label{eq.almost_2_n}
    \end{align}
    Given \eqref{eq.almost_2_n}, the left-hand side of the vertical Bianchi equation \eqref{eq.ex_bi_2} reads
    \begin{align}
        \begin{cases}
            \sum_{k=1}^{\frac{n-3}{2}} (2k+2-n) \rho^{2k-1} \cdot \mf{W}_2^{(2k)}+ o(\rho^{n-2}) \text{,} &\text{$n$ odd,} \\
            \sum_{k=1}^{\frac{n-4}{2}} (2k+2-n) \rho^{2k-1} \cdot \mf{W}_2^{(2k)} + \rho^{n-3} \cdot \mf{W}_2^{(\star)} + o(\rho^{n-2}) \text{,} &\text{$n$ even.}
        \end{cases}
    \end{align}
    Given \eqref{eq.almost_0_n} and \eqref{eq.almost_1_n}, the right-hand side thus only contains $\mi{O}(\rho^3)$ terms. Collecting strictly order $\rho$ terms in \eqref{eq.ex_bi_2}, one finds
    \begin{align}
        \mf{W}_{2}^{(2)} = 0 \text{.}
    \end{align}
    In other words, for $n=5$ we have \eqref{eq.improved_w2_exp} as required and, for $n>5$,
    \begin{align}
        \wv^2=
        \begin{cases}
            \rho^4 \log \rho \cdot \mf{W}_2^{(\star)} + \rho^4 \cdot \mf{W}_2^{(4)} + \rho^4 \cdot \ms{r}_2 \text{,} &\text{$n=6$,}\\
            \sum_{k=2}^{\frac{n-3}{2}} \rho^{2k} \cdot \mf{W}_2^{(2k)}+\rho^{n-2}\cdot\mf{W}_2^{(n-2)}+\rho^{n-2}\cdot\ms{r}_2 \text{,} &\text{$n>6$ odd,} \\
            \sum_{k=2}^{\frac{n-4}{2}} \rho^{2k} \cdot \mf{W}_2^{(2k)}+\rho^{n-2}\log\rho \cdot \mf{W}_2^{(\star)}+\rho^{n-2}\cdot\mf{W}_2^{(n-2)}+\rho^{n-2}\cdot\ms{r}_2 \text{,} &\text{$n>6$ even.}
        \end{cases} \label{eq.almost_2_n_a}
    \end{align}
    Given \eqref{eq.almost_0_n}, the left-hand side of \eqref{eq.renorm_bianchi_3} reads
    \begin{align}
        \begin{cases}
            \sum_{k=1}^{\frac{n-3}{2}} (2k) \rho^{2k} \cdot \mf{W}_0^{(2k)} + (n-2) \rho^{n-2} \cdot \mf{W}_0^{(n-2)} + o(\rho^{n-2}) \text{,} \,\, &\text{$n$ odd,} \\
            \sum_{k=1}^{\frac{n-4}{2}} (2k) \rho^{2k} \cdot \mf{W}_0^{(2k)} + \rho^{n-2} \brak{ 1+(n-2)\log\rho } \cdot \mf{W}_0^{(\star)} \\
            \qquad + (n-2) \rho^{n-2} \cdot \mf{W}_0^{(n-2)} + o(\rho^{n-2}) \text{,} &\text{$n$ even.}
        \end{cases}
    \end{align}
    Given \eqref{eq.almost_0_n}, \eqref{eq.almost_1_n} and \eqref{eq.almost_2_n_a} (or \eqref{eq.improved_w2_exp} if $n=5$), the right-hand side thus only contains
    \begin{align}
        \begin{cases}
            \mi{O}(\rho^3) \,\, \text{terms if $n=5$,} \\
            \mi{O}(\rho^4 \log \rho) \,\, \text{terms if $n=6$,} \\
            \mi{O}(\rho^4) \,\, \text{terms if $n>6$.} \\
        \end{cases}
    \end{align}
    Collecting strictly order $\rho^2$ terms in \eqref{eq.renorm_bianchi_3}, one hence finds
    \begin{align}
        \mf{W}_{0}^{(2)} = 0 \text{.}
    \end{align}
    In other words, for $n=5$, we have \eqref{eq.improved_w0_exp} as required and, for $n>5$,
    \begin{align}
        \wv^0=
        \begin{cases}
            \rho^4 \log \rho \cdot \mf{W}_0^{(\star)} + \rho^4 \cdot \mf{W}_0^{(4)} + \rho^4 \cdot \ms{r}_0 \text{,} &\text{$n=6$,}\\
            \sum_{k=2}^{\frac{n-3}{2}} \rho^{2k} \cdot \mf{W}_0^{(2k)}+\rho^{n-2}\cdot\mf{W}_0^{(n-2)}+\rho^{n-2}\cdot\ms{r}_0 \text{,} &\text{$n>6$ odd,} \\
            \sum_{k=2}^{\frac{n-4}{2}} \rho^{2k} \cdot \mf{W}_0^{(2k)}+\rho^{n-2}\log\rho \cdot \mf{W}_0^{(\star)}+\rho^{n-2}\cdot\mf{W}_0^{(n-2)}+\rho^{n-2}\cdot\ms{r}_0 \text{,} &\text{$n>6$ even.}
        \end{cases} \label{eq.almost_0_n_a}
    \end{align}
    Given \eqref{eq.almost_1_n}, the left-hand side of \eqref{eq.renorm_bianchi_2} reads
    \begin{align}
        \begin{cases}
            \sum_{k=1}^{\frac{n-3}{2}} (2k) \rho^{2k} \cdot \mf{W}_1^{(2k+1)} + (n-2) \rho^{n-2} \cdot \mf{W}_1^{(n-1)} + o(\rho^{n-2}) \text{,} \,\, &\text{$n$ odd,} \\
            \sum_{k=1}^{\frac{n-4}{2}} (2k) \rho^{2k} \cdot \mf{W}_1^{(2k+1)} + \rho^{n-2} \brak{ 1+(n-2)\log\rho } \cdot \mf{W}_1^{(\star)} \\
            \qquad + (n-2) \rho^{n-2} \cdot \mf{W}_1^{(n-1)} + o(\rho^{n-2}) \text{,} &\text{$n$ even.}
        \end{cases}
    \end{align}
    Given \eqref{eq.almost_1_n} and \eqref{eq.almost_2_n_a} (or \eqref{eq.improved_w2_exp} if $n=5$), the right-hand side thus only contains
    \begin{align}
        \begin{cases}
            \mi{O}(\rho^3) \,\, \text{terms if $n=5$,} \\
            \mi{O}(\rho^4 \log \rho) \,\, \text{terms if $n=6$,} \\
            \mi{O}(\rho^4) \,\, \text{terms if $n>6$.} \\
        \end{cases}
    \end{align}
    Collecting strictly order $\rho^2$ terms in \eqref{eq.renorm_bianchi_2}, one hence finds
    \begin{align}
        \mf{W}_{1}^{(3)} = 0 \text{.}
    \end{align}
    In other words, for $n=5$, we have \eqref{eq.improved_w1_exp} as required and, for $n>5$,
    \begin{align}
        \wv^1=
        \begin{cases}
            \rho^5 \log \rho \cdot \mf{W}_1^{(\star)} + \rho^5 \cdot \mf{W}_1^{(5)} + \rho^5 \cdot \ms{r}_1 \text{,} &\text{$n=6$,} \\
            \sum_{k=2}^{\frac{n-3}{2}} \rho^{2k+1} \cdot \mf{W}_1^{(2k+1)}+\rho^{n-1}\cdot\mf{W}_1^{(n-1)}+\rho^{n-1}\cdot\ms{r}_1 \text{,} &\text{$n>6$ odd,} \\
            \sum_{k=2}^{\frac{n-4}{2}} \rho^{2k+1} \cdot \mf{W}_1^{(2k+1)}+\rho^{n-1}\log\rho \cdot \mf{W}_1^{(\star)}+\rho^{n-1}\cdot\mf{W}_1^{(n-1)}+\rho^{n-1}\cdot\ms{r}_1 \text{,} &\text{$n>6$ even.}
        \end{cases} \label{eq.almost_1_n_a}
    \end{align}
    This completes the proof for the $n=5$ case. For $n>5$, one may now iterate this process to derive vanishing of successively higher-order coefficients; at the $k\textsuperscript{th}$ iteration, substitute the updated expansions of $\wv^{0}$, $\wv^{1}$ and $\wv^{2}$ into \eqref{eq.ex_bi_2}, then \eqref{eq.renorm_bianchi_3} and \eqref{eq.renorm_bianchi_2}, collecting order $\rho^{2k+1}$, $\rho^{2k+2}$ and $\rho^{2k+2}$ terms respectively (or order $\rho^{2k+1}$, $\rho^{2k+2}\log\rho$ and $\rho^{2k+1}\log\rho$ if dealing with $\log$ coefficients).
    
    Note that this process cannot be continued to derive vanishing of $\mf{W}_{2}^{(n-2)}$ (and hence $\mf{W}_{0}^{(n-2)}$, $\mf{W}_{1}^{(n-1)}$ too), since this coefficient is always eliminated upon substitution of \eqref{eq.improved_w2_exp} into \eqref{eq.ex_bi_2}.
\end{proof}

\begin{corollary} \label{th.conf_flat_met_exp}
    Fix $n>2$ and $M_0 \geq n+2$. Suppose $\paren{\mi{M},g}$ is an $(n+1)$--dimensional, $M_0$--regular vacuum FG-aAdS segment. If $\gm$ is conformally flat then
    \begin{align} \label{eq.conf_flat_metric_exp}
        \gv=
            \begin{cases}
                \gm - \rho^2 \, \mf{p} + \rho^n \, \gm^{(n)} + \rho^n \, \ms{r} \text{,}\qquad &n \leq 4 \text{,}\\
                \gm - \rho^2 \, \mf{p} + \frac{1}{4} \rho^4 \, \gm^{-1} \, \mf{p} \, \mf{p} + \rho^n \, \gm^{(n)} + \rho^n \, \ms{r} \text{,}\qquad &n>4 \text{,}
            \end{cases}
    \end{align}
    where the remainder $\ms{r}$ satisfies
    \begin{align}
        \ms{r} \to^{M_0-n} \,\, 0 \text{.}
    \end{align}
    Furthermore, there exist tensor fields $\gm_{(k)}$ on $\mi{I}$ such that
    \begin{align}
        \li^k\gv^{-1} \rightarrow^{M_0-k} \, k! \, \gm_{(k)} \text{,}\qquad 0 \leq k < n \text{,}
    \end{align}
    where $\gm_{(k)}$ vanishes for $k$ odd and, for $0 < 2l < n$,
    \begin{align} 
        \mf{g}_{(2l)}^{ab} &=\frac{l+1}{2^l} \cdot \paren{ \gm^{a c_1} \gm^{c_2 c_3} \dots \gm^{c_{2l} b} } \cdot \paren{ \mf{p}_{c_1 c_2} \mf{p}_{c_3 c_4} \dots \mf{p}_{c_{2l-1} c_{2l}} } \label{eq.inv_metric_comps} \text{,} \\
        &= \frac{l+1}{2l} \, \gm^{ac} \, \mf{p}_{cd} \, \gm^{db}_{(2l-2)} \text{.} \label{eq.inv_met_rec_rel}
    \end{align}
\end{corollary}

\begin{proof}
    Theorem \ref{th.metric_expansion_for_vac_fgaads} gave the near--boundary expansion of $\gv$ and established, furthermore, that $\gm^{(2)}=-\mf{p}$. To prove \eqref{eq.conf_flat_metric_exp}, we must therefore show that
    \begin{enumerate}
        \item When $n>4$, $\gm^{(4)}=\frac{1}{4} \, \gm^{-1} \, \mf{p} \, \mf{p}$.
        \item When $n>6$, $\gm^{(k)}=0$ for $6 \leq k < n$.
        \item When $n \geq 4$, $\gm^{(\star)}=0$.
    \end{enumerate}
    We address each of these statements in turn:
    
    (1) This follows from \eqref{eq.w2_3} and \eqref{eq.improved_w2_exp} (i.e. $\mf{W}_2^{(2)}=0$).
    
    (2) We proceed via induction. \eqref{eq.w2_4} combined with \eqref{eq.improved_w2_exp} (i.e. $\mf{W}_2^{(4)}=0$) provides the base case. Fix $6 \leq 2K < n-2$ and assume that $\gm^{(j)}=0$ for all $6 \leq j \leq 2K$, i.e.
    \begin{align} \label{eq.ind_g_exp}
        \gv &= \gm - \rho^2 \cdot \mf{p} + \frac{1}{4} \rho^4 \cdot \gm^{-1} \mf{p} \mf{p} \\ &\qquad +
        \begin{cases}
            \sum_{k=K+1}^{\frac{n-1}{2}} \rho^{2k} \, \gm^{(2k)} + \rho^n \, \gm^{(n)} + \rho^n \, \ms{r} \text{,} \qquad &\text{$n$ odd}\text{,} \\
            \sum_{k=K+1}^{\frac{n-2}{2}} \rho^{2k} \, \gm^{(2k)} + \rho^n\log\rho \, \gm^{(\star)} + \rho^n \, \gm^{(n)} + \rho^n \, \ms{r} \text{,} \qquad &\text{$n$ even}\text{.}
        \end{cases} \notag
    \end{align}
    Substituting this into the relation
    \begin{align} \label{eq.inv_met_rel}
        \gv_{ab} \cdot \gv^{bc}=\delta_a^c \text{,}
    \end{align}
    and applying the inductive assumption to match coefficients order-by-order, one finds that
    \begin{align}\label{eq.inverse_metric_expansion_relations}
        \begin{cases}
            \mf{g}_{(0)}^{ab}=\mf{g}^{ab} \\
            \mf{g}_{(2)}^{ab}=\mf{g}^{ac} \mf{g}^{bd} \mf{p}_{cd} \\
            \mf{g}_{ab}\cdot \mf{g}_{(2l)}^{bc} - \mf{p}_{ab} \cdot \mf{g}_{(2l-2)}^{bc} + \frac{1}{4}\gm^{de} \cdot \mf{p}_{ad} \cdot \mf{p}_{be} \cdot \mf{g}_{(2l-4)}^{bc}=0 \text{,}\qquad 4 \leq 2l \leq 2K \text{,}
        \end{cases}
    \end{align}
    where $\mf{g}^{ab}:=\paren{\mf{g}^{-1}}^{ab}$. One may straightforwardly verify that this is satisfied by the ansatz \eqref{eq.inv_metric_comps}, which in turn satisfies the recurrence relation \eqref{eq.inv_met_rec_rel}.
    
    Substituting \eqref{eq.improved_w2_exp} into the left and \eqref{eq.ind_g_exp} into the right-hand side of \eqref{eq.w2.rel} and collecting strictly order $\rho^{2K}$ terms, one finds
    \begin{align}
        0 &= -2K(K+1) \gm^{(2K+2)}_{ab} + \gm_{(2K-2)}^{cd} \mf{p}_{ac} \mf{p}_{bd} \\
        &\qquad - \frac{1}{2} \gm_{(2K-4)}^{cd} \paren{ \mf{p}_{ac} \gm^{ef} \mf{p}_{be} \mf{p}_{df} + \gm^{ef} \mf{p}_{ae} \mf{p}_{cf} \mf{p}_{bd} } \notag\\
        &\qquad + \frac{1}{4} \gm_{(2K-6)}^{cd} \gm^{ef} \mf{p}_{ae} \mf{p}_{cf} \gm^{gh} \mf{p}_{bg} \mf{p}_{dh} \text{.}\notag
    \end{align}
    Applying the relations \eqref{eq.inv_met_rec_rel}, one finds that all but the first of the above terms cancel. Hence
    \begin{align}
        \gm^{(2K+2)} = 0 \text{,}
    \end{align}
    which closes the inductive argument.
    
    (3) Substituting \eqref{eq.improved_w2_exp} and the properties derived above in parts (1) and (2) into the left-hand side of \eqref{eq.w2.rel}, \eqref{eq.metric_FG_expansion} into the right-hand side of \eqref{eq.w2.rel} and collecting strictly order $\rho^{n-2} \log \rho$ terms, one finds
    \begin{align}
        \gm^{(\star)} =0 \text{,}
    \end{align}
    as required.
\end{proof}

\begin{proposition} \label{th.van_iff_bc}
    Fix $n>2$ and $M_0 \geq n+2$. Suppose $(\mi{M},g)$ is an $(n+1)$--dimensional, $M_0$--regular vacuum FG-aAdS segment. Then there exist vertical tensor fields $\ms{r}_0$, $\ms{r}_1$ and $\ms{r}_2$ such that
    \begin{align}
        \wv^0 &=\rho^{n-2} \cdot \ms{r}_0 \text{,}\qquad \ms{r}_0 \to^{M_0-n} 0 \text{,} \label{eq.w0_final_rate}\\
        \wv^1 &=\rho^{n-1} \cdot \ms{r}_1 \text{,}\qquad \ms{r}_1 \to^{M_0-n-1} 0 \text{,} \label{eq.w1_final_rate}\\
        \wv^2 &=\rho^{n-2} \cdot \ms{r}_2 \text{,}\qquad \ms{r}_2 \to^{M_0-n} 0 \text{.} \label{eq.w2_final_rate}
    \end{align}
    if and only if $\paren{\mi{M},g}$ has boundary data $\paren{\gm^{(0)} , \gm^{(n)}}$ for which
    \begin{align}
        &\text{$\gm^{(0)}$ is conformally flat, and} \label{eq.g0_condition}\\
        &\gm^{(n)}=
        \begin{cases}
            \frac{1}{4} \, \gm^{-1} \cdot \mf{p} \cdot \mf{p} \text{,}\qquad &n=4\text{,}\\
            0 \text{,}\qquad &n \neq 4\text{,}
        \end{cases} \label{eq.gn_condition}
    \end{align}
    where $\mf{p}$ is the $\gm^{(0)}=\gm$--Schouten tensor.
\end{proposition}

\begin{proof}
    We begin by proving the reverse direction of this statement. Given \eqref{eq.g0_condition}, we obtain a simplified expansion of $\gv$ via Corollary \ref{th.conf_flat_met_exp}. Substituting this along with \eqref{eq.improved_w2_exp} into \eqref{eq.w2.rel} and collecting order $\rho^{n-2}$ terms, we find
    \newlength{\temporary}
    \settowidth{\temporary}{$\qquad\qquad\qquad\qquad\qquad+ \frac{1}{4}\gm^{ef} \gm^{gh} \mf{p}_{ae} \mf{p}_{cf} \mf{p}_{bg} \mf{p}_{dh} \gm_{(n-8)}^{cd} \text{,}\,\,\,\,$}
    \begin{align}
        {\mf{W}_2}^{(n-2)}_{ab}&=-\frac{1}{2}n(n-2) \gm^{(n)}_{ab} \nonumber\\
        &\qquad +
        \begin{cases}
            \makebox[\temporary][l]{$0 \text{,}$} \text{$n$ odd,}\\
            \makebox[\temporary][l]{$\mf{p}_{ac} \mf{p}_{bd} \gm^{cd} \text{,}$} \text{$n=4$,}\\
            \makebox[\temporary][l]{$\mf{p}_{ac} \mf{p}_{bd} \gm_{(2)}^{cd}-\gm^{ef}\mf{p}_{ac}\mf{p}_{be}\mf{p}_{df} \gm^{cd} \text{,}$} \text{$n=6$,}\\
            \mf{p}_{ac} \mf{p}_{bd} \gm_{(n-4)}^{cd}-\frac{1}{2}\gm^{ef} \paren{ \mf{p}_{ac}\mf{p}_{be}\mf{p}_{df} + \mf{p}_{ae} \mf{p}_{cf} \mf{p}_{bd} } \gm_{(n-6)}^{cd}\\ \qquad\qquad\qquad\qquad\qquad+ \frac{1}{4}\gm^{ef} \gm^{gh} \mf{p}_{ae} \mf{p}_{cf} \mf{p}_{bg} \mf{p}_{dh} \gm_{(n-8)}^{cd} \text{,}\,\,\,\, n \geq 8 \text{.}
        \end{cases}
    \end{align}
    The relations \eqref{eq.inv_met_rec_rel} result in each of the $n \geq 6$ terms in the brace vanishing identically. The above hence reduces to
    \begin{align} \label{eq.w.(n-2)_g.n}
        {\mf{W}_2}^{(n-2)}_{ab}=
        \begin{cases}
            -\frac{1}{2}n(n-2) \gm^{(n)}_{ab} \text{,}\qquad &n \neq 4 \text{,}\\
            -4 \gm^{(4)}_{ab}+  \gm^{cd} \, \mf{p}_{ac} \, \mf{p}_{bd} \text{,}\qquad &n = 4 \text{,}
        \end{cases}
    \end{align}
    from which we see that \eqref{eq.gn_condition} implies
    \begin{align}
        \mf{W}_{2}^{(n-2)} = 0 \text{.}
    \end{align}
    We have thus used this condition to overcome the barrier to the iterative process identified in the proof of Lemma \ref{th.conf_flat_implies_vanishing_star}; substituting the updated expansions \eqref{eq.w0_exp}, \eqref{eq.w1_exp} and \eqref{eq.w2_exp} into the vertical Bianchi equation \eqref{eq.renorm_bianchi_3} and collecting strictly order $\rho^{n-2}$ terms, one finds
    \begin{align}
        \mf{W}_{0}^{(n-2)} = 0 \text{.}
    \end{align}
    Finally, substituting the updated expansions \eqref{eq.w0_exp}, \eqref{eq.w1_exp} and \eqref{eq.w2_exp} into the vertical Bianchi equation \eqref{eq.renorm_bianchi_2} and collecting strictly order $\rho^{n-2}$ terms, one finds
    \begin{align}
        \mf{W}_{1}^{(n-1)} = 0 \text{.}
    \end{align}
    This yields \eqref{eq.w0_final_rate}, \eqref{eq.w1_final_rate} and \eqref{eq.w2_final_rate} as required.

    For the forward direction, recall the leading-order expressions for $\wv^0$ and $\wv^1$ found Proposition \ref{th.weyl_leading_order}. If $n=3$ then the vanishing of ${\mf{W}_1} ^{(1)} = \mf{C}$ guaranteed by \eqref{eq.w1_final_rate} implies that $\gm=\gm^{(0)}$ is conformally flat. If $n > 3$ then the vanishing of ${\mf{W}_0} ^{(0)} = \mf{W}$ guaranteed by \eqref{eq.w0_final_rate} implies that $\gm=\gm^{(0)}$ is conformally flat.
    
    We hence have \eqref{eq.g0_condition} and, as above, can apply Corollary \ref{th.conf_flat_met_exp} to substitute the simplified expansion of $\gv$ into \eqref{eq.w2.rel} to obtain \eqref{eq.w.(n-2)_g.n}. The vanishing of $\mf{W}_2^{(n-2)}$ guaranteed by \eqref{eq.w2_final_rate} thus implies the condition \eqref{eq.gn_condition} as required.
\end{proof}

%%%%%%%%%%%%%%%%%%%%%%%%%%%%%%%%%%%%%%%%%%%%%%%%%%%

\subsection{The Main Result} \label{sec.main_result}

\begin{theorem} \label{th.primary}
    Fix $n>2$. \footnote{If $n=2$, the Weyl curvature identically vanishes. Hence, in this case, every vacuum FG-aAdS segment is locally isometric to pure AdS and the result in Theorem \ref{th.primary} is trivial.} Suppose $(\mi{M},g)$ is an $(n+1)$--dimensional vacuum FG-aAdS segment with a global time $t$ for which:
    \begin{itemize}
        \item $\mi{I}$ has compact cross-sections.
        \item The null convexity criterion holds on $\mi{I}$ with associated constants $0 \leq B <C$.
    \end{itemize}
    Fix constants $B<b<c<C$ and $t_0 \in \mathbb{R}$ such that \eqref{eq.5.5analogue} holds. There exist constants $M_0 \geq n+2$ and $f_\ast>0$ such that if $\paren{\mi{M},g}$ is $M_0$--regular then it is locally isometric to pure AdS on $\Omega(f_\ast)$ if and only if $\paren{\mi{M},g}$ has boundary data $\paren{\gm^{(0)} , \gm^{(n)}}$ for which the following hold on $\mi{I} \cap \Omega$:
    \begin{align}
        &\text{$\gm^{(0)}$ is conformally flat, and} \label{eq.g0_condition_final}\\
        &\gm^{(n)}=
        \begin{cases}
            \frac{1}{4} \, \gm^{-1} \cdot \mf{p} \cdot \mf{p} \text{,}\qquad &n=4\text{,}\\
            0 \text{,}\qquad &n \neq 4\text{,}
        \end{cases} \label{eq.gn_condition_final}
    \end{align}
    where $\mf{p}$ is the $\gm$--Schouten tensor.
\end{theorem}

\begin{proof}
    Solutions of \eqref{eq.eins} are locally isometric to AdS if and only if \footnote{See, for example, \cite{Carroll_2004}.}
    \begin{align}
        R_{\al\be\ga\de} = - g_{\al\ga}g_{\be\de} + g_{\al\de}g_{\be\ga} \text{,}
    \end{align}
    or, equivalently,
    \begin{align}
        W_{\al\be\ga\de} = 0 \text{.}
    \end{align}
    Vanishing of the spacetime Weyl curvature implies vanishing of $\wv^0$, $\wv^1$ and $\wv^2$ to arbitrarily high order at $\mi{I}$. By Proposition \ref{th.van_iff_bc}, the conditions \eqref{eq.g0_condition_final} and \eqref{eq.gn_condition_final} follow.
    
    For the other direction, let us write $o_{K}(\rho^L)$ to denote a vertical tensor field $\ms{t}$ for which
    \begin{align}
        \rho^{-L} \cdot \ms{t} \rightarrow^K \, 0 \text{.}
    \end{align}
    In terms of this notation, \eqref{eq.g0_condition_final} and \eqref{eq.gn_condition_final} imply
    \begin{align} \label{eq.int_rates}
        \wv^0=o_{M_0-n}(\rho^{n-2}) \text{,}\qquad \wv^1=o_{M_0-n-1}(\rho^{n-1}) \text{,}\qquad \wv^2=o_{M_0-n}(\rho^{n-2}) \text{,} 
    \end{align}
    by Proposition \ref{th.van_iff_bc}.
    
    Substituting \eqref{eq.int_rates} into \eqref{eq.ex_bi_2}, one has
    \begin{align} \label{eq.w2_int}
        \li \paren{ \rho^{-(n-2)} \wv^2 } &= \rho^{-(n-2)} \cdot o_{M_0-n-2}(\rho^{n-1}) + \rho^{-(n-2)} \cdot o_{M_0-n}(\rho^{n-1}) \\
        &= o_{M_0-n-2}(\rho) \text{,} \notag
    \end{align}
    where a degree of vertical regularity has been lost due to the vertical derivative present in the right-hand side of \eqref{eq.ex_bi_2}. Given \eqref{eq.int_rates}, $\rho^{-(n-2)} \wv^2$ has vanishing boundary limit so we may integrate \eqref{eq.w2_int} from the boundary to deduce
    \begin{align} \label{eq.imp_w2_int}
        \rho^{-(n-2)} \cdot \wv^2 = o_{M_0-n-2} ( \rho^2 ) \implies \wv^2 = o_{M_0-n-2} ( \rho^n ) \text{.}
    \end{align}
    Substituting \eqref{eq.int_rates} and \eqref{eq.imp_w2_int} into \eqref{eq.renorm_bianchi_3}, one has
    \begin{align}
        \li \wv^0 &= o_{M_0-n-2} (\rho^{n-1}) + \rho^{-1} \cdot o_{M_0-n-2} (\rho^n) + \rho^{-1} \cdot o_{M_0-n} (\rho^n) \\
        &= o_{M_0-n-2} (\rho^{n-1}) \text{.} \notag
    \end{align}
    Integrating this from the boundary, we deduce \footnote{Note that if we hadn't first derived \eqref{eq.imp_w2_int}, then one would only have $\li \wv^0 = o_{M_0-n-2} (\rho^{n-3})$ and we would not have been able to obtain this improvement via integration.}
    \begin{align} \label{eq.imp_w0_int}
        \wv^0 = o_{M_0-n-2}(\rho^n) \text{.}
    \end{align}
    Substituting \eqref{eq.int_rates} and \eqref{eq.imp_w2_int} into \eqref{eq.renorm_bianchi_2}, one has
    \begin{align}
        \li \paren{ \rho^{-1} \wv^1 } &= \rho^{-1} \cdot o_{M_0-n-3} (\rho^{n}) + \rho^{-1} \cdot o_{M_0-n-1} (\rho^n)\\
        &= o_{M_0-n-3} (\rho^{n-1}) \text{.} \notag
    \end{align}
    Again, a degree of vertical regularity has been lost due to the vertical derivative present in the right-hand side of \eqref{eq.renorm_bianchi_2}. Integrating from the boundary, we deduce \footnote{Again, it was crucial to derive \eqref{eq.imp_w2_int} first.}
    \begin{align} \label{eq.imp_w1_int}
        \rho^{-1} \cdot \wv^1 = o_{M_0-n-3}(\rho^n) \implies \wv^1 = o_{M_0-n-3}(\rho^{n+1}) \text{.}
    \end{align}
    Iterating the above process (that is, integrating \eqref{eq.ex_bi_2} followed by \eqref{eq.renorm_bianchi_3} and \eqref{eq.renorm_bianchi_2} using the improved rates \eqref{eq.imp_w2_int}, \eqref{eq.imp_w0_int} and \eqref{eq.imp_w1_int}) $i$ times, we deduce
    \begin{align}
        \wv^0 = o_{M_0-n-2i} ( \rho^{n-2+2i} ) \text{,} \qquad
        \wv^1 = o_{M_0-n-1-2i} ( \rho^{n-1+2i} ) \text{,} \qquad
        \wv^2 = o_{M_0-n-2i} ( \rho^{n-2+2i} ) \text{.}
    \end{align}
    At each iteration, additional orders of vanishing are exchanged for degrees of vertical regularity. By Lemma \ref{th.UC_van_rates}, for unique continuation to hold it will suffice for $n-2+2i \geq \kappa_{\text{max}}+2$, where
    \begin{align}
        \kappa_{\text{max}} :=
        \begin{cases}
            \frac{1}{2}\paren{n-2+\sqrt{n^2+4C_0}} \text{,}\qquad &C_0 \leq n+1 \text{,}\\
            \frac{1}{2}\paren{n-1+C_0} \text{,}\qquad &C_0 > n+1 \text{.}
        \end{cases}
    \end{align}
    In other words, we must iterate $i = \left \lceil{ \frac{1}{2} \paren{ \kappa_{\text{max}}-n+4 } }\right \rceil$ times. If we choose $M_0 = \left \lceil{ \kappa_{\text{max}} }\right \rceil + 6$ then this yields
    \begin{align}
        \wv^0 = o_{2} ( \rho^{\kappa_{\text{max}}+2} ) \text{,}\qquad
        \wv^1 = o_{1} ( \rho^{\kappa_{\text{max}}+3} ) \text{,}\qquad
        \wv^2 = o_{2} ( \rho^{\kappa_{\text{max}}+2} ) \text{,}
    \end{align}
    which in turn implies \eqref{eq.tilde_van_cond}, \eqref{eq.tilde_van_cond_1} and \eqref{eq.tilde_van_cond_2} as required.
\end{proof}

The following Proposition demonstrates that, for the result of Theorem \ref{th.primary} to hold, it is only necessary for the relevant boundary data conditions to hold in \textit{one} Fefferman-Graham coordinate system.

\begin{proposition} \label{th.gauge_invariance}
    Fix $n>2$ and $M_0 \geq n+2$. Suppose $(\mi{M},g)$ is an $(n+1)$--dimensional, $M_0$--regular vacuum FG-aAdS segment. The conditions \eqref{eq.g0_condition_final} and \eqref{eq.gn_condition_final} are invariant under coordinate transformations $\paren{\rho , x} \to \paren{\tilde{\rho},\tilde{x}}$ of the form
    \begin{align}
        \rho &= \tilde{\rho} \cdot e^{-\sigma(\tilde{x})} + \tilde{\rho}^2 \cdot \mf{a}^{(2)}(\tilde{x}) + \tilde{\rho}^3 \cdot \mf{a}^{(3)}(\tilde{x}) + \tilde{\rho}^3 \cdot \ms{s} \text{,}\\
        x^b &= \tilde{x}^b + \tilde{\rho} \cdot \mf{a}_{(1)}^b(\tilde{x}) + \tilde{\rho}^2 \cdot \mf{a}_{(2)}^b(\tilde{x}) + \tilde{\rho}^2 \cdot \ms{t}^b \text{,}
    \end{align}
    where $\sigma(\tilde{x})$ is $C^{M_0+1}$ on $\mi{I}$ and $\ms{s}$ and $\ms{t}$ are vertical (with respect to $\tilde{\rho}$) tensor fields satisfying
    \begin{align}
        \ms{s} \to^{M_0-1} \, 0 \text{,}\qquad \ms{t} \to^{M_0-1} \, 0 \text{,}
    \end{align}
    for which the Fefferman-Graham gauge \eqref{eq.FG_gauge} is preserved, i.e. for which there exists some vertical (with respect to $\tilde{\rho}$) tensor field $\tilde{\gv}$ such that
\begin{align}
    g = \tilde{\rho}^{-2} \brak{d \tilde{\rho}^2 + \tilde{\gv}_{ab} \, d\tilde{x}^a d \tilde{x}^b} \text{.}
\end{align}
\end{proposition}

\begin{proof}
As in \cite{Skenderis_2001}, one may uniquely determine the coefficients $\mf{a}^{(k)}$ and $\mf{a}_{(k)}^b$ order-by-order. At leading-order one finds
\begin{align}
    \mf{a}^{(2)} &= 0 \text{,}\qquad \mf{a}^{(3)} = - \frac{1}{4} e^{-3\sigma} \gm^{bc} \Dm_b \sigma \, \Dm_c \sigma \text{,} \\
    \mf{a}_{(1)}^b &= 0 \text{,}\qquad \mf{a}_{(2)}^b = \frac{1}{2} e^{-2\sigma} \gm^{bc}\Dm_{c} \sigma \text{,}
\end{align}
One may use this to compute
\begin{align}
    \frac{\partial}{\partial\tilde{\rho}} = \frac{\partial \rho}{\partial \tilde{\rho}} \frac{\partial}{\partial \rho} + \frac{\partial x^b}{\partial \tilde{\rho}} \frac{\partial}{\partial x^b} \text{,}\qquad  \frac{\partial}{\partial\tilde{x}^a} = \frac{\partial \rho}{\partial \tilde{x}^a} \frac{\partial}{\partial \rho} + \frac{\partial x^b}{\partial \tilde{x}^a} \frac{\partial}{\partial x^b} \text{,}
\end{align}
and thus derive (to leading-order) the transformations of the vertical components of the Weyl curvature under the above change of coordinates. For example, one finds
\begin{align}
    \frac{\partial \rho}{\partial \tilde{x}^a} = \mi{O}_{M_0-2}(\rho)_a \text{,}\qquad \frac{\partial x^b}{\partial \tilde{x}^a} = \delta^b_a + \mi{O}_{M_0-2}(\rho^2)^b_a \text{,}
\end{align}
which implies that
\begin{align}
    \tilde{\wv}^0_{\tilde{a}\tilde{b}\tilde{c}\tilde{d}} &= \paren{ \rho e^\sigma + \mi{O}_{M_0-2}(\rho^2) }^2 \cdot W \paren{ \mi{O}_{M_0-2}(\rho)_a \frac{\partial}{\partial \rho} + \paren{ 1 + \mi{O}_{M_0-2}(\rho^2) } \frac{\partial}{\partial x^a}, \dots } \\
    &= e^{2\sigma} \wv^0_{abcd} + \mc{O}_{M_0-2}(\rho^2;\wv^0)_{abcd} + \mc{O}_{M_0-2}(\rho;\wv^1)_{abcd} + \mc{O}_{M_0-2}(\rho^2;\wv^2)_{abcd} \text{.} \notag
\end{align}
Similarly, one finds
\begin{align}\label{eq.w2_tilde}
    \tilde{\wv}^1_{\tilde{a}\tilde{b}\tilde{c}} &= e^\sigma \wv^1_{abc} + \mc{O}_{M_0-2}(\rho;\wv^0)_{abc} + \mc{O}_{M_0-2}(\rho^2;\wv^1)_{abc} + \mc{O}_{M_0-2}(\rho;\wv^2)_{abc} \text{,}\\
    \tilde{\wv}^2_{\tilde{a}\tilde{b}} &= \wv^2_{ab} +\mc{O}_{M_0-2}(\rho^2;\wv^0)_{ab}+\mc{O}_{M_0-2}(\rho;\wv^1)_{ab}+\mc{O}_{M_0-2}(\rho^2;\wv^2)_{ab} \text{.}
\end{align}
Applying Proposition \ref{th.van_iff_bc}, the conditions \eqref{eq.g0_condition_final} and \eqref{eq.gn_condition_final} imply \eqref{eq.w0_final_rate}, \eqref{eq.w1_final_rate} and \eqref{eq.w2_final_rate}. By the above, we may hence deduce the existence of vertical tensor fields $\tilde{\ms{r}}_0$, $\tilde{\ms{r}}_1$ and $\tilde{\ms{r}}_2$ for which
\begin{align}
    \tilde{\wv}^0 &= \tilde{\rho}^{n-2} \cdot \tilde{\ms{r}}_0 \text{,}\qquad \tilde{\ms{r}}_0 \to^{M_0-n} 0 \text{,} \label{eq.w0_tilde_final_rate}\\
    \tilde{\wv}^1 &= \tilde{\rho}^{n-1} \cdot \tilde{\ms{r}}_1 \text{,}\qquad \tilde{\ms{r}}_1 \to^{M_0-n-1} 0 \text{,} \label{eq.w1_tilde_final_rate}\\
    \tilde{\wv}^2 &= \tilde{\rho}^{n-2} \cdot \tilde{\ms{r}}_2 \text{,}\qquad \tilde{\ms{r}}_2 \to^{M_0-n} 0 \text{.} \label{eq.w2_tilde_final_rate}
\end{align}
Since Proposition \ref{th.van_iff_bc} also applies with respect to the new coordinate system, we may apply it to conclude that
\begin{align}
    &\text{$\tilde{\gm}^{(0)}$ is conformally flat, and}\\
    &\tilde{\gm}^{(n)}=
    \begin{cases}
        \frac{1}{4} \, \tilde{\gm}^{-1} \cdot \tilde{\mf{p}} \cdot \tilde{\mf{p}} \text{,}\qquad &n=4\text{,}\\
        0 \text{,}\qquad &n \neq 4\text{,}
    \end{cases}
\end{align}
as required.
\end{proof}

%%%%%%%%%%%%%%%%%%%%%%%%%%%%%%%%%%%%%%%%%%%%%%%%%%%%%%%%%%

\appendix

\section{Proofs of Conversion Formulae} \label{sec.appendixA}

Given a mixed tensor field $\mathbf{A}$ of rank $( \kappa, \lambda; k, l )$, Propositions \ref{def.aads_vertical_connection} and \ref{thm.aads_mixed_connection} give the following formula for the mixed derivative in terms of $\varphi$- and $\varphi$-coordinates:
\begin{align}
    \label{eq.extra_deriv} \nablam_\gamma \mathbf{A}^{ \ix{\alpha} }{}_{ \ix{\beta} }{}^{ \ix{a} }{}_{ \ix{b} } &= \partial_\gamma ( \mathbf{A}^{ \ix{\alpha} }{}_{ \ix{\beta} }{}^{ \ix{a} }{}_{ \ix{b} } ) + \sum_{ i = 1 }^\kappa \Gamma^{ \alpha_i }_{ \gamma \delta } \, \mathbf{A}^{ \ixr{\alpha}{i}{\delta} }{}_{ \ix{\beta} }{}^{ \ix{a} }{}_{ \ix{b} } - \sum_{ j = 1 }^\lambda \Gamma^\delta_{ \gamma \beta_j } \, \mathbf{A}^{ \ix{\alpha} }{}_{ \ixr{\beta}{j}{\delta} }{}^{ \ix{a} }{}_{ \ix{b} } \\
    \notag &\qquad + \sum_{ i = 1 }^k \Gammav^{ a_i }_{ \gamma d } \, \mathbf{A}^{ \ix{\alpha} }{}_{ \ix{\beta} }{}^{ \ixr{a}{i}{d} }{}_{ \ix{b} } - \sum_{ j = 1 }^l \Gammav^d_{ \gamma b_j } \, \mathbf{A}^{ \ix{\alpha} }{}_{ \ix{\beta} }{}^{ \ix{a} }{}_{ \ixr{b}{j}{d} } \text{,}
\end{align}
where the Christoffel symbols
\begin{equation}
    \label{eq.extra_Gamma_def} \nabla_\alpha \partial_\beta := \Gamma^\gamma_{ \alpha \beta } \partial_\gamma \text{,} \qquad \Dvm_\alpha \partial_b := \bar{\Gamma}^c_{ \alpha b } \partial_c \text{,}
\end{equation}
are given by
\begin{align}
    \label{eq.extra_Gamma} \Gamma_{ \rho \rho }^\alpha = - \rho^{-1} \delta^\alpha{}_\rho \text{,} &\qquad \Gamma_{ a \rho }^\rho = 0 \text{,} \\
    \notag \Gamma_{ a \rho }^c = - \rho^{-1} \delta^c{}_a + \frac{1}{2} \gv^{ c d } \li \gv_{ a d } \text{,} &\qquad \Gamma_{ \rho a }^c - \Gammav_{ \rho a }^c = - \rho^{-1} \delta^c{}_a \text{,} \\
    \notag \Gamma_{ a b }^\rho = \rho^{-1} \gv_{ a b } - \frac{1}{2} \li \gv_{ a b } \text{,} &\qquad \Gamma_{ a b }^c - \Gammav_{ a b }^c = 0 \text{.}
\end{align}
 
%%%%%%%%%%%%%%%%%%%%%%%%%%%%%%%%%%%%%%%%%%%%%%%%%%%%%%%%%%%%%%%%%

\subsection{Proposition \ref{thm.prelim_comm}} \label{sec.st_vert_conversion_proof1}

Suppose that $\ms{A}$ has rank $( k, l )$. \eqref{eq.rho_deriv} gives, with respect to any coordinates $( U, \varphi )$ on $\mi{I}$,
\begin{align} \label{eq.prelim_comm_a}
    \bar{\Dv}_\rho \Dv_c \ms{A}^{ \ix{a} }{}_{ \ix{b} } &= \li \Dv_c \ms{A}^{ \ix{a} }{}_{ \ix{b} } - \frac{1}{2} \gv^{ d e } \li \gv_{ c d } \, \Dv_e \ms{A}^{ \ix{a} }{}_{ \ix{b} } + \frac{1}{2} \sum_{ i = 1 }^k \gv^{ a_i d } \li \gv_{ d e } \, \Dv_c \ms{A}^{ \ixr{a}{i}{e} }{}_{ \ix{b} } \\
    &\qquad - \frac{1}{2} \sum_{ j = 1 }^l \gv^{ d e } \li \gv_{ b_j d } \, \Dv_c \ms{A}^{ \ix{a} }{}_{ \ixr{b}{j}{e} } \text{,} \notag\\
    \Dv_c ( \bar{\Dv}_\rho \ms{A} )^{ \ix{a} }{}_{ \ix{b} } &= \Dv_c ( \li \ms{A} )^{ \ix{a} }{}_{ \ix{b} } + \frac{1}{2} \sum_{ i = 1 }^k \Dv_c ( \gv^{ a_i d } \li \gv_{ d e } \, \ms{A}^{ \ixr{a}{i}{e} }{}_{ \ix{b} } ) - \frac{1}{2} \sum_{ j = 1 }^l \Dv_c ( \gv^{ d e } \li \gv_{ b_j d } \, \ms{A}^{ \ix{a} }{}_{ \ixr{b}{j}{e} } ) \text{.}\notag
\end{align}
The difference of these equations reads
\begin{align}
    \label{eql.prelim_comm_0} \bar{\Dv}_\rho \Dv_c \ms{A}^{ \ix{a} }{}_{ \ix{b} } &= \Dv_c ( \bar{\Dv}_\rho \ms{A} )^{ \ix{a} }{}_{ \ix{b} } + \li \Dv_c \ms{A}^{ \ix{a} }{}_{ \ix{b} } - \Dv_c ( \li \ms{A} )^{ \ix{a} }{}_{ \ix{b} } - \frac{1}{2} \gv^{ d e } \li \gv_{ c d } \, \Dv_e \ms{A}^{ \ix{a} }{}_{ \ix{b} } \\
    \notag &\qquad - \frac{1}{2} \sum_{ i = 1 }^k \gv^{ a_i d } \Dv_c \li \gv_{ d e } \, \ms{A}^{ \ixr{a}{i}{e} }{}_{ \ix{b} } + \frac{1}{2} \sum_{ j = 1 }^l \gv^{ d e } \Dv_c \li \gv_{ b_j d } \, \ms{A}^{ \ix{a} }{}_{ \ixr{b}{j}{e} } \text{.}
\end{align}
We apply the commutation formula contained in \cite[Proposition 2.27]{Shao_2020} to obtain
\begin{align} \label{eqn.2_derivs}
    \bar{\Dv}_\rho \Dv_c \ms{A}^{ \ix{a} }{}_{ \ix{b} } &= \Dv_c ( \bar{\Dv}_\rho \ms{A} )^{ \ix{a} }{}_{ \ix{b} } - \frac{1}{2} \gv^{ d e } \li \gv_{ c d } \, \Dv_e \ms{A}^{ \ix{a} }{}_{ \ix{b} } + \frac{1}{2} \sum_{ i = 1 }^k \gv^{ a_i d } ( \Dv_e \li \gv_{ c d } - \Dv_d \li \gv_{ c e } ) \ms{A}^{ \ixr{a}{i}{e} }{}_{ \ix{b} } \\
    &\qquad - \frac{1}{2} \sum_{ j = 1 }^l \gv^{ d e } ( \Dv_{ b_j } \li \gv_{ c d } - \Dv_d \li \gv_{ c b_j } ) \ms{A}^{ \ix{a} }{}_{ \ixr{b}{j}{e} } \text{.}\notag
\end{align}
Combining \eqref{eqn.2_derivs} with \eqref{eq.prelim_regular} yields the identity \eqref{eq.prelim_comm}, as required.

Next, recalling Definition \ref{def.aads_mixed_wave} of $\Boxm$, we expand (partially in $\varphi_\rho$-coordinates)
\begin{align}
    \label{eql.prelim_comm_rho_0} \Boxm ( \rho^p \ms{A} ) &= \gv^{ \alpha \beta } \nablam_\alpha ( \rho^p \nablam_\beta \ms{A} + p \rho^{ p - 1 } \nablam_\beta \rho \cdot \ms{A} ) \\
    \notag &= \rho^p \Boxm \ms{A} + 2 p \rho^{ p - 1 } g^{ \alpha \beta } \nabla_\alpha \rho \, \Dvm_\beta \ms{A} + p ( p - 1 ) \rho^{ p - 2 } g^{ \alpha \beta } \nabla_\alpha \rho \nabla_\beta \rho \, \ms{A} + p \rho^{ p - 1 } \Box \rho \, \ms{A} \\
    \notag &= \rho^p \Boxm \ms{A} + 2 p \rho^{ p + 1 } \, \Dvm_\rho \ms{A} + p ( p - 1 ) \rho^p \, \ms{A} + p \rho^{ p - 1 } \Box \rho \, \ms{A} \text{.}
\end{align}
The formulae in \eqref{eq.extra_Gamma} imply that
\begin{align}\label{eq.box_rho}
    \Box \rho &= - \rho^2 \Gamma^\rho_{ \rho \rho } - \rho^2 \gv^{ a b } \Gamma^\rho_{ a b } \\
    &= - ( n - 1 ) \rho + \frac{1}{2} \rho^2 \gv^{ a b } \li \gv_{ a b } \text{.}\notag
\end{align}
Substituting \eqref{eq.box_rho} into \eqref{eql.prelim_comm_rho_0} yields
\begin{equation}\label{eq.box_rho_A}
    \Boxm ( \rho^p \ms{A} ) = \rho^p \Boxm \ms{A} + 2 p \rho^{ p + 1 } \, \Dvm_\rho \ms{A} - p ( n - p ) \rho^p \, \ms{A} + \frac{1}{2} p \rho^{ p + 1 } \gv^{ a b } \li \gv_{ a b } \, \ms{A} \text{.}
\end{equation}
Combining \eqref{eq.box_rho_A} with \eqref{eq.prelim_regular} immediately implies \eqref{eq.prelim_comm_rho}, as required.

\qed

%%%%%%%%%%%%%%%%%%%%%%%%%%%%%%%%%%%%%%%%%%%%%%%%%%%%%%%%%%%%%

\subsection{Proposition \ref{thm.prelim_decomp}} \label{sec.st_vert_conversion_proof2}

Assume all indices are with respect to $\varphi$- and $\varphi_\rho$-coordinates and let $\Gamma$ and $\Gammav$ be the corresponding Christoffel symbols defined \eqref{eq.extra_Gamma}.
For future convenience, we set $l := r_1 + r_2$ and define (via local coordinates) the vertical tensor fields
\begin{equation}
    \label{eql.prelim_decomp_0} \kv_{ a c } := \rho^{-1} \gv_{ a c } - \frac{1}{2} \li \gv_{ a c } \text{,} \qquad \kv^b{}_a := \rho^{-1} \delta^b{}_a - \frac{1}{2} \gv^{ b c } \li \gv_{ a c } \text{.}
\end{equation}
By the definitions of $\nabla$ and $\Dvm$, we have
\begin{align*}
    \nabla_\rho A_{ \ix{\rho} \ix{a} } &= \partial_\rho ( A_{ \ix{\rho} \ix{a} } ) - \sum_{ i = 1 }^{ r_1 } \Gamma^\rho_{ \rho \rho } \, A_{ \ix{\rho} \ix{a} } - \sum_{ j = 1 }^{ r_2 } \Gamma^b_{ \rho a_j } \, A_{ \ix{\rho} \ixr{a}{j}{b} } \text{,} \\
    \Dvm_\rho \ms{A}_{ \ix{a} } &= \partial_\rho ( \ms{A}_{ \ix{a} } ) - \sum_{ j = 1 }^{ r_2 } \Gammav^b_{ \rho a_j } \ms{A}_{ \ixr{a}{j}{b} } \text{.}
\end{align*}
Subtracting these equations and applying \eqref{eq.prelim_decomp_phi} and \eqref{eq.extra_Gamma} yields
\begin{align*}
    \nabla_\rho A_{ \ix{\rho} \ix{a} } &= \Dvm_\rho \ms{A}_{ \ix{a} } - \sum_{ i = 1 }^{ r_1 } \Gamma^\rho_{ \rho \rho } \, \ms{A}_{ \ix{a} } - \sum_{ j = 1 }^{ r_2 } ( \Gamma^b_{ \rho a_j } - \Gammav^b_{ \rho a_j } ) \ms{A}_{ \ixr{a}{j}{b} } \\
    &= \Dvm_\rho \ms{A}_{ \ix{a} } + r_1 \rho^{-1} \, \ms{A}_{ \ix{a} } + r_2 \rho^{-1} \, \ms{A}_{ \ix{a} } \\
    &= \Dvm_\rho \ms{A}_{ \ix{a} } + ( r_1 + r_2 ) \rho^{-1} \, \ms{A}_{ \ix{a} } \text{,}
\end{align*}
from which \eqref{eq.prelim_decomp_rho} follows.

Similarly, the definitions of $\nabla$ and $\Dvm$ imply
\begin{align*}
    \nabla_c A_{ \ix{\rho} \ix{a} } &= \partial_c ( A_{ \ix{\rho} \ix{a} } ) - \sum_{ i = 1 }^{ r_1 } \Gamma_{ c \rho }^b \, A_{ \ixr{\rho}{i}{b} \ix{a} } - \sum_{ j = 1 }^{ r_2 } \Gamma_{ c a_j }^\beta \, A_{ \ix{\rho} \ixr{a}{j}{\beta} } \text{,} \\
    \Dvm_c \ms{A}_{ \ix{a} } &= \partial_c ( \ms{A}_{ \ix{a} } ) - \sum_{ j = 1 }^{ r_2 } \Gammav_{ c a_j }^b \, \ms{A}_{ \ixr{a}{j}{b} } \text{.}
\end{align*}
Subtracting the above equations and recalling \eqref{eq.prelim_decomp_phi_rv}, \eqref{eq.prelim_decomp_phi_vr}, and \eqref{eq.extra_Gamma}, we obtain
\begin{align}
    \nabla_c A_{ \ix{\rho} \ix{a} } &= \Dvm_c \ms{A}_{ \ix{a} } - \sum_{ i = 1 }^{ r_1 } \Gamma_{ c \rho }^b \, A_{ \ixr{\rho}{i}{b} \ix{a} } - \sum_{ j = 1 }^{ r_2 } \Gamma_{ c a_j }^\rho \, A_{ \ix{\rho} \ixr{a}{j}{\rho} } \\
    &= \Dvm_c \ms{A}_{ \ix{a} } + \sum_{ i = 1 }^{ r_1 } \kv^b{}_c \, ( \ms{A}^\rho_i )_{ b \ix{a} } - \sum_{ j = 1 }^{ r_2 } \kv_{ c a_j } \, ( \ms{A}^v_j )_{ \ixd{a}{j} } \text{.} \label{eq.exact_vert_conv}
\end{align}
Combining the above with \eqref{eq.prelim_regular} and \eqref{eql.prelim_decomp_0} yields \eqref{eq.prelim_decomp_a}.

For \eqref{eq.box}, we start by computing $\rho$-derivatives.
By \eqref{eq.extra_deriv},
\begin{align}
    \label{eql.prelim_decomp_1} \nabla_{ \rho \rho } A_{ \ix{\rho} \ix{a} } &= \partial_\rho ( \nabla_\rho A_{ \ix{\rho} \ix{a} } ) - \Gamma^\rho_{ \rho \rho } \, \nabla_\rho A_{ \ix{\rho} \ix{a} } - \sum_{ i = 1 }^{ r_1 } \Gamma^\rho_{ \rho \rho } \, \nabla_\rho A_{ \ix{\rho} \ix{a} } - \sum_{ j = 1 }^{ r_2 } \Gamma^b_{ \rho a_j } \, \nabla_\rho A_{ \ix{\rho} \ixr{a}{j}{b} } \text{,} \\
    \notag &= \partial_\rho ( \nabla_\rho A_{ \ix{\rho} \ix{a} } ) + ( r_1 + 1 ) \rho^{-1} \, \nabla_\rho A_{ \ix{\rho} \ix{a} } - \sum_{ j = 1 }^{ r_2 } \Gamma^b_{ \rho a_j } \, \nabla_\rho A_{ \ix{\rho} \ixr{a}{j}{b} } \text{,}
\end{align}
Similarly, for the corresponding mixed derivatives, we apply \eqref{eq.extra_deriv} and compute
\begin{align}
    \label{eql.prelim_decomp_2} \rho^{ -l } \nablam_{ \rho \rho } ( \rho^l \ms{A} )_{ \ix{a} } &= \rho^{ -l } \partial_\rho [ \Dvm_\rho ( \rho^l \ms{A} )_{ \ix{a} } ] - \Gamma^\rho_{ \rho \rho } \, \rho^{ -l } \Dvm_\rho  ( \rho^l \ms{A} )_{ \ix{a} } - \sum_{ j = 1 }^{ r_2 } \Gammav^b_{ \rho a_j } \, \rho^{ -l } \Dvm_\rho ( \rho^l \ms{A} )_{ \ixr{a}{j}{b} } \\
    \notag &= \partial_\rho [ \rho^{ -l } \Dvm_\rho ( \rho^l \ms{A} )_{ \ix{a} } ] + ( l + 1 ) \rho^{-1} \, \rho^{ -l } \Dvm_\rho  ( \rho^l \ms{A} )_{ \ix{a} } - \sum_{ j = 1 }^{ r_2 } \Gammav^b_{ \rho a_j } \, \rho^{ -l } \Dvm_\rho ( \rho^l \ms{A} )_{ \ixr{a}{j}{b} } \text{,}
\end{align}
where we also applied \eqref{eq.extra_Gamma} and the properties contained in Proposition \ref{thm.aads_mixed_connection}.
Subtracting \eqref{eql.prelim_decomp_2} from \eqref{eql.prelim_decomp_1}, while applying both \eqref{eq.extra_Gamma} and \eqref{eq.prelim_decomp_rho}, we obtain that
\begin{align}\label{eql.prelim_decomp_10} 
    \nabla_{ \rho \rho } A_{ \ix{\rho} \ix{a} } &= \rho^{ -l } \nablam_{ \rho \rho } ( \rho^l \ms{A} )_{ \ix{a} } + ( r_1 - l ) \rho^{-1} \, \rho^{ -l } \Dvm_\rho ( \rho^l \ms{A} )_{ \ix{a} } - \sum_{ j = 1 }^{ r_2 } ( \Gamma^b_{ \rho a_j } - \Gammav^b_{ \rho a_j } ) \, \rho^{ -l } \Dvm_\rho ( \rho^l \ms{A} )_{ \ixr{a}{j}{b} } \\
    \notag &= \rho^{ -l } \nablam_{ \rho \rho } ( \rho^l \ms{A} )_{ \ix{a} } + ( r_1 - l ) \rho^{-1} \, \rho^{ -l } \Dvm_\rho ( \rho^l \ms{A} )_{ \ix{a} } - r_2 \rho^{-1} \, \rho^{ -l } \Dvm_\rho ( \rho^l \ms{A} )_{ \ix{a} } \\
    \notag &= \rho^{ -l } \nablam_{ \rho \rho } ( \rho^l \ms{A} )_{ \ix{a} } \text{.}
\end{align}
Next, we apply \eqref{eq.extra_deriv} to compute
\begin{align*}
    \nabla_{ b c } A_{ \ix{\rho} \ix{a} } &= \partial_b ( \nabla_c A_{ \ix{\rho} \ix{a} } ) - \Gamma^\alpha_{ b c } \, \nabla_\alpha A_{ \ix{\rho} \ix{a} } - \sum_{ i = 1 }^{ r_1 } \Gamma^d_{ b \rho } \, \nabla_c A_{ \ixr{\rho}{i}{d} \ix{a} } - \sum_{ j = 1 }^{ r_2 } \Gamma^\delta_{ b a_j } \, \nabla_c A_{ \ix{\rho} \ixr{a}{j}{\delta} } \text{,} \\
    \rho^{-l} \nablam_{ b c } ( \rho^l \ms{A} )_{ \ix{a} } &= \partial_b ( \Dvm_c \ms{A}_{ \ix{a} } ) - \Gamma^\alpha_{ b c } \, \rho^{-l} \Dvm_\alpha ( \rho^l \ms{A} )_{ \ix{a} } - \sum_{ j = 1 }^{ r_2 } \Gammav^d_{ b a_j } \, \Dvm_c \ms{A}_{ \ixr{a}{j}{d} } \text{.}
\end{align*}
Subtracting the two equations and recalling \eqref{eq.extra_Gamma} yields
\begin{align}\label{eql.prelim_decomp_11} 
    \nabla_{ b c } A_{ \ix{\rho} \ix{a} } &= \rho^{-l} \nablam_{ b c } ( \rho^l \ms{A} )_{ \ix{a} } + \partial_b ( \nabla_c A_{ \ix{\rho} \ix{a} } - \Dvm_c \ms{A}_{ \ix{a} } ) - \Gamma^\alpha_{ b c } [ \nabla_\alpha A_{ \ix{\rho} \ix{a} } - \rho^{-l} \Dvm_\alpha ( \rho^l \ms{A} )_{ \ix{a} } ] \\
    \notag &\qquad - \sum_{ i = 1 }^{ r_1 } \Gamma^d_{ b \rho } \, \nabla_c A_{ \ixr{\rho}{i}{d} \ix{a} } - \sum_{ j = 1 }^{ r_2 } \Gamma^\rho_{ b a_j } \, \nabla_c A_{ \ix{\rho} \ixr{a}{j}{\rho} } - \sum_{ j = 1 }^{ r_2 } \Gammav^d_{ b a_j } \, ( \nabla_c A_{ \ix{\rho} \ix{a} } - \Dvm_c \ms{A}_{ \ix{a} } ) \\
    \notag &:= \rho^{-l} \nablam_{ b c } ( \rho^l \ms{A} )_{ \ix{a} } + I_1 + I_2 + I_3 + I_4 + I_5 \text{.}
\end{align}
To simplify the upcoming computations, we define, for all $1 \leq i \leq r_1$ and $1 \leq j \leq r_2$, the vertical tensor fields $\ms{z}$, $\ms{z}^\rho_i$, $\ms{z}^v_j$---of ranks $( 0, r_2 + 1 )$, $( 0, r_2 + 2 )$, $( 0, r_2 )$, respectively---via the index formulae
\begin{align}\label{eql.prelim_decomp_20} 
    \ms{z}_{ c \ix{a} } &:= \nabla_c A_{ \ix{\rho} \ix{a} } - \Dvm_c \ms{A}_{ \ix{a} } \text{,} \\
    \notag ( \ms{z}^\rho_i )_{ c b \ix{a} } &:= \nabla_c A_{ \ixr{\rho}{i}{b} \ix{a} } - \Dvm_c ( \ms{A}^\rho_i )_{ b \ix{a} } \text{,} \\
    \notag ( \ms{z}^v_j )_{ c \ixd{a}{j} } &:= \nabla_c A_{ \ix{\rho} \ixr{a}{j}{\rho} } - \Dvm_c ( \ms{A}^v_j )_{ \ixd{a}{j} } \text{.}
\end{align}

Applying \eqref{eq.prelim_decomp_rho}, \eqref{eq.extra_Gamma}, and \eqref{eql.prelim_decomp_20} to the term $I_2$ from \eqref{eql.prelim_decomp_11}, we obtain
\begin{align*}
    I_2 &= - \kv_{ b c } [ \nabla_\rho A_{ \ix{\rho} \ix{a} } - \rho^{ -l } \Dvm_\rho ( \rho^l \ms{A} )_{ \ix{a} } ] - \Gamma^d_{ b c } ( \nabla_d A_{ \ix{\rho} \ix{a} } - \Dvm_d \ms{A}_{ \ix{a} } ) \\
    &= - \Gamma^d_{ b c } ( \nabla_d A_{ \ix{\rho} \ix{a} } - \Dvm_d \ms{A}_{ \ix{a} } ) \text{.}
\end{align*}
From \eqref{eql.prelim_decomp_11}, the first part of \eqref{eql.prelim_decomp_20}, and the above, we see that
\begin{equation}
    \label{eql.prelim_decomp_21} I_1 + I_2 + I_5 = \Dvm_b \ms{z}_{ c \ix{a} } \text{.}
\end{equation}
Similarly, for $I_3$ and $I_4$, we again apply \eqref{eq.extra_Gamma} and \eqref{eql.prelim_decomp_20}:
\begin{align}\label{eql.prelim_decomp_22} 
    I_3 &= \sum_{ i = 1 }^{ r_1 } \kv^d{}_b \, \Dvm_c ( \ms{A}^\rho_i )_{ d \ix{a} } + \sum_{ i = 1 }^{ r_1 } \kv^d{}_b \, ( \ms{z}^\rho_i )_{ c d \ix{a} } \text{,} \\
    \notag I_4 &= - \sum_{ j = 1 }^{ r_2 } \kv_{ a_j b } \, \Dvm_c ( \ms{A}^v_j )_{ \ixd{a}{j} } - \sum_{ j = 1 }^{ r_2 } \kv_{ a_j b } \, ( \ms{z}^v_j )_{ c \ixd{a}{j} } \text{.}
\end{align}

Now, recalling \eqref{eq.prelim_decomp_a}, along with \eqref{eql.prelim_decomp_20}, we deduce
\begin{align}\label{eql.prelim_decomp_30} 
    \ms{z}_{ c \ix{a} } &= \sum_{ i = 1 }^{ r_1 } \kv^e{}_c \, ( \ms{A}^\rho_i )_{ e \ix{a} } - \sum_{ j = 1 }^{ r_2 } \kv_{ a_j c } \, ( \ms{A}^v_j )_{ \ixd{a}{j} } \text{,} \\
    \notag ( \ms{z}^\rho_i )_{ c d \ix{a} } &= \sum_{ \substack{ 1 \leq j \leq r_1 \\ j \neq i } } \kv^e{}_c \, ( \ms{A}^{ \rho, \rho }_{ i, j } )_{ e d \ix{a} } - \sum_{ j = 1 }^{ r_2 } \kv_{ a_j c } \, ( \ms{A}^{ \rho, v }_{ i, j } )_{ d \ixd{a}{j} } - \kv_{ b c } \, \ms{A}_{ \ix{a} } \text{,} \\
    \notag ( \ms{z}^v_j )_{ c \ixd{a}{j} } &= \sum_{ i = 1 }^{ r_1 } \kv^e{}_c \, ( \ms{A}^{ \rho, v }_{ i, j } )_{ e \ixd{a}{j} } + \kv^e{}_c \, \ms{A}_{ \ixr{a}{j}{e} } - \sum_{ \substack{ 1 \leq i \leq r_2 \\ i \neq j } } \kv_{ a_i c } \, ( \ms{A}^{ v, v }_{ i, j } )_{ \ixd{a}{i,j} } \text{.}
\end{align}
Combining \eqref{eq.prelim_regular}, \eqref{eql.prelim_decomp_0}, \eqref{eql.prelim_decomp_21}, and the above, we conclude that
\begin{align}\label{eql.prelim_decomp_31} 
    I_1 + I_2 + I_5 &= \sum_{ i = 1 }^{ r_1 } [ \kv^e{}_c \, \Dvm_b ( \ms{A}^\rho_i )_{ e \ix{a} } + \Dvm_b \kv^d{}_c \, ( \ms{A}^\rho_i )_{ e \ix{a} } ] - \sum_{ j = 1 }^{ r_2 } [ \kv_{ a_j c } \, \Dvm_b ( \ms{A}^v_j )_{ \ixd{a}{j} } + \Dvm_b \kv_{ a_j c } \, ( \ms{A}^v_j )_{ \ixd{a}{j} } ] \\
    \notag &= \rho^{-1} \sum_{ i = 1 }^{ r_1 } \Dvm_b ( \ms{A}^\rho_i )_{ c \ix{a} } - \rho^{-1} \sum_{ j = 1 }^{ r_2 } \gv_{ a_j c } \, \Dvm_b ( \ms{A}^v_j )_{ \ixd{a}{j} } + \sum_{ i = 1 }^{ r_1 } \mi{O}_{ M - 2 } ( \rho; \Dvm \ms{A}^\rho_i )_{ c b \ix{a} } \\
    \notag &\qquad + \sum_{ j = 1 }^{ r_2 } \mi{O}_{ M - 2 } ( \rho; \Dvm \ms{A}^v_j )_{ c b \ix{a} } + \sum_{ i = 1 }^{ r_2 } \mi{O}_{ M - 2 } ( 1; \ms{A}^\rho_i )_{ c b \ix{a} } + \sum_{ j = 1 }^{ r_2 } \mi{O}_{ M - 2 } ( 1; \ms{A}^v_j )_{ c b \ix{a} } \text{.}
\end{align}
Similar computations using \eqref{eql.prelim_decomp_22} yield
\begin{align*}
    I_3 &= \sum_{ i = 1 }^{ r_1 } \kv^d{}_b \, \Dvm_c ( \ms{A}^\rho_i )_{ d \ix{a} } + 2 \sum_{ 1 \leq i < j \leq r_1 } \kv^d{}_b \kv^e{}_c \, ( \ms{A}^{ \rho, \rho }_{ i, j } )_{ e d \ix{a} } - \sum_{ i = 1 }^{ r_1 } \sum_{ j = 1 }^{ r_2 } \kv^d{}_b \kv_{ a_j c } \, ( \ms{A}^{ \rho, v }_{ i, j } )_{ d \ixd{a}{j} } - r_1 \kv^d{}_b \kv_{ d c } \, \ms{A}_{ \ix{a} } \\
    &= \rho^{-1} \sum_{ i = 1 }^{ r_1 } \Dvm_c ( \ms{A}^\rho_i )_{ b \ix{a} } - r_1 \rho^{-2} \gv_{ b c } \, \ms{A}_{ \ix{a} } + 2 \rho^{-2} \sum_{ 1 \leq i < j \leq r_1 } ( \ms{A}^{ \rho, \rho }_{ i, j } )_{ c b \ix{a} } - \rho^{-2} \sum_{ i = 1 }^{ r_1 } \sum_{ j = 1 }^{ r_2 } \gv_{ a_j c } \, ( \ms{A}^{ \rho, v }_{ i, j } )_{ b \ixd{a}{j} } \text{,} \\
    &\qquad + \sum_{ i = 1 }^{ r_1 } \mi{O}_{ M - 2 } ( \rho; \Dvm \ms{A}^\rho_i )_{ c b \ix{a} } + \mi{O}_{ M - 2 } ( 1; \ms{A} )_{ c b \ix{a} } + \sum_{ 1 \leq i < j \leq r_1 } \mi{O}_{ M - 2 } ( 1; \ms{A}^{ \rho, \rho }_{ i, j } )_{ c b \ix{a} } \\
    &\qquad + 2 \rho^{-2} \sum_{ i = 1 }^{ r_1 } \sum_{ j = 1 }^{ r_2 } \mi{O}_{ M - 2 } ( 1; \ms{A}^{ \rho, v }_{ i, j } )_{ c b \ixd{a}{j} } \text{,}
\end{align*}
and
\begin{align*}
    I_4 &= - \sum_{ j = 1 }^{ r_2 } \kv_{ a_j b } \, \Dvm_c ( \ms{A}^v_j )_{ \ixd{a}{j} } - \sum_{ i = 1 }^{ r_1 } \sum_{ j = 1 }^{ r_2 } \kv_{ a_j b } \kv^e{}_c \, ( \ms{A}^{ \rho, v }_{ i, j } )_{ e \ixd{a}{j} } - \sum_{ j = 1 }^{ r_2 } \kv_{ a_j b } \kv^e{}_c \, \ms{A}_{ \ixr{a}{j}{e} } \\
    &\qquad + 2 \sum_{ 1 \leq i < j \leq r_2 } \kv_{ a_i c } \kv_{ a_j b } \, ( \ms{A}^{ v, v }_{ i, j } )_{ \ixd{a}{i,j} } \\
    &= - \rho^{-1} \sum_{ j = 1 }^{ r_2 } \gv_{ a_j b } \, \Dvm_c ( \ms{A}^v_j )_{ \ixd{a}{j} } - \rho^{-2} \sum_{ i = 1 }^{ r_1 } \sum_{ j = 1 }^{ r_2 } \gv_{ a_j b } \, ( \ms{A}^{ \rho, v }_{ i, j } )_{ c \ixd{a}{j} } - \rho^{-2} \sum_{ j = 1 }^{ r_2 } \gv_{ a_j b } \, \ms{A}_{ \ixr{a}{j}{c} } \\
    &\qquad + 2 \rho^{-2} \sum_{ 1 \leq i < j \leq r_2 } \gv_{ a_i c } \gv_{ a_j b } \, ( \ms{A}^{ v, v }_{ i, j } )_{ \ixd{a}{i,j} } + \sum_{ j = 1 }^{ r_2 } \mi{O}_{ M - 2 } ( \rho; \Dv \ms{A}^v_j )_{ c b \ix{a} } \\
    &\qquad + \sum_{ i = 1 }^{ r_1 } \sum_{ j = 1 }^{ r_2 } \mi{O}_{ M - 2 } ( 1; \ms{A}^{ \rho, v }_{ i, j } )_{ c b \ix{a} } + \mi{O}_{ M - 2 } ( 1; \ms{A} )_{ c b \ix{a} } + \sum_{ 1 \leq i < j \leq r_2 } \mi{O}_{ M - 2 } ( 1; \ms{A}^{ v, v }_{ i, j } )_{ c b \ix{a} } \text{.}
\end{align*}

Finally, combining \eqref{eql.prelim_decomp_11}, \eqref{eql.prelim_decomp_31}, and the above, we obtain
\begin{align*}
    \gv^{ b c } \, \nabla_{ b c } A_{ \ix{\rho} \ix{a} } &= \rho^{-l} \gv^{ b c } \nablam_{ b c } ( \rho^l \ms{A} )_{ \ix{a} } + 2 \rho^{-1} \left[ \sum_{ i = 1 }^{ r_1 } \gv^{ b c } \Dvm_b ( \ms{A}^\rho_i )_{ c \ix{a} } - \sum_{ j = 1 }^{ r_2 } \Dvm_{ a_j } ( \ms{A}^v_j )_{ \ixd{a}{j} } - ( n r_1 + r_2 ) \rho^{-2} \, \ms{A}_{ \ix{a} } \right] \\
    &\qquad + 2 \rho^{-2} \left[ \sum_{ 1 \leq i < j \leq r_1 } \gv^{ b c } \, ( \ms{A}^{ \rho, \rho }_{ i, j } )_{ c b \ix{a} } - \sum_{ i = 1 }^{ r_1 } \sum_{ j = 1 }^{ r_2 } ( \ms{A}^{ \rho, v }_{ i, j } )_{ a_j \ixd{a}{j} } + \sum_{ 1 \leq i < j \leq r_2 } \gv_{ a_i a_j } \, ( \ms{A}^{ v, v }_{ i, j } )_{ \ixd{a}{i,j} } \right] \\
    &\qquad + \sum_{ i = 1 }^{ r_1 } \mi{O}_{ M - 2 } ( \rho; \Dvm \ms{A}^\rho_i )_{ \ix{a} } + \sum_{ j = 1 }^{ r_2 } \mi{O}_{ M - 2 } ( \rho; \Dvm \ms{A}^v_j )_{ \ix{a} } + \sum_{ i = 1 }^{ r_2 } \mi{O}_{ M - 3 } ( \rho; \ms{A}^\rho_i )_{ \ix{a} } \\
    &\qquad + \sum_{ j = 1 }^{ r_2 } \mi{O}_{ M - 3 } ( \rho; \ms{A}^v_j )_{ \ix{a} } + \mi{O}_{ M - 2 } ( 1; \ms{A} )_{ \ix{a} } + \sum_{ 1 \leq i < j \leq r_1 } \mi{O}_{ M - 2 } ( 1; \ms{A}^{ \rho, \rho }_{ i, j } )_{ \ix{a} } \\
    &\qquad + \sum_{ i = 1 }^{ r_1 } \sum_{ j = 1 }^{ r_2 } \mi{O}_{ M - 2 } ( 1; \ms{A}^{ \rho, v }_{ i, j } )_{ \ix{a} } + \sum_{ 1 \leq i < j \leq r_2 } \mi{O}_{ M - 2 } ( 1; \ms{A}^{ v, v }_{ i, j } )_{ \ix{a} } \text{.}
\end{align*}
\eqref{eq.box} now follows from \eqref{eql.prelim_decomp_10}, the above, and the fact that
\begin{equation*}
    \Box A = \rho^2 ( \nabla_{ \rho \rho } A + \gv^{ b c } \nabla_{ b c } A ) \text{,} \qquad \rho^{-1} \Boxm ( \rho^l \ms{A} ) = \rho^2 [ \rho^{-l} \nablam_{ \rho \rho } ( \rho^l \ms{A} ) + \gv^{ b c } \rho^{-l} \nablam_{ b c } ( \rho^l \ms{A} ) ] \text{.}
\end{equation*}

%%%%%%%%%%%%%%%%%%%%%%%%%%%%%%%%%%%%%%%%%%%%%

\section{$\ms{Q}^0$, $\ms{Q}^1$ and $\ms{Q}^2$} \label{sec.quadratic_terms}

In this section we derive the precise forms of the vertical tensor fields $\ms{Q}^0$, $\ms{Q}^1$, $\ms{Q}^2$ appearing in Proposition \ref{th.vert_wave_eqns}. Let $Q_{\al\be\ga\de}$ denote the right-hand side of \eqref{eq.stwave}:
\begin{align}
    Q_{\al\be\ga\de} :&= 4\W{^\la_\al^\mu_{[\de|}} \W{_\la_\be_\mu_{|\ga]}}-\W{^\la^\mu_\ga_\de}\W{_\al_\be_\la_\mu} \\
    &= g^{\la\nu}g^{\mu\sigma} \paren{4W_{\nu\al\sigma[\de|}W_{\la\be\mu|\ga]} - W_{\nu\sigma\ga\de} W_{\al\be\la\mu}} \nonumber\\
    &= \rho^4 \paren{ 4 W_{\rho \al \rho [\de|} W_{\rho \be \rho |\ga]} } \nonumber\\
    &\qquad + \rho^4\gv^{ef} \paren{ 4 W_{\rho \al f [\de|} W_{\rho \be e |\ga]} - W_{\rho f \ga \de} W_{\al \be \rho e} } \nonumber\\
    &\qquad + \rho^4 \gv^{ef} \paren{ 4 W_{f \al \rho [\de|} W_{e \be \rho |\ga]} - W_{f \rho \ga \de} W_{\al \be e \rho} } \nonumber\\
    &\qquad + \rho^4 \gv^{ef} \gv^{gh} \paren{ 4 W_{f \al h [\de|} W_{e \be g |\ga]} - W_{f h \ga \de} W_{\al \be e g} } \text{,}
\end{align}
where we have applied the Fefferman-Graham gauge condition \eqref{eq.FG_gauge}. $\ms{Q}^0$, $\ms{Q}^1$ and $\ms{Q}^2$ are given by
\begin{align}
    \ms{Q}^0_{abcd} = Q_{abcd} \text{,}\qquad \ms{Q}^1_{bcd} = Q_{\rho bcd} \text{,}\qquad \ms{Q}^2_{bd} = Q_{\rho b \rho d} \text{.}
\end{align}
In particular, one finds for $\ms{Q}^0$:
\begin{align}
    \ms{Q}^0_{abcd} &= 4 \wv^2_{a[d}\wv^2_{c]b} + 2 \gv^{ef} \brak{ 2 \wv^1_{[d|fa} \wv^1_{|c]eb} + 2 \wv^1_{af[d|} \wv^1_{be|c]} - \wv^1_{fcd} \wv^1_{eab} } \\
    &\qquad + \gv^{ef} \gv^{gh} \bigg[ 2 \paren{ \hat{\wv}^0_{fahd} - \frac{2}{n-2} \paren{ \gv_{f[h} \wv^2_{d]a} + \gv_{a[d} \wv^2_{h]f} } } \paren{ \hat{\wv}^0_{ebgc} + \frac{2}{n-2} \paren{ \gv_{e[g} \wv^2_{c]b} + \gv_{b[c} \wv^2_{g]e} } } \nonumber\\
    &\qquad\qquad - 2 \paren{ \hat{\wv}^0_{fahc} - \frac{2}{n-2} \paren{ \gv_{f[h} \wv^2_{c]a} + \gv_{a[c} \wv^2_{h]f} } } \paren{ \hat{\wv}^0_{ebgd} + \frac{2}{n-2} \paren{ \gv_{e[g} \wv^2_{d]b} + \gv_{b[d} \wv^2_{g]e} } } \nonumber\\
    &\qquad\qquad - \paren{ \hat{\wv}^0_{fhcd} - \frac{2}{n-2} \paren{ \gv_{f[c} \wv^2_{d]h} + \gv_{h[d} \wv^2_{c]f} } } \paren{ \hat{\wv}^0_{abeg} + \frac{2}{n-2} \paren{ \gv_{a[e} \wv^2_{g]b} + \gv_{b[g} \wv^2_{e]a} } } \bigg] \nonumber\\
    &= \mc{O}_{M_0-2} \paren{ 1; \hat{\wv}^0 }_{abcd} + \mc{O}_{M_0-3} \paren{ \rho; \wv^1 }_{abcd} + \mc{O}_{M_0-2} \paren{ 1; \wv^2 }_{abcd} \text{,} \nonumber
\end{align}
for $\ms{Q}^1$:
\begin{align}
    \ms{Q}^1_{bcd} &= -2 \gv^{ef} \brak{ \wv^1_{fcd} \wv^2_{be} + 2 \wv^2_{f[d} \wv^1_{c]eb} } \\
    &\qquad - \gv^{ef} \gv^{gh} \bigg[ 2 \wv^1_{fhd} \paren{ \hat{\wv}^0_{ebgc} - \frac{2}{n-2} \paren{ \gv_{e[g} \wv^2_{c]b} + \gv_{b[c} \wv^2_{g]e} } } \nonumber\\
    &\qquad\qquad - 2 \wv^1_{fhc} \paren{ \hat{\wv}^0_{ebgd} -\frac{2}{n-2} \paren{ \gv_{e[g} \wv^2_{d]b} + \gv_{b[d} \wv^2_{g]e} } } \nonumber\\
    &\qquad\qquad + \paren{ \hat{\wv}^0_{fhcd} -\frac{2}{n-2} \paren{ \gv_{f[c} \wv^2_{d]h} + \gv_{h[d} \wv^2_{c]f} } } \wv^1_{beg} \bigg] \nonumber\\
    &= \mc{O}_{M_0-3} \paren{ \rho; \hat{\wv}^0 }_{bcd} + \mc{O}_{M_0-2} \paren{ 1; \wv^1 }_{bcd} + \mc{O}_{M_0-3} \paren{ \rho; \wv^2 }_{bcd} \text{,} \nonumber
\end{align}
and, finally, for $\ms{Q}^2$:
\begin{align}
    \ms{Q}^2_{bd} &= -2 \gv^{ef} \wv^2_{fd} \wv^2_{be} + \gv^{ef} \gv^{gh} \bigg[ 2 \wv^1_{fhd} \wv^1_{geb} - \wv^1_{dfh} \wv^1_{beg} \\
    &\qquad\qquad - \wv^2_{fh} \paren{ \hat{\wv}^0_{ebgd} - \frac{2}{n-2} \paren{ \gv_{e[g} \wv^2_{d]b} + \gv_{b[d} \wv^2_{g]e} } } \bigg] \nonumber \\
    &= \mc{O}_{M_0-2} \paren{ 1; \hat{\wv}^0 }_{bd} + \mc{O}_{M_0-3} \paren{ \rho; \wv^1 }_{bd} + \mc{O}_{M_0-2} \paren{ 1; \wv^2 }_{bd} \text{,} \nonumber
\end{align}
in which we have applied \eqref{eq.w0_exp}, \eqref{eq.w1_exp} and \eqref{eq.w2_exp}.

\medskip

\printbibliography

@article{Holzegel_2016,
   title={Unique Continuation from Infinity in Asymptotically Anti-de Sitter Spacetimes},
   volume={347},
   ISSN={1432-0916},
   url={http://dx.doi.org/10.1007/s00220-016-2576-0},
   DOI={10.1007/s00220-016-2576-0},
   number={3},
   journal={Communications in Mathematical Physics},
   publisher={Springer Science and Business Media LLC},
   author={Holzegel, Gustav and Shao, Arick},
   year={2016},
   month={2},
   pages={723–775}
}

@article{Holzegel_2017,
   title={Unique continuation from infinity in asymptotically anti-de Sitter spacetimes II: Non-static boundaries},
   volume={42},
   ISSN={1532-4133},
   url={http://dx.doi.org/10.1080/03605302.2017.1390677},
   DOI={10.1080/03605302.2017.1390677},
   number={12},
   journal={Communications in Partial Differential Equations},
   publisher={Informa UK Limited},
   author={Holzegel, Gustav and Shao, Arick},
   year={2017},
   month={11},
   pages={1871–1922}
}

@misc{McGill_2020,
    title={Null Geodesics and Improved Unique Continuation for Waves in Asymptotically Anti-de Sitter Spacetimes},
    author={Alex McGill and Arick Shao},
    year={2020},
    eprint={2008.07416},
    archivePrefix={arXiv},
    primaryClass={gr-qc}
}

@article{Fefferman_1985,
   title={Conformal invariants},
   volume={},
   ISSN={},
   url={http://www.numdam.org/item/AST_1985__S131__95_0},
   DOI={},
   number={S131},
   journal={\'Elie Cartan et les math\'ematiques d'aujourd'hui - Lyon, 25-29 juin 1984},
   publisher={Soci\'et\'e math\'ematique de France},
   author={Fefferman, Charles and Graham, C. Robin},
   year={1985},
   month={},
   pages={95-116}
}

@article{Alinhac_1995,
    title={A non uniqueness result for operators of principal type},
    volume={220},
    ISSN={1},
    url={https://doi.org/10.1007/BF02572631},
    DOI={10.1007/BF02572631},
    number={},
    journal={Mathematische Zeitschrift},
    publisher={},
    author={Alinhac, S. and Baouendi, M. S.},
    year={1995},
    month={12},
    pages={561-568}
}

@misc{Shao_2020,
    title={The Near-Boundary Geometry of Einstein-Vacuum Asymptotically Anti-de Sitter Spacetimes},
    author={Arick Shao},
    year={2020},
    eprint={2008.07396},
    archivePrefix={arXiv},
    primaryClass={gr-qc}
}

@article{Maldacena_1999, volume={38},
   ISSN={0020-7748},
   url={http://dx.doi.org/10.1023/A:1026654312961},
   DOI={10.1023/a:1026654312961},
   number={4},
   journal={International Journal of Theoretical Physics},
   publisher={Springer Science and Business Media LLC},
   author={Maldacena, Juan},
   year={1999},
   pages={1113–1133}
}

@book{ONeill_1983,
    author = "O'Neill, Barrett",
    title = "{Semi-Riemannian geometry with applications to relativity}",
    isbn = "9780125267403",
    publisher = "Academic Press",
    year = "1983"
}

@book{Wald_1984,
    author = "Wald, Robert",
    title = "{General relativity}",
    doi = "10.1017/CBO9780511635397",
    isbn = "978-0226870335",
    publisher = "The University of Chicago Press Press",
    year = "1984"
}

@book{Carroll_2004,
    author = "Carroll, Sean M.",
    title = "{Spacetime and Geometry}",
    isbn = "978-0-8053-8732-2",
    publisher = "Cambridge University Press",
    month = "7",
    year = "2004"
}

@techreport{Graham_1999_2,
      author        = "Graham, C R",
      title         = "{Volume and Area Renormalizations for Conformally Compact
                       Einstein Metrics}",
      number        = "math.DG/9909042",
      month         = "9",
      year          = "1999",
      reportNumber  = "math.DG/9909042",
      url           = "https://cds.cern.ch/record/399185",
}

@article{Graham_1991,
title = "Einstein metrics with prescribed conformal infinity on the ball",
journal = "Advances in Mathematics",
volume = "87",
number = "2",
pages = "186 - 225",
year = "1991",
issn = "0001-8708",
doi = "https://doi.org/10.1016/0001-8708(91)90071-E",
url = "http://www.sciencedirect.com/science/article/pii/000187089190071E",
author = "C.Robin Graham and John M Lee"
}

@article{Graham_1999_1,
author        = "Graham, C R and Witten, Edward",
title         = "{Conformal Anomaly Of Submanifold Observables In AdS/CFT Correspondence}",
journal       = "Nucl. Phys. B",
number        = "hep-th/9901021",
volume        = "546",
pages         = "52-64",
year          = "1999",
reportNumber  = "hep-th/9901021",
url           = "https://cds.cern.ch/record/375864",
}

@article{Biquard_2008,
author        = "Biquard, Olivier",
title         = "{Continuation unique a partir de l'infini conforme pour les metriques d'Einstein}",
journal       = "Mathematical Research Letters",
number        = "6",
volume        = "15",
pages         = "1091 – 1099",
publishers    = "International Press",
year          = "2008",
}

@article{Chrusciel_2011,
title = "Unique continuation and extensions of Killing vectors at boundaries for stationary vacuum space-times",
journal = "Journal of Geometry and Physics",
volume = "61",
number = "8",
pages = "1249 - 1257",
year = "2011",
issn = "0393-0440",
doi = "https://doi.org/10.1016/j.geomphys.2011.02.011",
url = "http://www.sciencedirect.com/science/article/pii/S0393044011000386",
author = "Piotr T. Chruściel and Erwann Delay",
}

@article{DeHaro_2001,
title = "Holographic Reconstruction of Spacetime and Renormalization in the AdS/CFT Correspondence",
journal = "Communications in Mathematical Physics",
volume = "217",
number = "3",
pages = "595–622",
year = "2001",
doi = "10.1007/s002200100381",
author = "De Haro, Sebastian and Skenderis, Kostas and Solodukhin, Sergey N.",
}

@article{Chatzikaleas_2021,
   title={A gauge-invariant unique continuation criterion for waves in asymptotically anti-de Sitter spacetimes},
   volume={},
   ISSN={},
   url={},
   DOI={},
   number={},
   journal={},
   publisher={},
   author={Chatzikaleas, Athanasios and Shao, Arick},
   year={in preparation},
   month={},
   pages={}
}

@article{Banados_1992,
  title = {Black hole in three-dimensional spacetime},
  author = {Ba\~{n}ados, M\'{a}ximo and Teitelboim, Claudio and Zanelli, Jorge},
  journal = {Phys. Rev. Lett.},
  volume = {69},
  issue = {13},
  pages = {1849--1851},
  numpages = {0},
  year = {1992},
  month = {9},
  publisher = {American Physical Society},
  doi = {10.1103/PhysRevLett.69.1849},
  url = {https://link.aps.org/doi/10.1103/PhysRevLett.69.1849}
}

@article{Banados_1993,
   title={Geometry of the 2+1 black hole},
   volume={48},
   ISSN={0556-2821},
   url={http://dx.doi.org/10.1103/PhysRevD.48.1506},
   DOI={10.1103/physrevd.48.1506},
   number={4},
   journal={Physical Review D},
   publisher={American Physical Society (APS)},
   author={Bañados, Máximo and Henneaux, Marc and Teitelboim, Claudio and Zanelli, Jorge},
   year={1993},
   month={8},
   pages={1506–1525}
}

@article{Minneborg_1996,
   title={Making anti-de Sitter black holes},
   volume={13},
   ISSN={1361-6382},
   url={http://dx.doi.org/10.1088/0264-9381/13/10/010},
   DOI={10.1088/0264-9381/13/10/010},
   number={10},
   journal={Classical and Quantum Gravity},
   publisher={IOP Publishing},
   author={Åminneborg, Stefan and Bengtsson, Ingemar and Holst, Sören and Peldán, Peter},
   year={1996},
   month={10},
   pages={2707–2714}
}

@article{Banados_1997,
    author = "Ba\~{n}ados, Maximo",
    title = "{Constant curvature black holes}",
    eprint = "gr-qc/9703040",
    archivePrefix = "arXiv",
    doi = "10.1103/PhysRevD.57.1068",
    journal = "Phys. Rev. D",
    volume = "57",
    pages = "1068--1072",
    year = "1998"
}

@article{Mann_1993,
   title={Gravitationally collapsing dust in 2 + 1 dimensions},
   volume={47},
   ISSN={0556-2821},
   url={http://dx.doi.org/10.1103/PhysRevD.47.3319},
   DOI={10.1103/physrevd.47.3319},
   number={8},
   journal={Physical Review D},
   publisher={American Physical Society (APS)},
   author={Mann, R. B. and Ross, S. F.},
   year={1993},
   month={4},
   pages={3319–3322}
}

@article{Holzegel_2021,
   title={Unique Continuation for the Einstein Equations in Asymptotically Anti-de Sitter Spacetimes},
   volume={},
   ISSN={},
   url={},
   DOI={},
   number={},
   journal={},
   publisher={},
   author={Holzegel, Gustav and Shao, Arick},
   year={in preparation},
   month={},
   pages={}
}

@article{Imbimbo_2000,
   title={Diffeomorphisms and holographic anomalies},
   volume={17},
   ISSN={1361-6382},
   url={http://dx.doi.org/10.1088/0264-9381/17/5/322},
   DOI={10.1088/0264-9381/17/5/322},
   number={5},
   journal={Classical and Quantum Gravity},
   publisher={IOP Publishing},
   author={Imbimbo, C and Schwimmer, A and Theisen, S and Yankielowicz, S},
   year={2000},
   month={2},
   pages={1129–1138}
}

@article{Skenderis_2001,
   title={Asymptotically Anti-de Sitter spacetimes and their stress energy tensor},
   volume={16},
   ISSN={1793-656X},
   url={http://dx.doi.org/10.1142/S0217751X0100386X},
   DOI={10.1142/s0217751x0100386x},
   number={05},
   journal={International Journal of Modern Physics A},
   publisher={World Scientific Pub Co Pte Lt},
   author={Skenderis, Kostas},
   year={2001},
   month={2},
   pages={740–749}
}

@article{Skenderis_2000,
   title={Quantum effective action from the AdS/CFT correspondence},
   volume={472},
   ISSN={0370-2693},
   url={http://dx.doi.org/10.1016/S0370-2693(99)01467-7},
   DOI={10.1016/s0370-2693(99)01467-7},
   number={3-4},
   journal={Physics Letters B},
   publisher={Elsevier BV},
   author={Skenderis, Kostas and Solodukhin, Sergey N.},
   year={2000},
   month={1},
   pages={316–322}
}

@article{Guisset_2021,
   title={Construction of counter-examples to the wave equation unique continuation problem with a critically singular potential},
   volume={},
   ISSN={},
   url={},
   DOI={},
   number={},
   journal={},
   publisher={},
   author={Guisset, Simon},
   year={in preparation},
   month={},
   pages={}
}

\end{document}